\documentclass{article} 

\usepackage[accepted]{icml2024}

\usepackage{amsmath,amsfonts,bm}

\def\eqref#1{equation~\ref{#1}}

\def\1{\bm{1}}

\DeclareMathAlphabet{\mathsfit}{\encodingdefault}{\sfdefault}{m}{sl}
\SetMathAlphabet{\mathsfit}{bold}{\encodingdefault}{\sfdefault}{bx}{n}

\newcommand{\Var}{\mathrm{Var}}

\newcommand{\Cov}{\mathrm{Cov}}

\DeclareMathOperator*{\argmax}{arg\,max}

\usepackage{hyperref}
\usepackage{url}
\usepackage{enumitem}

\usepackage[utf8]{inputenc} 
\usepackage[T1]{fontenc}    
\usepackage{booktabs}       
\usepackage{amsfonts}       
\usepackage{nicefrac}       
\usepackage{microtype}      
\usepackage{xcolor}         

\usepackage{mathtools}
\usepackage{multirow}
\usepackage{graphicx}
\usepackage{subfigure}
\usepackage{booktabs} 
\usepackage[font=small,labelfont=bf]{caption}

\usepackage{caption}
\usepackage{graphicx}
\usepackage{float} 

\usepackage{parskip}

\usepackage{amsthm}
\usepackage{amsmath}
\usepackage{amsfonts}
\usepackage{nicefrac}
\usepackage{microtype}
\usepackage{comment}
\usepackage{subfigure}
\usepackage{color}
\usepackage{cases}

\usepackage{tikz}
\usetikzlibrary{matrix}
\usetikzlibrary{calc}
\usetikzlibrary{decorations.pathreplacing}

\usepackage{amscd,latexsym,amsthm,amsfonts,amssymb,amsmath,amsxtra}

\usepackage{tikz}
\usepackage{pgf, tikz}
\usetikzlibrary{arrows, automata}
\RequirePackage{amssymb}
\definecolor{mygray}{gray}{0.85}

\usepackage{thmtools}
\usepackage{thm-restate}

\theoremstyle{plain}
\newtheorem{theorem}{Theorem}[section]

\newtheorem{proposition}[theorem]{Proposition}

\theoremstyle{definition}
\newtheorem{definition}[theorem]{Definition}

\newtheorem{observation}{Observation}

\newcommand{\abs}[1]{\ensuremath{\left|#1\right|}}
\newcommand{\norm}[2][]{\ensuremath{\left\Vert #2 \right\Vert}}

\renewcommand{\vec}[1]{\mathbf{#1}}

    \newcommand{\BC}{{\mathbb {C}}} 
    \newcommand{\BE}{{\mathbb {E}}}

     \newcommand{\BR}{{\mathbb {R}}}

     \newcommand{\CL}{{\mathcal {L}}}
    \newcommand{\CM}{{\mathcal {M}}} \newcommand{\CN}{{\mathcal {N}}}
    \newcommand{\CO}{{\mathcal {O}}}

    \renewcommand{\Im}{{\mathrm{Im}}}

    \newcommand{\Ker}{{\mathrm{Ker}}}

    \renewcommand{\Re}{{\mathrm{Re}}}

    \newcommand{\Spec}{{\mathrm{Spec}}}

\newcommand\coolover[2]{\mathrlap{\smash{\overbrace{\phantom{%
    \begin{matrix} #2 \end{matrix}}}^{\mbox{$#1$}}}}#2}

    \theoremstyle{plain}
    \newtheorem{thm}{Theorem}[section] \newtheorem{cor}[thm]{Corollary}
    \newtheorem{lem}[thm]{Lemma}  \newtheorem{prop}[thm]{Proposition}
     \newtheorem{defn}[thm]{Definition}
     \newtheorem {rem}[thm]{Remark}

\usepackage{microtype}
\usepackage{graphicx}
\usepackage{subfigure}
\usepackage{booktabs} 

\usepackage{hyperref}
\definecolor{mydarkblue}{rgb}{0,0.08,0.45}
\hypersetup{
  colorlinks=true,
  linkcolor=mydarkblue ,    
  filecolor=mydarkblue ,  
  urlcolor=red ,     
  citecolor=mydarkblue       
}

\usepackage[capitalize,noabbrev]{cleveref}

\usepackage[textsize=tiny]{todonotes}

\icmltitlerunning{Prediction Accuracy of Learning in Games :
Follow-the-Regularized-Leader meets Heisenberg}

\begin{document}
\twocolumn[
\icmltitle{Prediction Accuracy of Learning in Games : 
\\Follow-the-Regularized-Leader meets Heisenberg}

\begin{icmlauthorlist}
\icmlauthor{Yi Feng}{yyy}
\icmlauthor{Georgios Piliouras}{comp}
\icmlauthor{Xiao Wang}{yyy,yyy2}
\end{icmlauthorlist}

\icmlaffiliation{yyy}{Shanghai University of Finance and Economics, Shanghai, China}
\icmlaffiliation{comp}{Google DeepMind, London, United Kingdom}
\icmlaffiliation{yyy2}{Key Laboratory of Interdisciplinary Research of Computation and Economics, China}

\icmlcorrespondingauthor{Xiao Wang}{wangxiao@sufe.edu.cn}

\icmlkeywords{Machine Learning, ICML}

\vskip 0.3in
]



\printAffiliationsAndNotice{} 

\begin{abstract}
We investigate the accuracy of prediction in deterministic learning dynamics of zero-sum games with random initializations, specifically focusing on observer uncertainty and its relationship to the evolution of covariances. Zero-sum games are a prominent field of interest in machine learning due to their various applications. Concurrently, the accuracy of prediction in dynamical systems from mechanics has long been a classic subject of investigation since the discovery of the Heisenberg Uncertainty Principle. This principle employs covariance and standard deviation of particle states to measure prediction accuracy. In this study, we bring these two approaches together to analyze the Follow-the-Regularized-Leader (FTRL) algorithm in two-player zero-sum games. We provide growth rates of covariance information for continuous-time FTRL, as well as its two canonical discretization methods (Euler and Symplectic). A Heisenberg-type inequality is established for FTRL. Our analysis and experiments also show that employing Symplectic discretization enhances the accuracy of prediction in learning dynamics.
\end{abstract}

\section{Introduction}

In recent years understanding the behavior of learning algorithms in repeated games has attracted increasing interests from the machine learning community~\citep{lanctot2017unified,yang2020overview}. Follow-the-Regularized-Leader (FTRL) algorithm~\citep{abernethy2009competing,shalev2012online}, arguably the most well known class of no-regret dynamics, plays a prominent role in analysis of behavior of learning algorithms. The dynamics of such online learning algorithm in zero-sum games has been a particularly intense object of study as zero-sum games related to numerous recent applications and advances in AI such as, achieving super-human performance in Go~\citep{silver2016mastering}, Poker~\citep{brown2018superhuman} and Generative Adversarial Networks (GANs)~\citep{goodfellow2014generative} to name a few.



Predicting players' long term behaviors in a repeated game is a fundamental and challenging problem \citep{nachbar1997prediction,cesa2006prediction}. The conventional wisdom in this regard is that players' strategies will eventually converge to some equilibrium. However, recent studies have shown that such a belief usually fails, in particular, FTRL dynamics do not converge in zero-sum games and exhibit complex behaviors such as recurrence and divergence \citep{MPP18,bailey2018multiplicative}. Therefore, one needs to predict the future behaviors of players by tracing their day-to-day (a.k.a. last-iterate) behaviors \cite{daskalakis2017training,gidel2018variational,gidel2019negative}. By modeling the learning dynamics as a deterministic dynamical system, an observer with the computational ability to trace the learning dynamics can accurately predict future player states by knowing their \textit{exact} current states.

However, in practice, the observer may have some uncertainty about the current states of players. This kind of uncertainty is common both from a game-theoretic perspective (i.e., the unknown external action on agents’ preference) as well as from a Machine Learning perspective (i.e., system initialization by sampling from a distribution and noise introduced during training ). Thus, the observer may hope that \textit{slightly inaccurate} knowledge of current conditions can lead to \textit{slightly inaccurate} prediction. For example, by tracking the expectation of the evolving distribution along the learning dynamics that model his uncertainty, satisfactory predictions can be obtained. 

\begin{figure*}[t]
\centering
\subfigure[t = 0, variance = 0.00020781]{
\includegraphics[clip,width=0.65\columnwidth]{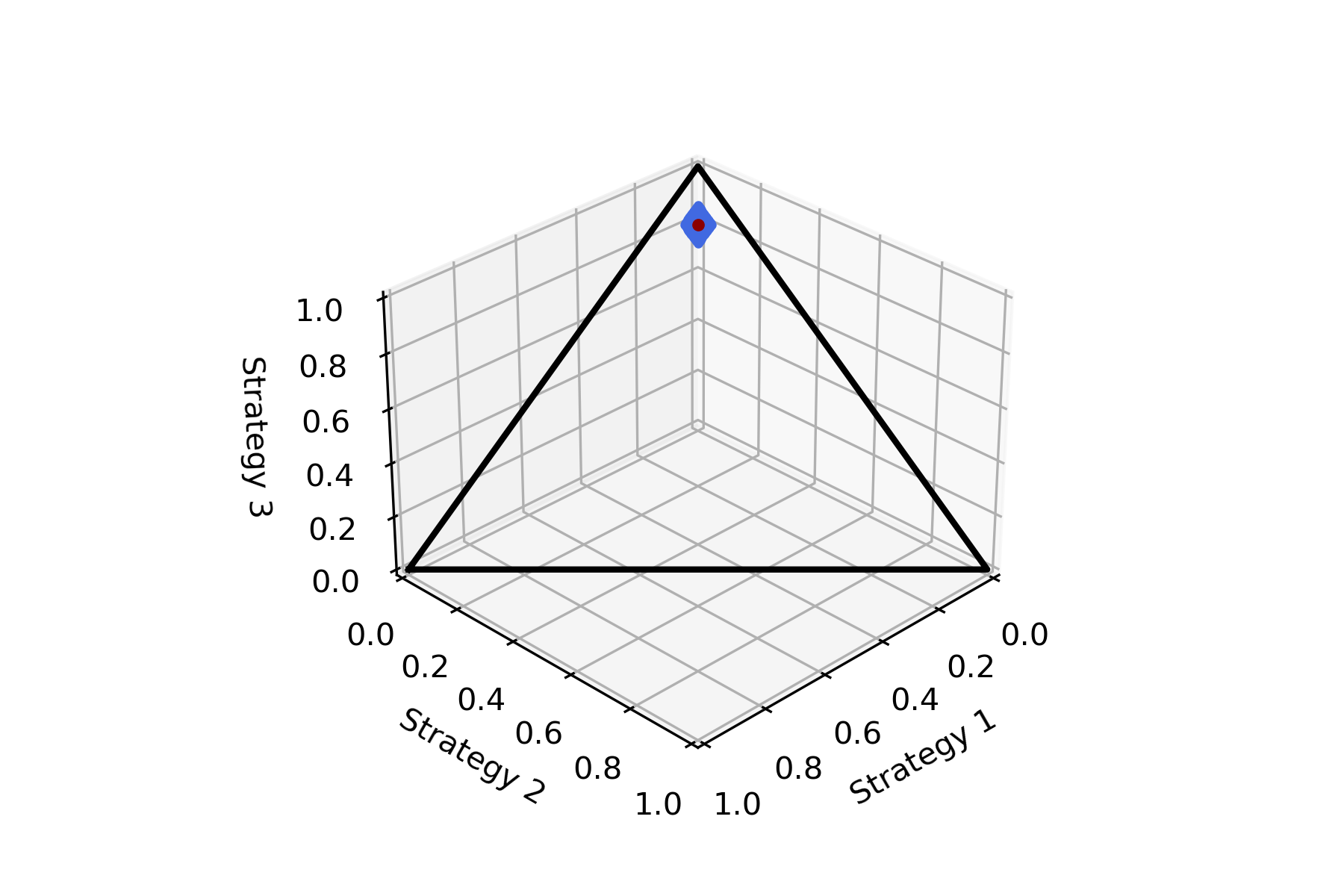}
}
\subfigure[t = 50, variance = 0.01038411]{
\includegraphics[clip,width=0.65\columnwidth]{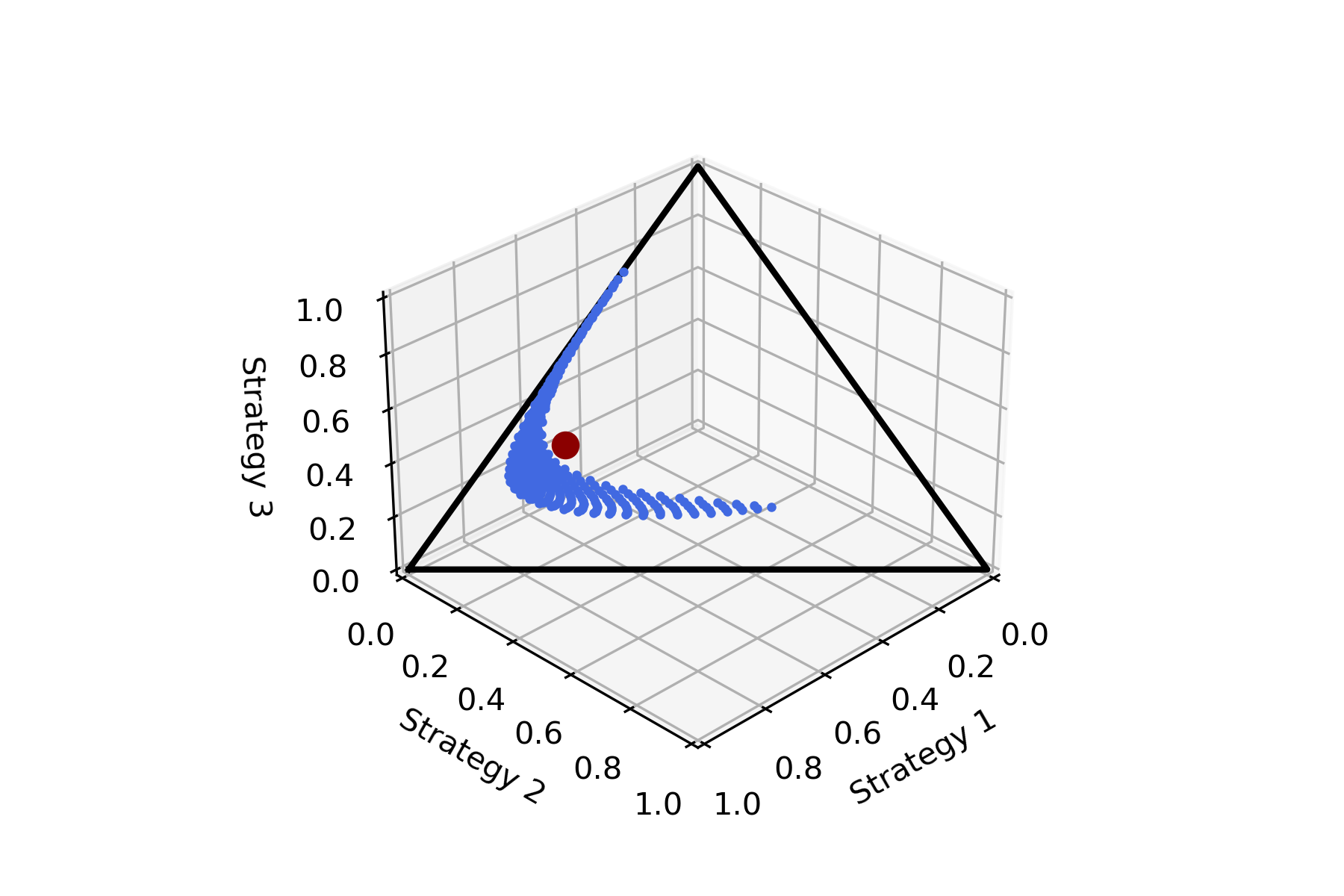}
}
\subfigure[t = 80, variance = 0.05732407]{
\includegraphics[clip,width=0.65\columnwidth]{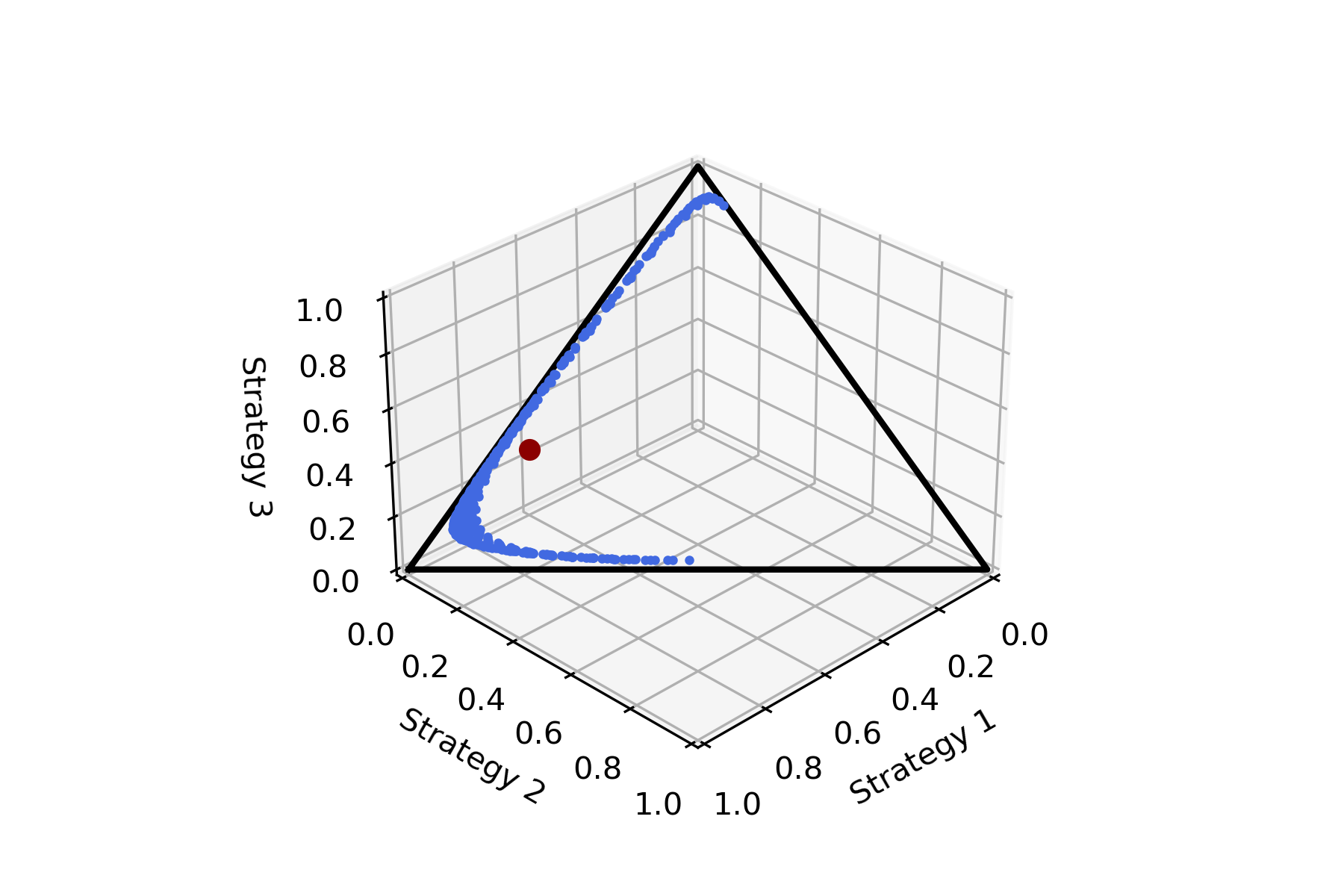}
}
\subfigure[t = 130, variance = 0.09969448]{
\includegraphics[clip,width=0.65\columnwidth]{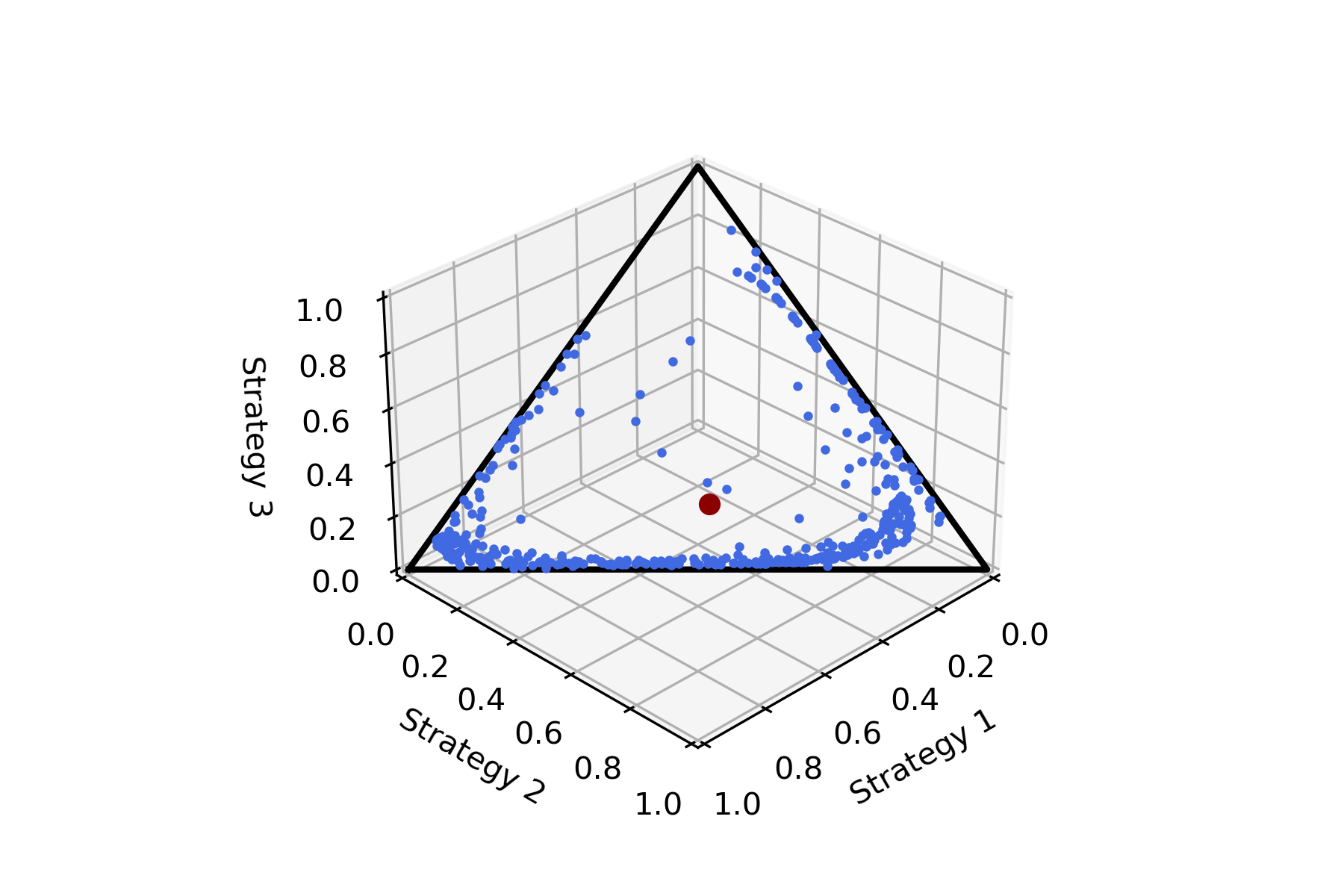}
}
\subfigure[t = 150, variance = 0.07895284]{
\includegraphics[clip,width=0.65\columnwidth]{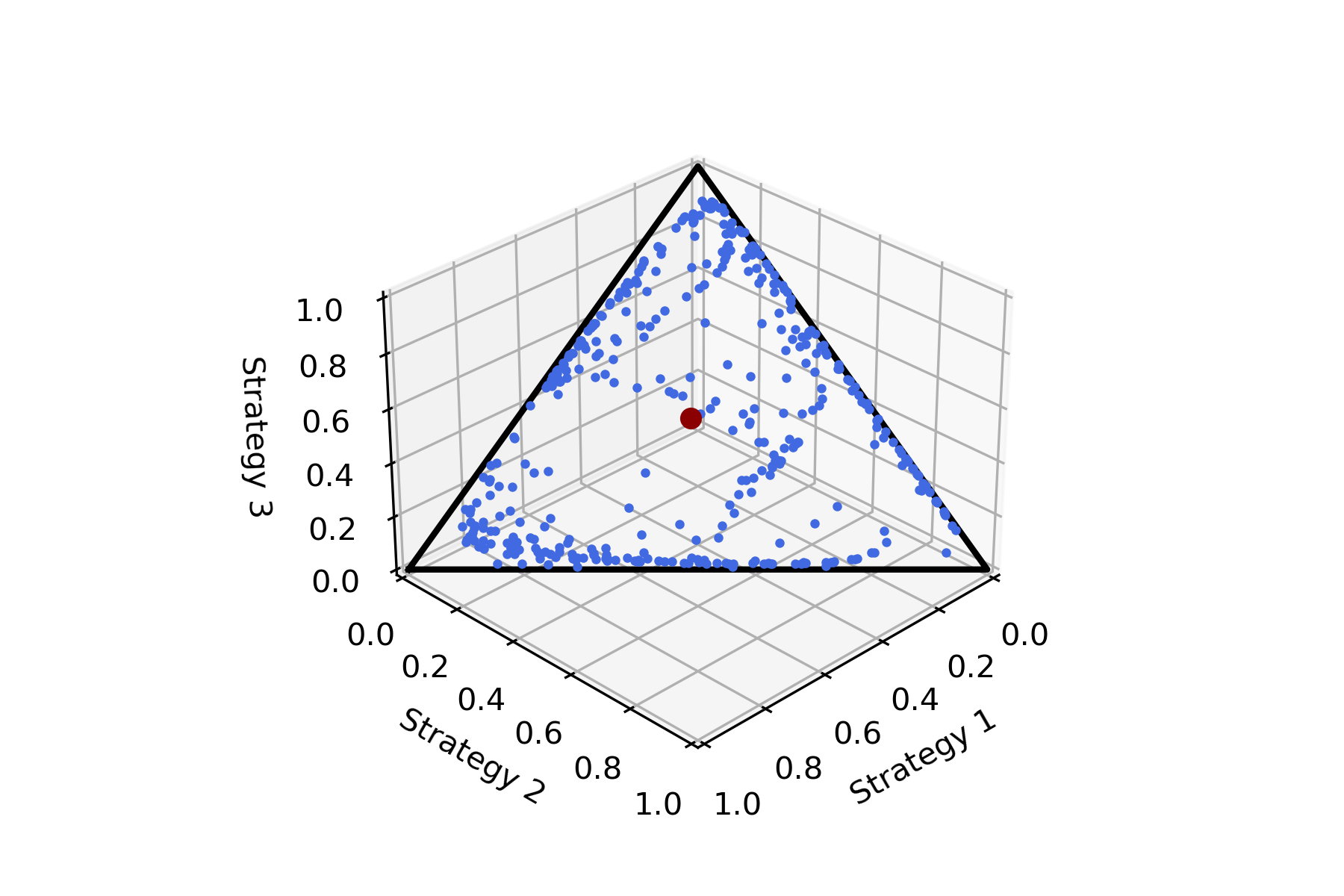}
}
\subfigure[t = 200, variance = 0.08126675]{
\includegraphics[clip,width=0.65\columnwidth]{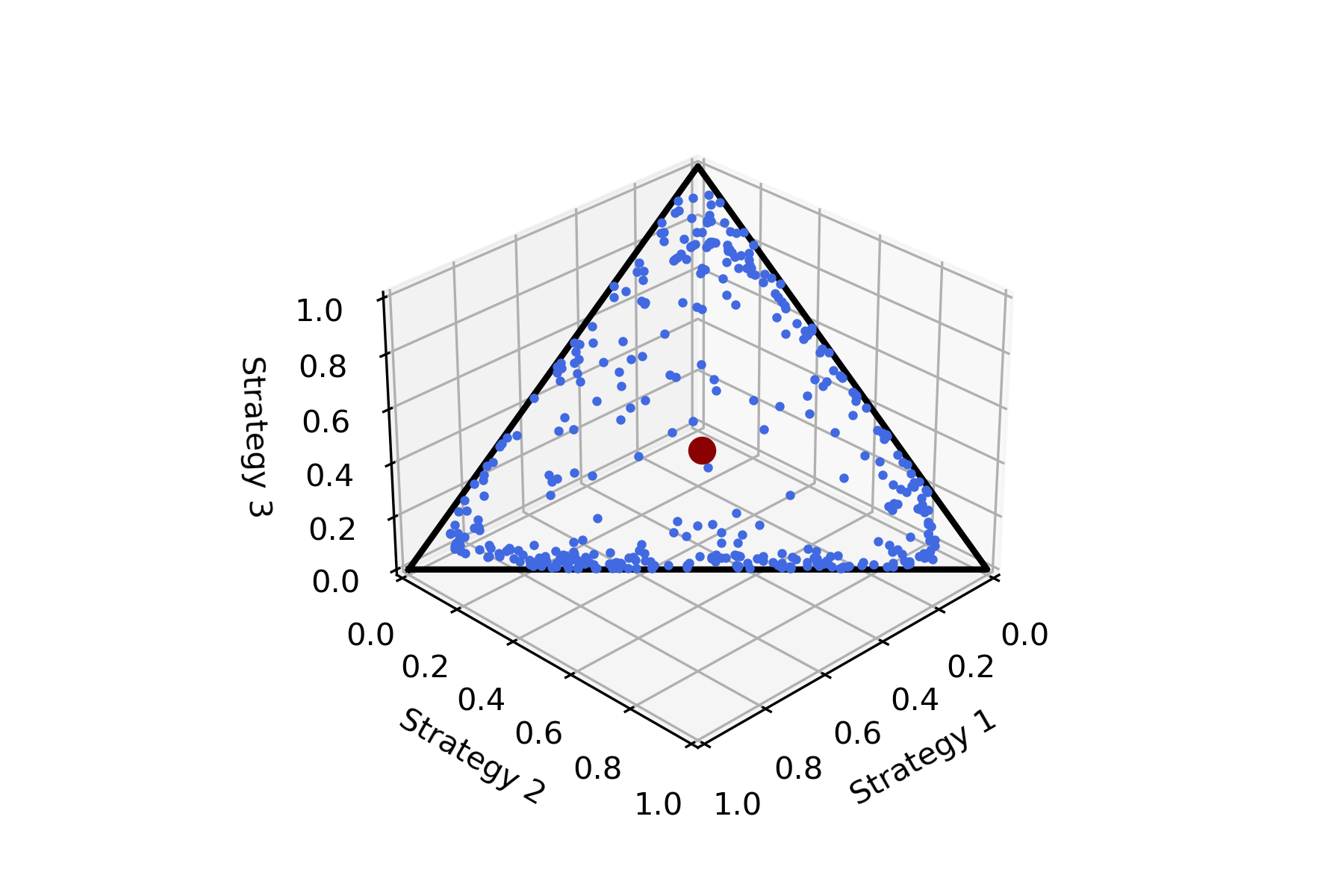}
}
\caption{Evolution of 400 initial conditions (blue points) for one player using (\ref{2pMWUA}) in R-P-S game, the red point is the expectation of blue points, the variance is calculated on the first pure strategy. At time $t=0$, these point are sampled in a small square. As time evolves, it appears that these points are randomly located on the simplex and the variance is magnified by a factor of 400. Moreover, the expectation (red point in the figures) cannot accurately predict future outcomes, as sample points in subsequent times may deviate significantly from the expected value. \textbf{A demo animation can be found \href{https://www.dropbox.com/scl/fi/u65qqyidnhzpu61mmyb7r/AltMWU1_newgif.gif?rlkey=jcuib751hn3ph69kf67hhobe4&dl=0}{here}.}}
\label{Tra 3d}
\end{figure*}
Unfortunately, this hope didn't materialize as expected. One remarkable aspect of several learning dynamics, like Multiplicative Weights Updates (FTRL with negative entropy regularizer), is that even slight initial deviation can give rise to significantly divergent strategy trajectories for players over extended periods \citep{sato2002chaos,vilone2011chaos,galla2013complex}. Figure \ref{Tra 3d} illustrates that in a repeated Rock-Paper-Scissor game, the evolution of many orbits starting from a small region when players use Alternating Multiplicative Weights Update. Figure \ref{Tra 3d} shows that even small initial uncertainty can be amplified and make the accurate prediction difficult. Moreover, tracking the expectation fails to accurately prediction as a large portion of  points will deviate significantly from the expectation; in terms of statistics, the \textbf{(co)variance} of the distribution can be large over time, hindering accurate predictions of players' future behaviors. Motivated by this example, we naturally formulate the following question:

\textit{How to track the accuracy of prediction in learning dynamics?}

The most relevant paper in this direction is \citep{CPT2022}, which utilized \textit{differential entropy} as their metric of uncertainty and demonstrated that it grows linearly fast in two-player zero-sum games, quantifying the amount of excess information an observer must gather to keep track of the uncertain system evolution. However, as we will show in following, differential entropy cannot capture the uncertainty evolution of the \emph{alternating} update rule of game dynamics, such as the alternating MWU in Figure \ref{Tra 3d}. 

In this work, we will study the evolution of \textit{covariance (standard deviation)} associated with related random variables that govern the dynamics of FTRL, utilizing the framework of the Hamiltonian formulation of FTRL \citep{PB1,wibisono2022alternating}. This perspective studies the evolution of covariance closely related to the well known \emph{Heisenberg Uncertainty Principle} in quantum mechanics, which states that the covariance of the momentum and position of a microscopic particle cannot be small at the same time, i.e., $\Delta p\Delta q\ge \text{some positive constant}$. The Hamiltonian formulation of FTRL endows the \textit{cumulative strategy} and \textit{cumulative payoff} of each agent the roles of position $q$ and momentum $p$ of each particle, and these quantities completely determine the dynamics of the game. Once the initialization is randomized, the deterministic learning dynamics still makes the cumulative strategy and payoff random variables whose mean and covariance matrix is related to that of the initialization. As what is implied from Heisenberg Uncertainty Principle \citep{busch2007heisenberg}, information on covariance (standard deviation) measures the \emph{accuracy of prediction} in learning dynamics. We will also demonstrate that this analogy is not vague; similar phenomena to the Heisenberg Uncertainty Principle already exist in FTRL dynamics.



\paragraph{Our contributions.}We highlight our main contributions as follows:
\begin{itemize}[leftmargin=*]
\item We prove that differential entropy remains constant when two players alternatively update their strategies. Therefore, it is not a strong concept for capturing the evolution of uncertainty in alternating update, see Proposition \ref{entropya}.

\item We propose the covariance matrix as an uncertainty measurement which captures both simultaneous and alternating updates, with rate of increasing calculated concretely in Euclidean regularized FTRL, see Theorem \ref{Covariance Evolution in FTRL with Euclidean Regularizer}. Our results imply a separation between simultaneous and continuous time/alternating plays. As an immediate application, Corollary \ref{cor:Cheb} provides a prediction with quantitative description of risk, i.e., probability of deviating from expectation up to time $t$;

\item For FTRL with general regularizers, a Heisenberg type inequality on variances of cumulative strategy and payoff is obtained, i.e., $\Delta X_{i,\alpha}\Delta y_{i,\alpha}\ge\text{positive constant}$. Similar to the Heisenberg Uncertainty Principle in quantum mechanics, this inequality indicates a  trade-off between prediction accuracy in strategy spaces versus payoff spaces for game dynamics, see Theorem \ref{prop:nonlinear}. In Figure \ref{UI} we present an example to illustrate this point.
\end{itemize}
\begin{figure}[h]
    \centering
    \includegraphics[width=0.45\textwidth]{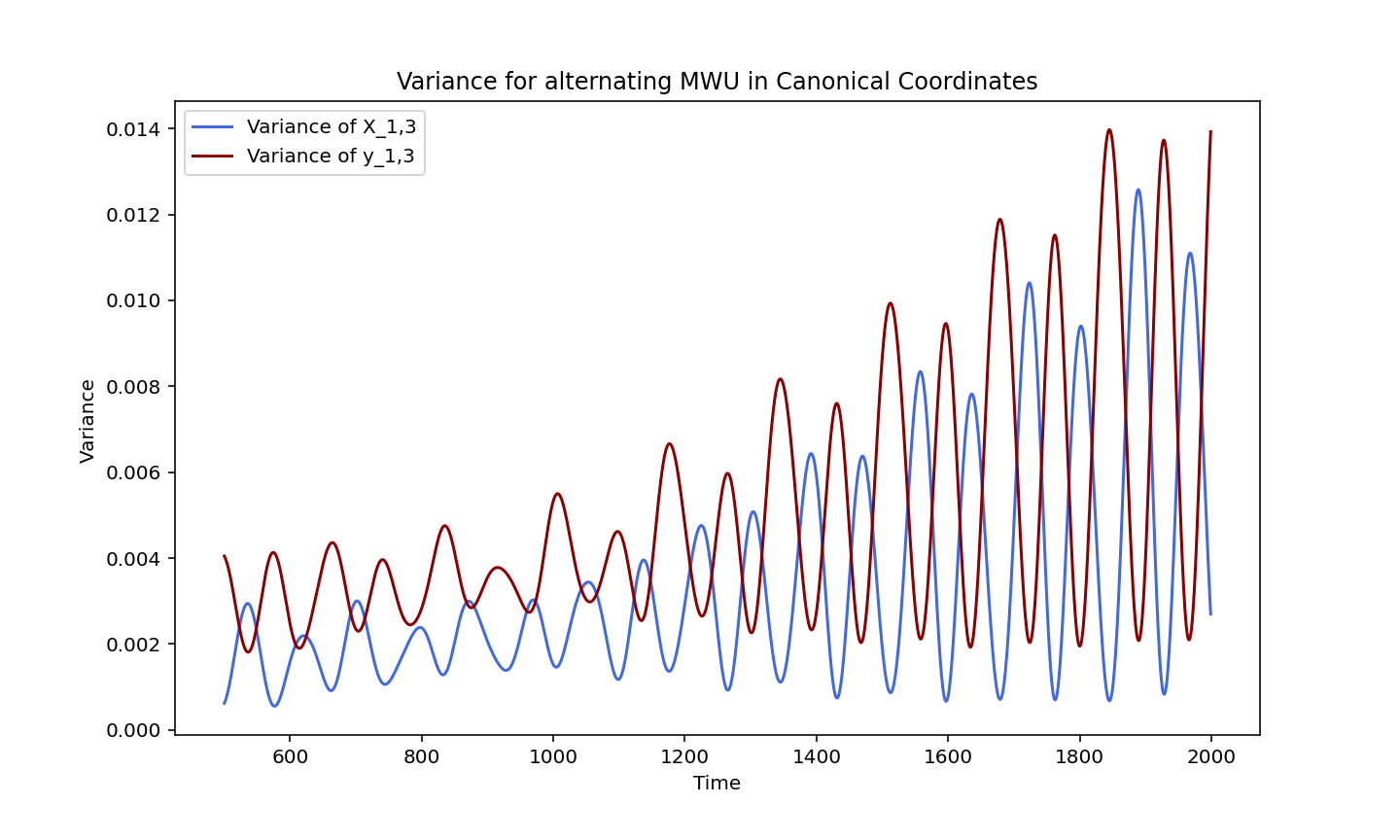}
    \caption{The two curves represent the evolution of  $\Delta X_{i,\alpha}$ and $\Delta y_{i,\alpha}$
    on 100 random samples of two players when they use (\ref{2pMWUA}).
    When one curve is decreasing, another curve is increasing.
    This implies a tradeoff between accuracy in strategy spaces versus payoff spaces.}
    \label{UI}
\end{figure}

\paragraph{Technical innovations.}
The techniques used in drawing the above conclusions come from different areas. The theoretical framework for analyzing the uncertainty evolution is the classic mechanical formulation of games \citep{PB1,wibisono2022alternating}. A rigorous correspondence has been established between symplectic discretization \citep{haier2006geometric} and alternating plays to study their properties. To demonstrate that differential entropy is constant in alternating plays, we utilize the volume preservation property of Symplectic discretization. Furthermore, the intuition in  deriving the covariance evolution of Symplectic discretization is from \citep{wang1994some}, and the proof combine tools from application of matrix analysis in dynamical systems \citep{colonius2014dynamical}. The uncertainty inequality for general FTRL is a consequence of a classic result from symplectic geometry known as non-squeezing theorem \citep{mcduff2017introduction} and variance analysis methods from multivariate statistics.

\section{Preliminaries}\label{Preliminaries}
\paragraph{Learning in games.}
A  two agent zero-sum game consists of two agent $\CN = \{1,2\}$, where agent $i$ selects a strategy from the strategy space (or primal space) $\mathcal{X}_i \subset \BR^{n_i}$ and $n_i$ represents the number of actions available to agent $i$.  Typically, $\mathcal{X}_i$ is chosen to be $\BR^{n_i}$, which we called the unconstrained zero-sum game, or it is chosen to be the simplex constrains
\begin{align*}
    \Delta_i = \{ x \  \lvert \ \sum^{n_i}_{s=1} x_{i,s} = 1, x_{i,s} \ge 0 \}.
\end{align*}
Utilities of both agents are determined via payoff matrix $A^{(ij)} \in \BR^{n_i \times n_j}$, and in a zero-sum game, the pay off matrix satisfy $A^{(ij)} = - A^{(ji)}$. 
For convenience, we will also use $A$ to refer to $A^{(12)}$, and thus $A^{(21)} = -A^{\top}$. Given that agent $i$ selects strategy $x_i \in \mathcal{X}_i \subset \BR^{n_i}$, agent 1 receives utility $u_1(x_1,x_2) = \langle x_1, Ax_2 \rangle$, and agent 2 receives utility $u_2(x_2,x_1) = - \langle x_2, A^{\top}x_1 \rangle$.  Naturally agents want to maximize their utility resulting the following max-min problem: 
\begin{align*}
    \max_{x_1 \in \mathcal{X}_1} \min_{x_2 \in \mathcal{X}_2} x^{\top}_1 Ax_2 .\tag{Zero-Sum Game}
\end{align*}
\paragraph{Follow-the-Regularized-Leader.} Follow-the-Regularized-Leader (FTRL) is a widely used class of no-regret online learning algorithms. In continuous time FTRL, at time $t$, agent $i$ updates strategies  $x_i(t)$ based on the cumulative  payoff vector $y_i(t)$, 
\begin{align*}\label{2pFTRL}
\tag{Continuous FTRL}
&y_i(t)=y_i(0)+\int_0^t A^{(ij)}x_j(s)ds
\\
&x_i(t)=\argmax_{x_i\in\mathcal{X}_i}\{\langle x_i,y_i(t)\rangle-h_i(x_i)\} 
\end{align*}
where $h_i$ is a strongly convex function, which is called the regularizer. It is also well known that 
\begin{align}\label{conv_conj}
    x_i(t) =  \nabla h^*_i\left(y_i(t) \right),
\end{align} where 
\begin{align}
h^*_i(y_i) = \max_{x_i \in \mathcal{X}_i} \{ \langle x_i,y_i  \rangle - h_i(x_i) \} 
\end{align}
is the convex conjugate of $h_i$ \citep{shalev2006convex}. Therefore, if an observer knows the regularizer used by players, they can convert the information of $y_i(t)$ into the primal space $x_i(t)$.

 Gradient descent ascent (GDA) and multiplicative weights updates (MWU) are two  of the most well known special cases of FTRL algorithms. For unconstrained GDA, the regularizers are chosen to be the Euclidean norm, i.e.,  $h_i(x_i)  = \lVert x_i \lVert^2$ and $\mathcal{X}_i = \BR^{n_i}$. For MWU, the regularizers are chosen to be negative entropy, i.e., 
 \begin{align*}
      h_i(x_i)  =\sum_{j \in [n_i]} x_{i,j} \ln(x_{i,j})
 \end{align*}
  and $\mathcal{X}_i$ be the simplex constrains. 

In discrete time, FTRL has two kinds of implementations : simultaneous and alternating. Let
$y^{t}_1 = A x^{t}_2, \ \ y^{t}_2 = A^{\top} x^{t}_1.$
In the case of GDA and MWU, the update rules with step size $\eta$ are :
\begin{align*}
&x^t_1  =x^{t-1}_1 + \eta y^{t-1}_1, \ \ \label{Sim GDAof1}
\tag{GDA}
x^t_2  =x^{t-1}_2  - \eta y^{t-1}_2.
\\
&x^t_2  = x^{t-1}_2 - \eta y^{t-1}_2,  \ \ \label{Alt GDAof1}
\tag{AltGDA}
x^t_1  = x^{t-1}_1 + \eta y^t_1.
\end{align*}
and
\begin{align*}
& x^t_1 =\left (\frac{ x^{t-1}_{1,s} e^{\eta  y^{t-1}_{1,s}}   }{ \sum^{n_1}_{j=1} x^{t-1}_{1,j} e^{\eta y^{t-1}_{1,j} }} \right )_s
\label{2pMWUS}
\tag{MWU}
 x^t_2 = \left (\frac{ x^{t-1}_{2,s} e^{- \eta y^{t-1}_{2,s} }   }{ \sum^{n_2}_{j=1} x^{t-1}_{2,j} e^{- \eta y^{t-1}_{2,j}}} \right )_s
\\
&x^t_2 = \left (\frac{ x^{t-1}_{2,s} e^{- \eta y^{t-1}_{2,s}}   }{ \sum^{n_2}_{j=1} x^{t-1}_{2,j} e^{- \eta y^{t-1}_{2,j}}} \right )_s
\label{2pMWUA}
\tag{AltMWU}
x^t_1 =\left (\frac{ x^{t-1}_{1,s} e^{\eta y^t_{1,s}}}{ \sum^{n_1}_{j=1} x^{t-1}_{1,j} e^{\eta y^t_{1,j}}} \right )_s
\end{align*}
for $s = \{1,2,..., n_i\}$ and $i = 1$ or $2$. Note that in simultaneous updates, both two players use the payoff feedback from round $t-1$ to update their strategies in round $t$, while in alternating updates, Player 2 first update his strategy $x^t_2$, and Player 1 subsequently updates her strategy $x^t_1$ based on payoff feedback of $x^t_2$. Comparing to its simultaneous partner, a series of recent works show that alternating update rule has a slower regret growth rate \citep{bailey2020finite,wibisono2022alternating,cevher2023alternation}.

\paragraph{Dynamical system.} A system of ordinary differential equations $\dot{x}=f(x)$ where $f:\mathbb{R}^n\rightarrow\mathbb{R}^n$ is a differentiable dynamical system. $f(x)$ is called the vector field of the dynamical system. If $f$ is Lipschitz continuous, there exists a continuous map $$\varphi(t,x_0):\mathbb{R}\times\mathbb{R}^n\rightarrow\mathbb{R}^n$$ such that for all $x_0\in\mathbb{R}^n$, $\varphi(t,x_0)$ is the unique solution of the initial condition problem $\{\dot{x} =f(x), x(0)=x_0\}$. The solution $\varphi(t,x_0)$ is called a \emph{trajectory} or \emph{orbit} of the dynamical system.

\textbf{Hamiltonian systems.}  A Hamiltonian system is a class of differential equations describing the evolution of momentums and positions of particles by a scalar function $H(X,Y)$  called Hamiltonian function. The state of the system, the momentum $Y=(y_1,...,y_n)^{\top}$ and position $X=(x_1,...,x_n)^{\top}$ evolves according to the following Hamilton's equations:
\begin{equation}\label{HamiltonianSystem}
\frac{dx_i}{dt}=\frac{\partial H}{\partial y_i},\ \ \ \ \frac{dy_i}{dt}=-\frac{\partial H}{\partial x_i},\ \ \ \ \text{for}\ \ \ i\in[n]. 
\end{equation}
The solution $\varphi(t,\cdot)$ of a Hamiltonian system is called a \emph{symplectic} map which is a special case of volume-preserving maps, thus the absolute value of determinant of the Jacobian matrix equals to $1$. The variables $(X,Y)$ are also referred to the \textit{canonical coordinates} of the system \citep{arnol2013mathematical}.


\subsection{Measure of Observer Uncertainty}

\textbf{Differential Entropy.} The concept of differential entropy was introduced by Shannon \citep{shannon1948mathematical} as a measure of the uncertainty associated with a continuous probability distribution. For a random vector $X \in \BR^n$ with probability density function $g(x)$ supported on $\mathcal{X} \subset \BR^n$, the differential entropy of $X$ is defined as
\begin{align*}\label{differential entropy}
\tag{Differential Entropy}
S(X) = -\int_{\mathcal{X}} g(x)\log g(x)dx.
\end{align*}

\textbf{Covariances of random vectors.} Given a random vector $X = (x_1, ..., x_m)^{\top}$ such that every $x_i$ is a random variable with finite variance and expected value, the covariance matrix $P(X) \in \BR^{m \times m}$ of $X$ is a  symmetric and positive semi-definite square matrix whose $(i,j)$ entry is the covariance, i.e.,
 \[
 \Cov(x_i,x_j) = \BE[(x_i - \BE(x_i))(x_j - \BE(x_j))].
 \]
  Note that the diagonal elements of $P(x)$ are variances of $\{x_1, ..., x_m\}$. 

In general, the differential entropy  of a random variable provides a lower bound on the determinant of its covariance matrix. Precisely, if a random vector $X \in \BR^m$ has zero mean and covariance matrix $P(X)$, then
\begin{align*}
    S(X) \le \frac{1}{2}\log \left(   (2 \pi e)^n \det P(X)  \right).
\end{align*}

\section{Setup}\label{setup}
In this section we leverage the power of the Hamiltonian formulation of game dynamics  \citep{PB1,wibisono2022alternating}. To study the strategy-payoff evolution from the perspective of each agent, it is convenient to apply the Hamiltonian formulation of the continuous time FTRL. We will establish the equivalence between discretization of the Hamiltonian system induced by continuous time FTRL and direct discretization of FTRL, where the latter leads to the GDA or MWU. 


\textbf{Euler discretization.}  Given an ordinary differential equation $\dot{x}=f(x)$ with initial condition $x(t_0)$ at time $t_0$, the Euler discretization begin the process by setting $x_0 = x(t_0)$, next choose a step size $\eta$ and set $t_{n} = t_0 + n\eta$, then the Euler discretization $\phi_{\eta}(\cdot)$  is defined by $x_{n+1} = \phi_{\eta}(x_{n}) = x_{n} + \eta f(x_n)$. The value $x_{n+1}$ is an approximation of the solution of $\dot{x}=f(x)$ at time $t_{n+1}$.

\textbf{Symplectic discretization.}  Given a Hamiltonian system as in (\ref{HamiltonianSystem}), 
a numerical method $\phi_{\eta} (\cdot)$ is called a Symplectic discretization if when applied to a Hamiltonian system, the discrete flow  $\phi_{\eta} : x \to \phi_{\eta}(x)$ is a Symplectic map for sufficient small step sizes. In this paper we focus on the following Symplectic discretizations: for $X = (x_1,...,x_n)^{\top}, Y=(y_1, ...,y_n)^{\top}$,
\begin{align*}\label{type1}
&Y^{t+1} = Y^t - \eta \nabla_{X}H(X^t,Y^{t+1}),\\ 
&X^{t+1} = X^t+  \eta \nabla_{Y}H(X^t,Y^{t+1}),
\tag{Type \uppercase\expandafter{\romannumeral1} method}
\end{align*}
or
\begin{align*}\label{type2}
&X^{t+1} = X^t + \eta \nabla_{Y}H(X^{t+1},Y^t),\\ 
&Y^{t+1} = Y^t -  \eta \nabla_{X}H(X^{t+1},Y^t).
\tag{Type \uppercase\expandafter{\romannumeral2} method}
\end{align*}
Both methods are Symplectic methods, i.e., they make the map $(X^t,Y^t) \to (X^{t+1},Y^{t+1})$ to be symplectic. More details of Symplectic method can be found in \citep{haier2006geometric}. Note that although both methods are generally implicit,  they become explicit when the Hamiltonian function is separable, i.e., can be expressed as $H(X,Y) = f(X) + g(Y)$ with functions f and g. This property holds for the Hamiltonian formulation of FTRL dynamics.

\paragraph{Canonical coordinates of FTRL dynamics.} In this paper we will focus on the dynamics of cumulative strategy and cumulative payoff. The cumulative strategy $X_i(t)$ of agent $i \in [2]$ is defined as follows : 
\begin{equation}\label{eqX}
X_i(t) = \int^t_0 x_i(s)ds .
\end{equation}
\citep{PB1} demonstrates that (\ref{2pFTRL}) can be formulated as a Hamiltonian system through $(X_i(t),y_i(t))$ using the Hamiltonian function
\begin{align}
    H(X_i,y_i) = h^*_i(y_i(t)) + h^*_j(y_j(0) + A^{(ji)}X_i(t))
\end{align}
for $j \ne i$ be the Hamiltonian function \footnote{Intuitive explanation of this Hamiltonian function can be found in Appendix \ref{Hamiltonian formulation of FTRL}.}. Thus $X_i(t)$ and $y_i(t)$ evolve according to the following Hamiltonian system
\begin{equation}\label{HamiltonianFTRL}
\frac{dX_i}{dt}=\frac{\partial H(X_i,y_i)}{\partial y_i},\ \ \ \frac{dy_i}{dt}=-\frac{\partial H(X_i,y_i)}{\partial X_i}.
\end{equation}
Following the tradition of Hamiltonian mechanics, we will refer to 
$(X_i(t),y_i(t))$ for $i = 1$ or $2$ as the \textit{canonical coordinates} of FTRL dynamics. This can be analogously understood as the position-momentum coordinates used to describe the dynamic of a particle. 

It may initially seem strange to trace the dynamics of players in a game based on their cumulative strategies/payoffs rather than their actual strategies $x_i(t)$. However, there is no loss in tracing these variables since $y_i(t)$ can be translated into $x_i(t)$ through the map $\nabla h^*(\cdot)$ introduced in (\ref{conv_conj}), and $X_i(t)$ can be easily translated into $y_j(t)$ through $y_j(t) = y_j(t_0) + A^{(ji)} \left(X_i(t)-X_i(t_0)\right)$.

\paragraph{Primal-dual correspondence via discretization.}Since we use Euler and Symplectic discretization on the Hamiltonian system, which is not obviously equivalent to the conventionally natural update rules in the strategy spaces $\mathcal{X}_i$, we next establish formally that the Euler or Symplectic discretization of  continuous time FTRL with Euclidean norm / negative entropy regularizer implies GDA/MWU or AltGDA/AltMWU respectively. This correspondence can be stated in the following proposition.
\begin{proposition}\label{prop:primal-duel}
For each agent $i\in[2]$, let $\mathcal{X}_i$ denote the strategy spaces. Then following statements holds:
\begin{itemize}[leftmargin=*]
\item If both players use Euclidean regularizers and $\mathcal{X}_i = \BR^{n_i}$, then the Euler discretization of (\ref{HamiltonianFTRL}) is equivalent to Gradient Descent-Ascent (GDA) on the strategy spaces; and the Symplectic discretization of (\ref{HamiltonianFTRL}) is equivalent to Alternating Gradient Descent-Ascent (AltGDA) on the strategy spaces.
\item If both players use entropy regularizers and $\mathcal{X}_i = \Delta_i$ be the simplex, then the Euler discretization of (\ref{HamiltonianFTRL}) is equivalent to Multiplicative Weights Update (MWU) on the strategy spaces; and the Symplectic discretization of (\ref{HamiltonianFTRL}) is equivalent to Alternating Multiplicative Weights Update (AltMWU) on the strategy spaces.
\end{itemize}
\end{proposition}
Here equivalent means the variables $(X^t_i, y^t_i)$ getting from the discretizations is the same as the cumulative strategy and payoff of the game dynamics. 

It is known that Euler discretization significantly alters the properties of continuous system \citep{holmes2007introduction}. This results the differences between continuous FTRL and simultaneous plays. For example, continuous FTRL exhibits cycle behaviors \citep{MPP18} while simultaneous plays typically diverge \citep{bailey2018multiplicative}. Compared to Euler discretization , Symplectic discretization can preserve the structure of the continuous system \citep{haier2006geometric}. Therefore, Proposition \ref{prop:primal-duel} suggests that 

\emph{The dynamical behaviors of alternating update rules should be consistent with those of continuous learning dynamics.} 

An example of this phenomenon
is (AltGDA) keeps the cycle behaviors of its continuous partner \citep{bailey2020finite}.

To prove Proposition \ref{prop:primal-duel}, we introduce a novel method of discretizing the continuous Hamiltonian system (\ref{HamiltonianFTRL})
by a combination of two types Symplectic methods, while still keep the symplectic structures on the dynamics of each agents. We believe this method is of independent interests. The detailed proofs of Proposition \ref{prop:primal-duel} are deferred to Appendix \ref{Proof of Proposition prop:primal-duel}.

\par \paragraph{Random initialization.} We consider the case when noise is introduced to the canonical coordinates $(X_i(t_0),y_i(t_0))$ at time moment $t_0>0$ in continuous FTRL or $(X_i^{t_0},y_i^{t_0})$ in discrete time settings. The main objective of this paper is to study the evolution of observer uncertainty given the covariance matrix $P_0$ of the initialization $(X_i(t_0),y_i(t_0))$ or $(X_i^{t_0},y_i^{t_0})$. Take discrete time FTRL for example, the covariance matrix $P_0$ consists of variances $\Var(X_{i,\alpha}^{t_0})$, $\Var(y_{i,\alpha}^{t_0})$, and covariances $\Cov(X_{i,\alpha}^{t_0},y_{i,\beta}^{t_0})$ for all $\alpha,\beta\in[n_i]$ . Tracing the evolution of $\Var(X_{i,\alpha}^t)$ and $\Var(y_{i,\alpha}^t)$ in iterations, we are able to quantify how accurate the prediction will be in FTRL dynamics. 


\section{Deficiency of Differential Entropy}\label{DE in FTRL}

Differential entropy, as a measure of observer uncertainty, was used in studying the predictability of MWU in zero-sum games \citep{CPT2022}. In this section we investigate the evolution of differential entropy in FTRL with differential discretization methods. In Proposition \ref{entropya} we show that differential entropy is insufficient in capturing the uncertainty evolution for alternating plays.

\begin{prop}\label{entropya}
When two players use (\ref{2pMWUA}) with arbitrary step size, the differential entropy of their cumulative strategy and payoff keeps constant, i.e., if $t>t_0$,
\begin{align}
S(X_i^t,y_i^t) = S(X_i^{t_0},y_i^{t_0}).
\end{align}
\end{prop}

\begin{prop}\label{entropys}
When two players use (MWU) with step sizes  $\eta < \min \{1, 1/ \lVert A \lVert_2^2 \}$, the differential entropy  of their cumulative strategy and payoff  has linear growth rate, i.e., 
\begin{align}
S(X_i^t,y_i^t) \ge S(X_i^{t_0},y_i^{t_0}) + c t
\end{align}
where $c > 0$ is a constant determined by payoff matrix $A$. 
\end{prop}

Different from the proof of \citep{CPT2022} for the differential entropy evolution of (\ref{2pMWUS}), which relies a detailed calculation of the Jacobin map of the dynamics, our proof of Proposition \ref{entropya} is established based on the relationship between Symplectic discretization and (\ref{2pMWUA}), as state in Proposition \ref{prop:primal-duel}. In fact, the evolution of differential entropy is determined by the determinant of the Jacobin matrix of  the update rule from $(X_i^t,y_i^t)$ to $(X_i^{t+1},y_i^{t+1})$, and in each update, the differential entropy is invariant if and only if the absolute value of this determinant equals to $1$. As shown in Proposition \ref{prop:primal-duel}, the update rule of $(X_i^t,y_i^t)$  in (AltMWU) constitutes a symplectic map, and the absolute value of the determinant of  Jacobin matrix for every symplectic map must equal to $1$, therefore we can conclude that the differential entropy in (AltMWU) keeps a constant. The detailed proofs of Propositions \ref{entropys} and \ref{entropya} are deferred to Appendix  \ref{Proof of Section DE in FTRL}, where we also provide numerical examples for these two propositions.

\section{Covariance in FTRL}\label{last}
In this section, we are presenting formally the evolution of covariances of cumulative strategies and payoffs. We start with continuous time FTRL with Euclidean regularizers, and proceed in considering Euler and Symplectic discretization of continuous time FTRL. In the end, for general FTRL, covariance evolution follows an inequality derived by using techniques of symplectic geometry.
\subsection{Covariance evolution in Euclidean regularizer.}The evolution of covariance matrix with continuous time FTRL can be deduced from the Hamiltonian formulation of learning dynamics. The Euler discretization of each agent's continuous time FTRL exponentially amplifies the covariance in the learning process. In contrast, symplectic discretization, which has been proven equivalent to alternating update in strategies, amplify covariance of cumulative strategies polynomially and keep that of cumulative payoffs bounded. In this section we will focus on the view point of
agent 1, and the same results also hold for agent 2 as they are symmetry.

\begin{thm}\label{Covariance Evolution in FTRL with Euclidean Regularizer}
In two-player zero-sum games, suppose both players use FTRL with Euclidean norm regularizers and unconstrained strategy sets $\mathcal{X}_i = \BR^{n_i}$. Suppose at time $t_0>0$ the random cumulative strategy and payoff form a random vector $\left(X_i^{t_0},y_i^{t_0}\right)$ with covariance matrix $P(t_0)\ne 0$. Then for all $t>t_0$ and $\alpha,\beta\in\{1,...,n_1 \}$, the covariance of $\left(X_i^t,y_i^t\right)$ evolves in continuous and discrete FTRL according to the following : 
\begin{enumerate}
\item In Euler discretization, for all $\alpha,\beta\in [n_1]$, it holds that $\Cov(X_{1,\alpha}^t,X_{1,\beta}^t)$, and $\Cov(y_{1,\alpha}^t,y_{1,\beta}^t)$  are of $\Theta \left( (1+\gamma \eta^2)^{2t}\right)$, where $\eta$ is the step size and $\gamma$ is the maximal eigenvalue of $AA^{\top}$.
\item In continuous time and symplectic discretization, for all $\alpha,\beta\in [n_1]$, it holds that
\begin{itemize}[leftmargin=*]
\item if $AA^{\top}$ is non-singular, then $\Cov(X_{1,\alpha}^t,X_{1,\beta}^t)$ and
$\Cov(y_{1,\alpha}^t,y_{1,\beta}^t)$ are of $\mathcal{O}(1)$.
\item if $AA^{\top}$ is singular, $\Cov(X_{1,\alpha}^t,X_{1,\beta}^t)$ is of $\Theta(t^2)$, $\Cov(y_{1,\alpha}^t,y_{1,\beta}^t)$ is of $\mathcal{O}(1)$. 
\end{itemize}
\end{enumerate}
\end{thm}

Theorem \ref{Covariance Evolution in FTRL with Euclidean Regularizer} distinguishes clearly between the covariance evolution on Euler discretization and Symplectic discretization: Euler discretization exhibits an exponential growth rate, while Symplectic discretization maintains a quadratic growth rate that matches the growth rate of the primary continuous dynamics.  Therefore, an observer can achieve higher prediction accuracy when players use Symplectic discretization. The proof of Theorem \ref{Covariance Evolution in FTRL with Euclidean Regularizer} are deferred to Appendix \ref{proofer}, and experimental results are presented in Section \ref{expp}.


The dependence on singularity of  $AA^{\top}$ can be intuitively explained in terms of the covariance evolution of the average strategies (i.e., $X^t/t$). In an unconstrained zero-sum game, the average strategies of (\ref{Alt GDAof1}) will converge to some equilibrium \citep{gidel2018variational}. When $AA^{\top}$ is non-singular, the only equilibrium is the all-zero vector. The average strategy converges to this point for all initial points, resulting in a small covariance of average strategy and implying slow cumulative strategy growth. However, when $AA^{\top}$ is singular, there exist multiple distinct equilibria and the average strategy will converge towards these different equilibrium. Consequently, throughout this process, the covariance of average strategy remains substantial which leads to a large covariance of cumulative strategy.

Theorem \ref{Covariance Evolution in FTRL with Euclidean Regularizer} also implies the following corollary regarding to the covariance evolution of primal space. The proof is directly as in this case $x^t_i = y^t_i$ holds. See Appendix \ref{Proof of Euclidean norm regularizers}.

\begin{cor}\label{E_a}
    Under the same conditions as in Theorem \ref{Covariance Evolution in FTRL with Euclidean Regularizer}, the covariance of the actual strategy $x^t_1$ evolves as following : for all $\alpha,\beta\in [n_1]$
\begin{enumerate}
\item In Euler discretization, it holds that $\Cov(x_{1,\alpha}^t,x_{1,\beta}^t)$ is of  $\Theta \left( (1+\gamma \eta^2)^{2t}\right)$.
\item In continuous time and symplectic discretization, it holds that
$\Cov(x_{1,\alpha}^t,x_{1,\beta}^t)$ is of $\mathcal{O}(1)$. 
\end{enumerate}
\end{cor}
An immediate application of Theorem 5.1 is to utilize Chebyshev's inequality for predicting the quantitative measure of the risk associated with random variables that deviate significantly from the expected value.
\begin{cor}\label{cor:Cheb}
Let $k$ be any constant that is greater than 1,  $X_{i,\alpha}(t)$ and $y_{i,\alpha}(t)$ are cumulative strategy and payoffs of the $i$th agent with respect to strategy $\alpha$ at time $t$. Then we have
\begin{align*}
\text{Pr}\{\abs{X_{i,\alpha}(t)-\mu_X(t)}>k\sqrt{c}t\}&<\frac{1}{k^2}
\\
\text{Pr}\{\abs{y_{i,\alpha}(t)-\mu_y(t)}>k\}&<\frac{c}{k^2}
\end{align*}
where $c$ is a constant, and $\mu_{X}(t)$ and $\mu_y(t)$ are expectations of cumulative strategy and payoffs up to time $t$.
\end{cor}


\subsection{Covariance evolution with general regularizer.} So far we have left the evolution of covariance in continuous time FTRL with general regularizers unaddressed. The challenge comes from the non-linearity of Hamiltonian system induced by continuous time FTRL algorithm. Suppose the integral flow of Hamiltonian system is $\phi_t(X_i(t_0),y_i(t_0))$. In the most general setting, we are able to provide a lower bound for the product of standard deviation of $X_{i,\alpha}(t)$ and $y_{i,\alpha}(t)$, i.e., $\Delta X_{i,\alpha}(t)\Delta y_{i,\alpha}(t)\ge\text{constant}$. We state the conditions and results formally in the following Theorem.

\begin{thm}\label{prop:nonlinear}
Let vector $X_{i}(t)$ and $y_{i}(t)$ be cumulative strategy and payoff, for $i\in[2]$ and $\alpha\in[n_i]$. Assume that the higher order differentials of $\phi_t(\cdot)$ are bounded by some constant $K$, and the standard deviations $\Delta X_{i,\alpha}(t_0) $ and $\Delta y_{i,\alpha}(t_0) $  at initial time $t_0$ are sufficient small,\footnote{"Sufficietly small" follows the convention in statistical modeling, e.g. p166 in \citep{benaroya2005probability}, refer to Appendix \ref{apd3} for more details.} 
then for $t> t_0$ it holds that 
\begin{align*}
\Delta X_{i,\alpha}(t)\Delta y_{i,\alpha}(t) \ge\frac{1}{\sqrt{2}}\frac{w_L (P(t_0))}{\pi },
\end{align*}
where $w_L(P(t_0))$ is the linear Gromov width of the ellipsoid defined by the initial covariance matrix $P(t_0)$.
\end{thm}

The definition of  Gromov width can be found in Appendix \ref{SG}.  Note that the inequality holds trivially when there is no uncertainty, i.e., $P(t_0) = 0$, since in this case  $w_L^2(P(t_0)) = 0$. The necessary background and details of proof are left in Appendix \ref{Details of Section "last"}. 

\begin{figure}[h] 
\centering 
\includegraphics[width=0.45\textwidth]{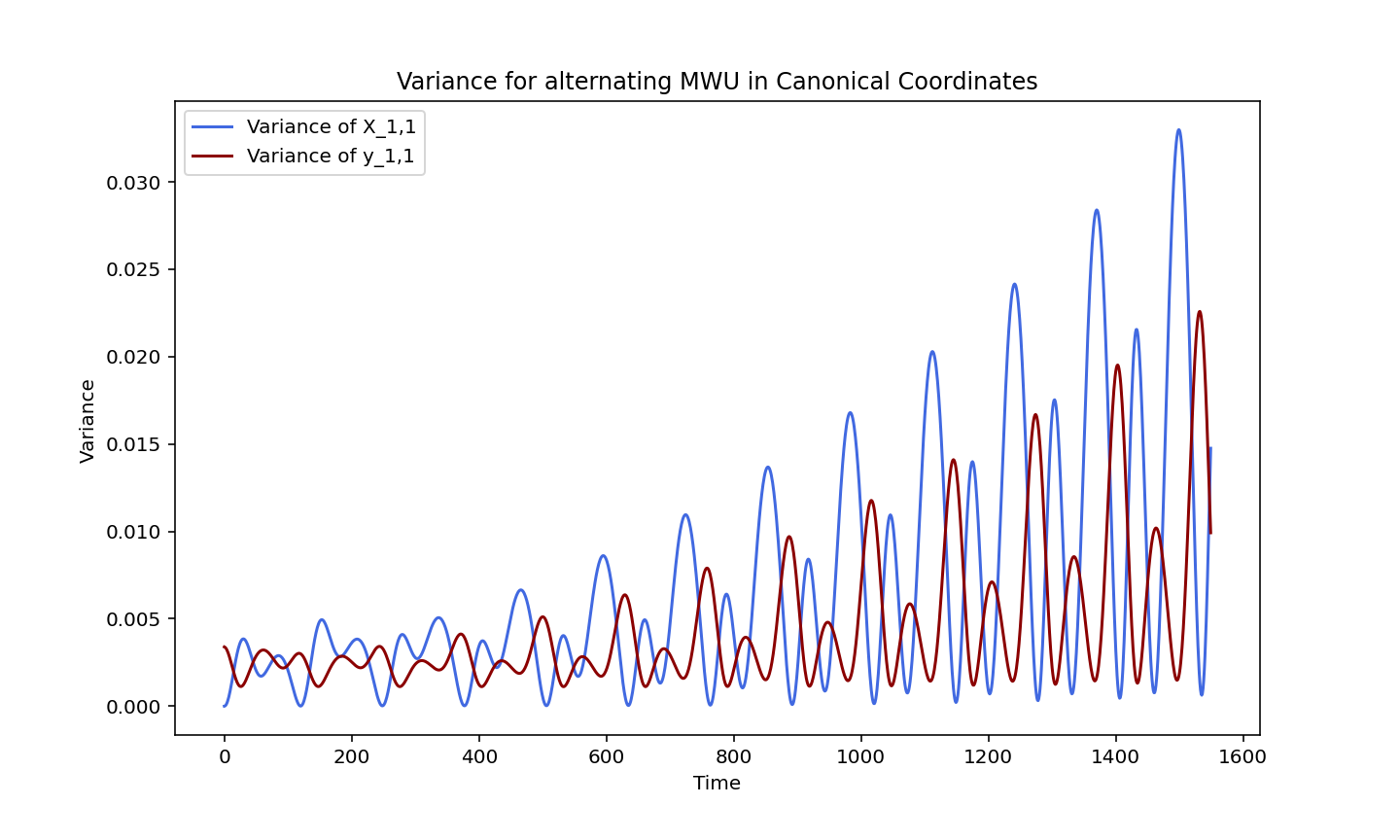} 
\caption{Covariance evolution of  $\Delta(X_{1,1}^t)$ and $\Delta(y_{1,1}^t)$ when two players use (\ref{2pMWUA}) in a randomly generated game. Covariance is calculate based on sample variance of 500 randomly generated initial conditions. }
\label{uncertaintyprinciple} 
\end{figure}

The significance of Theorem \ref{prop:nonlinear} lies not only in providing a lower bound on the covariance evolution of general FTRL dynamics. Similar to the \emph{Heisenberg Uncertainty Principle} in quantum mechanics, Theorem \ref{prop:nonlinear} implies that it is impossible for both $\Delta X_{i,\alpha}(t)$ and $\Delta y_{i,\alpha}(t)$ to be simultaneously small. A numerical experiment illustrating this point is presented in Figure \ref{uncertaintyprinciple}. In this figure, it can be observed that when the curve representing $\Delta X_{i,\alpha}(t)$ is on an increasing stage, the curve representing $\Delta y_{i,\alpha}(t)$ is on an decreasing stage, and vice versa. Moreover, the occurrence time of the local minima of $\Delta(y_{1,1}(t)) $ coincides with the local maxima of $\Delta(X_{1,1}(t))$. This implies  that $\Delta(X_{1,1}(t)) $ and $\Delta(y_{1,1}(t))$ cannot be small at the same time.

It is interesting to ask whether similar phenomena as in Theorem \ref{prop:nonlinear} occur in the primal space $(x_1(t), x_2(t))$. In Figure \ref{UIUI}, we demonstrate an experiment on the covariance evolution in the primal space, and observe that similar phenomena exist, at least when the number of pure strategies for players is small. Further discussions can be found in Appendix \ref{appdenix_actural_strategy}.

\begin{figure}[h]
    \centering
    \includegraphics[width=0.45\textwidth]{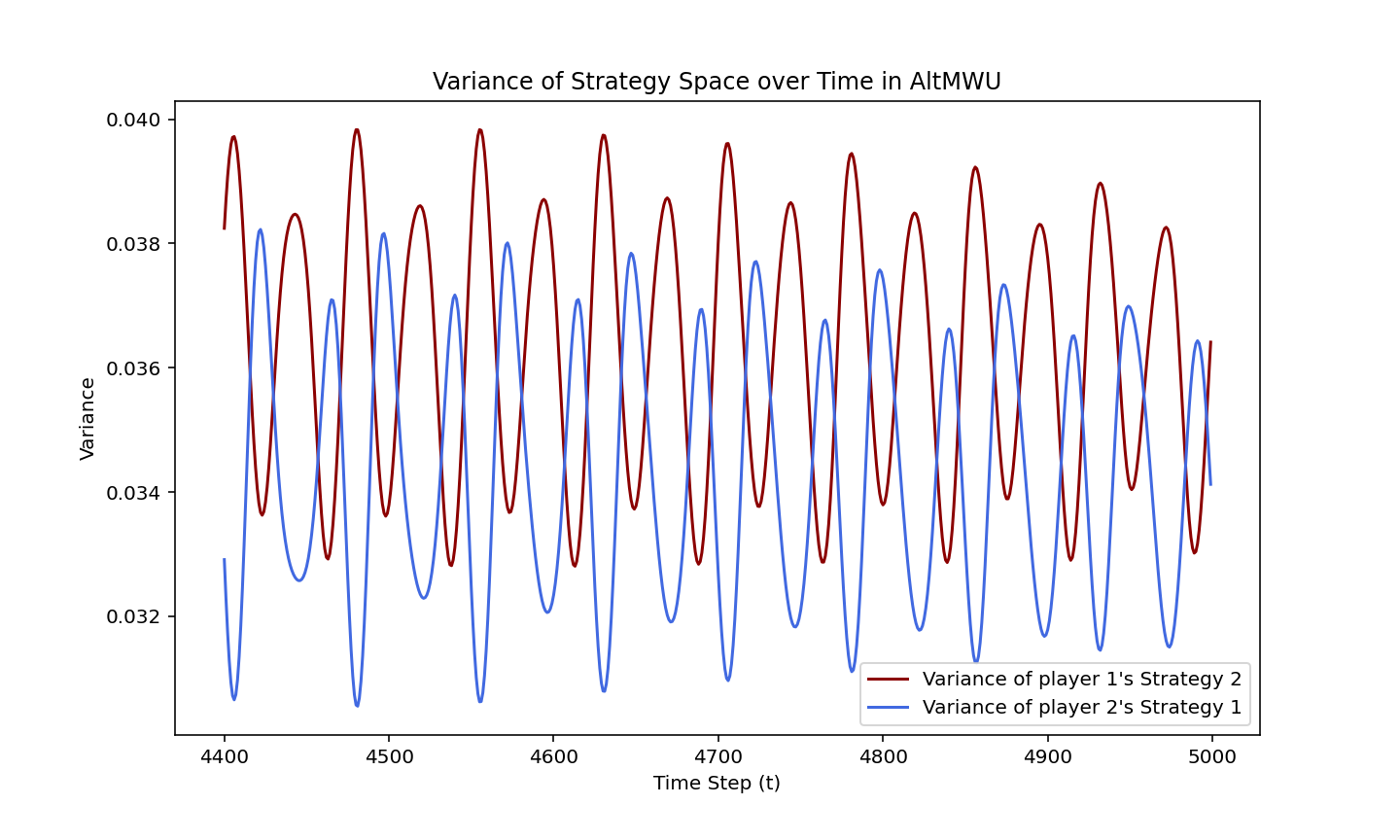}
    \caption{Covariance evolution on primal space. Calculate based on 100 random samples of the initial mixed strategies used by two players when they employ (\ref{2pMWUA}) in a randomly generated $3 \times 3$ game.}
    \label{UIUI}
\end{figure}

\section{Experiments}\label{expp}

In this section we provide numerical experiments illustrating the covariance evolution results proved for Euclidean norm regularized FTRL in Theorem \ref{Covariance Evolution in FTRL with Euclidean Regularizer}. 
More numerical experiments on the non-singular cases are presented in Appendix \ref{eeee1}. 

\textbf{Continuous time FTRL.} We illustrate how $\Var(X_{1,1}(t))$ and $\Var(y_{1,1}(t))$ evolve with continuous time FTRL with payoff matrices
\begin{align*}
    &\cdot A_1 = [[1,-1],[-1,1]], \ \cdot A_2 =[[1.2,-1.2],[-1,1]] \\
    \\
    &\cdot A_3 =[[1.5,-1.5],[-1,1]].
\end{align*}
See Figure \ref{continuous va}.
In (a), the $\Var(X_{1,1}(t))$ has a quadratic growth rate, 
and in (b) $\Var(y_{1,1}(t))$ is bounded, which support results of continuous time part in Theorem \ref{Covariance Evolution in FTRL with Euclidean Regularizer}.

\begin{figure}[h]
    \centering
    \subfigure[$\Var(X_{1,1}(t))$]{
        \includegraphics[width=1.5in]{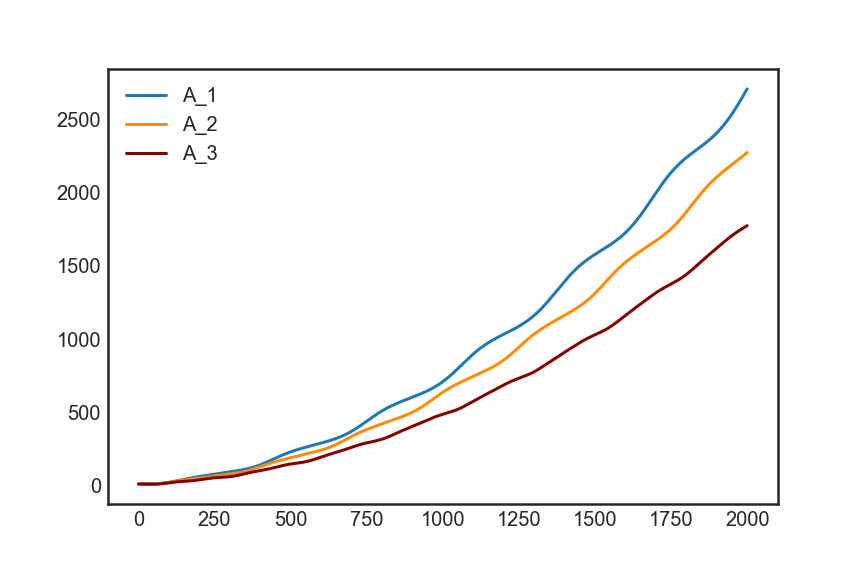}
        \label{}
    }
    \subfigure[$\Var(y_{1,1}(t))$]{
	\includegraphics[width=1.5in]{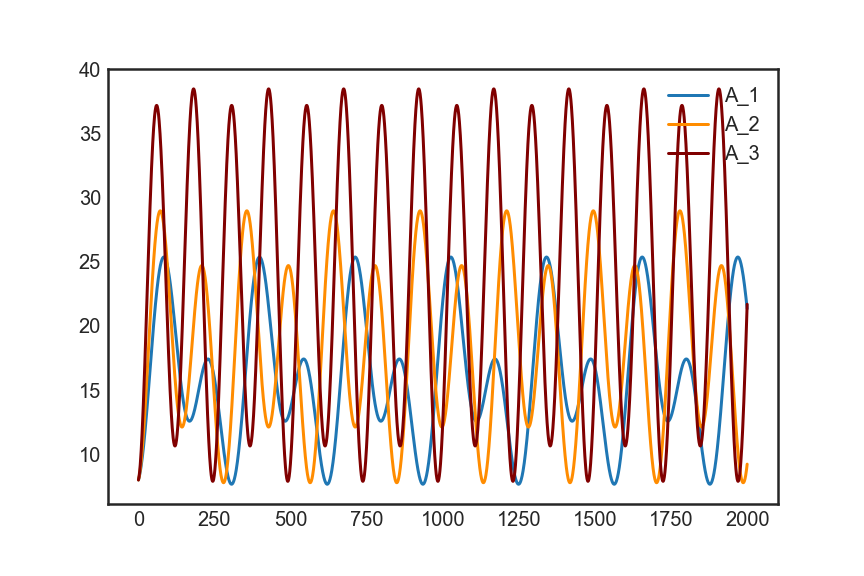}
        \label{}
    }
    \caption{Variance evolution of continuous FTRL, singular cases. }
    \label{continuous va}
\end{figure}

\textbf{Symplectic discretization.} We illustrate how $\Var(X^t_{1,1})$ and $\Var(y^t_{1,1})$ evolves with symplectic discretization, the payoff matrices are given as follows: 
\begin{align*}
    &\cdot B_1 = [[1,-1],[-1,1]],\ \cdot B_2 =[[1.2,-1.2],[-1,1]]\\ 
    \\
    &\cdot B_3 =[[1,-1.3],[-1,1.3]].
\end{align*}
See Figure \ref{sym va}. From the experimental results, we can see the variance behavior of symplectic discretization is same as continuous case, which support results of symplectic discretization part of Theorem \ref{Covariance Evolution in FTRL with Euclidean Regularizer}. 

\begin{figure}[h]
    \centering
    \subfigure[$\Var(X^t_{1,1})$]{
        \includegraphics[width=1.5in]{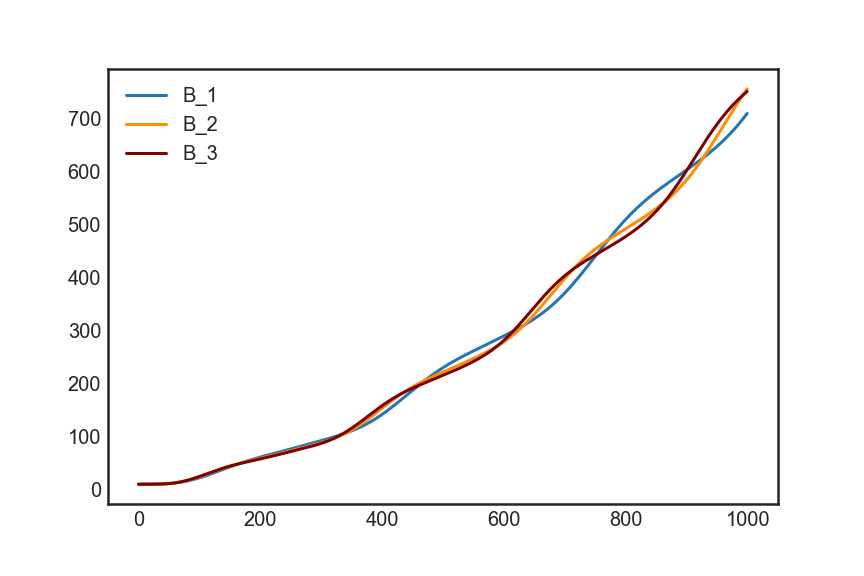}
        \label{}
    }
    \subfigure[$\Var(y^t_{1,1})$]{
	\includegraphics[width=1.5in]{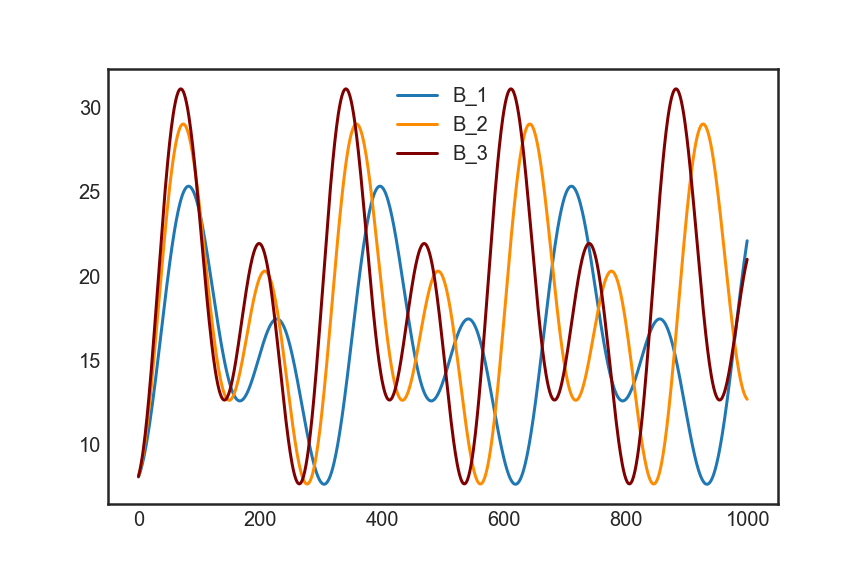}
        \label{}
    }
    \caption{Variance evolution of Symplectic discretization. }
    \label{sym va}
\end{figure}

\textbf{Euler discretization.}  We show experimental results on $\Var(X^t_{1,1}),\Var(y^t_{1,1})$ where $(X^t_1,y^t_1)$ evolve as Euler discretization and payoff matrices are given as follows:

\begin{itemize}
\item $C_1 = [[1,-1.31],[-1,1.31]]$ is  singular.
\item $C_2 = [[2,-1.7],[-1.7,1.5]]$ is non-singular.
\end{itemize}

\begin{figure}[h]
\centering
\subfigure[$C_1$, singular]{
\includegraphics[width=1.5in]{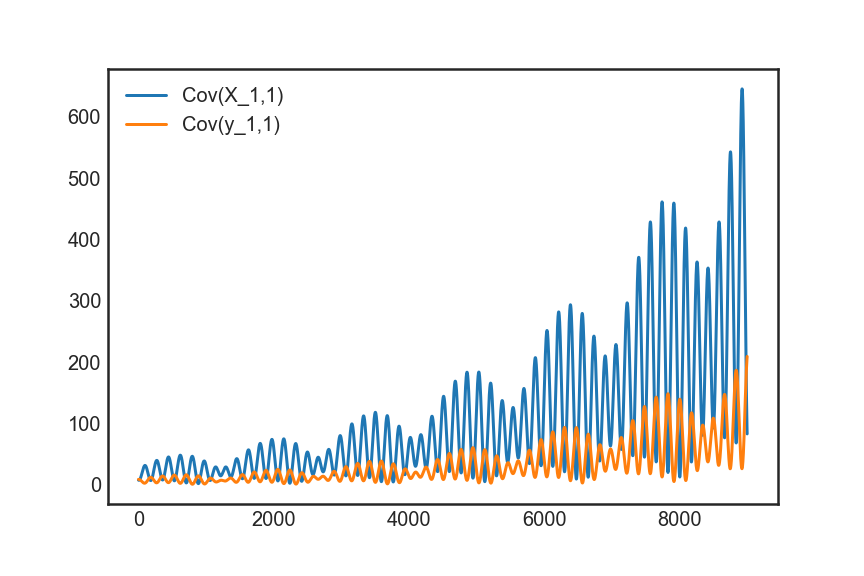}
}
\subfigure[$C_2$, non-singular  ]{
\includegraphics[width=1.5in]{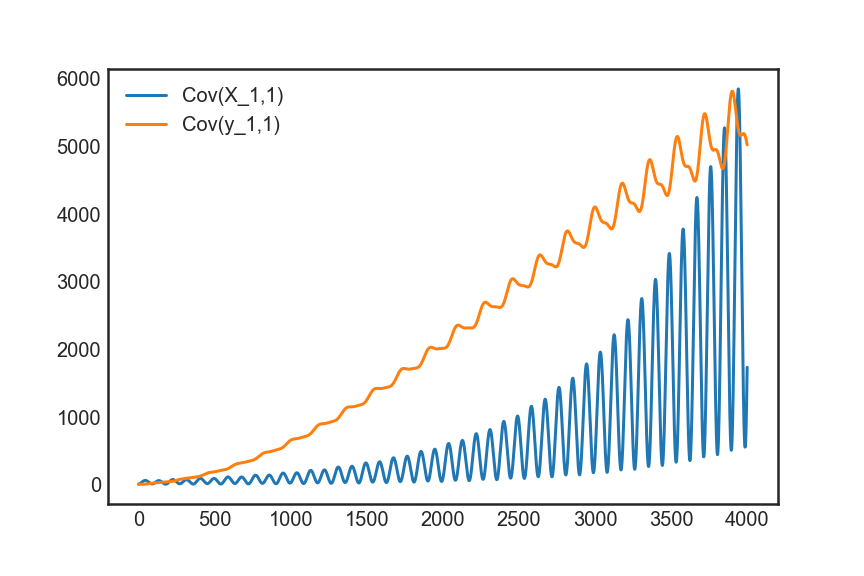}
}
\caption{Variance evolution of Euler discretization.}
\label{ned}
\end{figure}
In Figure \ref{ned} we can observe $\Var(X^t_{1,1})$ and $\Var(y^t_{1,1})$  exhibit an exponential growth rate which support the result of Euler discretization part in Theorem \ref{Covariance Evolution in FTRL with Euclidean Regularizer}.  As shown in Appendix \ref{Covariance Evolution in FTRL with Euclidean Regularizer}, the function of the covariance evolution process contains polynomials combinations of trigonometric functions, which cause the oscillations in Figure \ref{ned}.



\section{Conclusion}
In this paper we investigate the evolution of observer uncertainty in learning dynamics from a covariance perspective. We prove concrete rates of covariance evolution for different discretization schemes of FTRL dynamics and establish a Heisenberg-type uncertainty inequality that constrains the predictive ability of an observer. 
In our analysis, we leverage the techniques from symplectic geometry for analyzing the evolution of uncertainty, which to the best of our knowledge is the first of its kind. An interesting direction is to extend current results for different classes of games (e.g. potential games, multiplayer games).



\section*{Acknowledgments}

Xiao Wang acknowledges Grant 202110458 from Shanghai University of Finance and Economics and support from the Shanghai Research Center for Data Science and Decision Technology.

\section*{Impact Statement}
This paper presents work whose goal is to advance the field of Machine Learning. There are many potential societal consequences of our work, none which we feel must be specifically highlighted here.

\bibliography{Bibli}
\bibliographystyle{icml2024}

\newpage
\appendix
\onecolumn
\section{Proof of Proposition \ref{prop:primal-duel}}\label{Proof of Proposition prop:primal-duel}

\subsection{Hamiltonian formulation of FTRL}\label{Hamiltonian formulation of FTRL}
We first recall the Hamiltonian formulation of continuous FTRL in zero sum game from \citep{PB1}.

The Hamiltonian function $H(X_1,y_1)$ for agent $1$ is defined to be 
\begin{align}\label{ham1}
H(X_1,y_1) = h^*_1(y_1(t)) + h^*_2(y_2(0) + A^{(21)}X_1(t)),
\end{align}
and the Hamiltonian function $H(X_2,y_2)$ for agent $2$ is defined to be 
\begin{align}\label{ham2}
H(X_2,y_2) = h^*_2(y_2(t)) + h^*_1(y_1(0) + A^{(12)}X_2(t)),
\end{align}
where $h^*_i$ is the regularizer used by agent $i$, $i =1,2$.

Theorem 3.2 of  \citep{PB1} shows the dynamical behaviors of $(X_i(t),y_i(t))$ are completely determined by these two Hamiltonian functions. A similar non-canonical Hamiltonian system formulation (also known as a Possion system) for continuous time  mirror descent algorithms is also presented in \citep{wibisono2022alternating}. Theorem 5.1 of \citep{PB1} demonstrates that the Hamiltonian function defined here is inherently connected to the Bregman divergence, which is a commonly used concepts in optimization, plus a an additional term determined by the regularizers and equilibrium of the game.

More precisely, \citep{PB1} was shown that the cumulative strategies and payoffs of agent $1$, $(X_1,y_1)$, of continuous FTRL for agent 1 satisfies the following equations:
\begin{align}
&\frac{d}{dt} X_1(t) = \frac{\partial H}{\partial y_1}(X_1 ,y_1) = \nabla h^*_1(y_1(t)),\label{heq11} \\
&\frac{d}{dt} y_1(t) = - \frac{\partial H}{\partial X_1}(X_1, y_1) = A^{(12)} \nabla h^*_2(y_2(0) + A^{(21)}X_2(t)).\label{heq12}
\end{align}
Similarly results also hold for agent $2$, $(X_2,y_2)$, of continuous FTRL for agent 2 satisfies the following equations:
\begin{align}
&\frac{d}{dt} X_2(t) = \frac{\partial H}{\partial y_2}(X_2 ,y_2) = \nabla h^*_2(y_2(t)),\label{heq21}\\
&\frac{d}{dt} y_2(t) = - \frac{\partial H}{\partial X_2}(X_2, y_2) = A^{(21)} \nabla h^*_1(y_1(0) + A^{(12)}X_1(t)).\label{heq22}
\end{align}

The proof of Proposition \ref{prop:primal-duel} is divided into two parts : 
\begin{itemize}
\item The proof of entropy regularizers is presented  in Section \ref{Proof of Entropy regularizers},
\item The proof of Euclidean norm regularizers is presented  in Section \ref{Proof of Euclidean norm regularizers},
\end{itemize}
and in Section \ref{pdc}, we introduce the Euler and Symplectic discretization of FTRL.

\subsection{Euler and Symplectic discretization of FTRL}\label{pdc}

Both in Euler and Symplectic,  we denote the initial condition of the discrete equation on $(X^t_i,y^t_i),\ i =1,2$ to be $y^0_i = y_i(0)$ and $X^0_i = 0$.

\begin{lem}[Euler  discretization of FTRL]\label{Euler  discretization of FTRL}
Discretizing equation (\ref{HamiltonianFTRL})  with Euler method for both agent $i =1,2$ gives
\begin{align*}\label{agent 1 Euler discretize equation}
\tag{agent 1 Euler discretize equation}
&X^{t+1}_1 =X^{t}_1 + \eta  \frac{\partial H}{\partial y_1}(X^t_1,y^{t+1}_1) =X^{t}_1 + \eta \nabla h^*_1(y^{t}_1), \\
&y^{t+1}_1 =y^{t}_1 - \eta \frac{\partial H}{\partial X_1}(X^t_1,y^t_1) =y^{t}_1 + \eta  A^{(12)} \nabla h^*_2 (y^0_2 + A^{(21)} X^{t}_1),  
\end{align*}
 and 
\begin{align*}\label{agent 2 Euler discretize equation}
\tag{ agent 2 Euler discretize equation}
&X^{t+1}_2 =X^{t}_2 + \eta \frac{\partial H}{\partial y_2}(X^t_2,y^t_2) = X^t_2 + \eta \nabla h^*_2(y^{t}_2), \\
&y^{t+1}_2 =y^{t}_1 - \eta  \frac{\partial H}{\partial X_2} (X^{t+1}_2,y^{t}_2)=y^{t}_2 + \eta  A^{(21)} \nabla h^*_1 (y^0_1 + A^{(12)} X^{t}_2) . 
\end{align*}
\end{lem}

\begin{proof}Note that in Euler discretization, we use the derivative on point of $t$-th round  to find the point of $t+1$ round. Thus (\ref{agent 1 Euler discretize equation}) directly follows from applying Euler discretization to \ref{heq11} and \ref{heq12}. Similarly, (\ref{agent 2 Euler discretize equation}) directly follows from applying Euler discretization to \ref{heq21} and \ref{heq22}.
\end{proof}

\begin{lem}[Symplectic discretization of FTRL] \label{Symplectic discretization of FTRL}
Discretizing (\ref{HamiltonianFTRL}) for $i=1$ with \uppercase\expandafter{\romannumeral1}-type Euler symplectic method (\ref{type1}) gives

\begin{align*}\label{agent 1 Symplectic discretize equation}
\tag{agent 1 Symplectic discretize equation}
&y^{t+1}_1 =y^{t}_1 - \eta \frac{\partial H }{\partial X_1} (X^t_1,y^t_1)=y^{t}_1 + \eta  A^{(12)} \nabla h^*_2 (y^0_2 + A^{(21)} X^{t}_1)  \\
&X^{t+1}_1 =X^{t}_1 + \eta  \frac{\partial H }{\partial y_1}(X^t_1,y^{t+1}_1) =X^{t}_1 + \eta \nabla h^*_1(y^{t+1}_1) \\
\end{align*}
 and discrete  (\ref{HamiltonianFTRL}) for $i=2$  with  \uppercase\expandafter{\romannumeral2}-type Euler symplectic method (\ref{type2})gives
\begin{align*}\label{agent 2 Symplectic discretize equation}
\tag{agent 2 Symplectic discretize equation}
&X^{t+1}_2 =X^{t}_2 + \eta \frac{\partial H }{\partial y_2}(X^t_2,y^t_2) = X^t_2 + \eta \nabla h^*_2(y^{t}_2) \\
&y^{t+1}_2 =y^{t}_1 - \eta  \frac{\partial H}{\partial X_2}(X^{t+1}_2,y^{t}_2) =y^{t}_2 + \eta  A^{(21)} \nabla h^*_1 (y^0_1 + A^{(12)} X^{t+1}_2)  \\
\end{align*}
\end{lem}

\begin{proof}(\ref{agent 1 Symplectic discretize equation}) directly follows from applying  (\ref{type1}) to equation  \ref{heq11} and \ref{heq12}. Similarly, (\ref{agent 2 Symplectic discretize equation}) directly follows from applying  (\ref{type2}) to equation  \ref{heq21} and \ref{heq22}.
\end{proof}

We define $(x^t_1,x^t_2)$ to be 
\begin{align}\label{ss}
x^{t}_1  = \frac{X^{t+1}_1 - X^{t}_1 }{\eta} , \ x^{t}_2   = \frac{X^{t+1}_2 - X^{t}_2}{\eta}.
\end{align}

In the case of  Euler  discretization of FTRL (Lemma \ref{Euler  discretization of FTRL}), we have 
\begin{align}\label{sse}
x^{t}_1  = \nabla h^*_1(y^{t}_1), \ x^t_2 = \nabla h^*_2(y^{t}_2),
\end{align}
and in the case of Symplectic discretization of FTRL (Lemma \ref{Symplectic discretization of FTRL}), we have 
\begin{align}\label{sss}
x^{t}_1  = \nabla h^*_1(y^{t+1}_1), \ x^t_2 = \nabla h^*_2(y^{t}_2).
\end{align}
Note that in Symplectic method, $x^{t}_1$ is determined by $y_1^{t+1}$, but in Euler method, $x^{t}_1$ is determined by $y_1^{t}$.

In the following, we will show $(x^t_1,x^t_2)$ evolves as (\ref{2pMWUS}) under Euler method or (\ref{2pMWUA}) under Symplectic method on the strategy space if the regularizers $h_i(\cdot)$ are choose to be entropy functions, and the constrained sets $\mathcal{X}_i$ are chosen to be simplexes for $i=1,2$, this exactly the second part of Proposition \ref{prop:primal-duel}. 
\begin{lem}\label{cpayoff}
Both in Euler discretization of FTRL dynamics and  Symplectic discretization of FTRL dynamics, the equalities
\begin{align}
y^{n}_1 = y^0_1 + A^{(12)}X^{n}_2 \label{yX1} \\
y^{n}_2 = y^0_2 + A^{(21)}X^{n}_1 \label{yX2}
\end{align}
hold for any $n \ge 0$.
\end{lem}

\begin{proof}Here we only prove the case of Symplectic discretization of FTRL dynamics, as the case of Euler discretization of FTRL dynamics is similar. We prove this by induction. For $n=0$, (\ref{yX1}) and (\ref{yX2}) are 
\begin{align}
&y^{0}_1 = y^0_1 + A^{(12)}X^{0}_2 \\
&y^{0}_2 = y^0_2 + A^{(21)}X^{0}_1,
\end{align}
which hold trivially since by definition $X^{0}_1 = X^{0}_2 = 0$. 

Now assume (\ref{yX1}) and (\ref{yX2}) hold for $n$, i.e,
\begin{align}
&y^{n}_1 = y^0_1 + A^{(12)}X^{n}_2 \label{yX1n}\\
&y^{n}_2 = y^0_2 + A^{(21)}X^{n}_1 \label{yX2n}.
\end{align}

Then, we have
\begin{align}
y^{n+1}_1 & = y^n_1 + \eta A^{(12)} \nabla h^*_2 (y^0_2 + A^{(21)} X^n_1)   \\
& \overset{(\ref{yX2n})}{= }y^n_1 + \eta A^{(12)} \nabla h^*_2 ( y^n_2) \\
& = y^n_1 + \eta A^{(12)} x^n_2 \\
& \overset{(\ref{yX1n})}{= } y^0_1 + A^{(12)} X^n_2 + \eta A^{(12)}x^n_2 \\
& = y^0_1 + A^{(12)} X^{n+1}_2. \label{yX1n+1}
\end{align}

Moreover, we have
\begin{align}
y^{n+1}_2 & = y^n_2 + \eta A^{(21)} \nabla h^*_1 (y^0_1 + A^{(12)} X^{n+1}_2)   \\
& \overset{(\ref{yX1n+1})}{= } y^n_2 + \eta  A^{(21)} \nabla h^*_1(y^{n+1}_1) \\
& = y^n_2 + \eta  A^{(21)}x^n_1 \\
& \overset{(\ref{yX2n})}{= } y^0_2 + A^{(21)} X^n_1 + \eta A^{(21)} x^n_1 \\
& = y^0_1 + A^{(21)} X^{n+1}_1
\end{align}
This finish the proof.
\end{proof}

\subsection{Proof of Entropy regularizers}\label{Proof of Entropy regularizers}

\begin{lem} For entropy regularizer $h(x) = \sum^n_{i=1}x_i \ln x_i$ with simplex constrain, i.e., $\Delta = \{  x \in \BR^n |  \sum^n_{i=1} x_i = 1, x_i \ge 0 \}$, we have
\begin{align}\label{nablaentropy}
\nabla h^*(y) = \left( \frac {e^{y_i} } { \sum^n_{s=1} e^{y_s} } \right)^n_{i=1}.
\end{align}
\end{lem}

\begin{proof} By the definition, 
$
h^*(y)=\max_{x\in\Delta} ( \langle x, y\rangle-h(x) ),
$
and
$
\nabla h^*(y)=\arg\max_{x\in\Delta} ( \langle x, y\rangle-h(x) ).
$ 

Denote $f(x)=\left(\langle x, y\rangle-h(x)\right)$, by the KKT condition,   if $x^*$ is the maximum of $f$, then there exist $\mu_i$ and $\lambda$ such that
\[
-\nabla f(x^*)+\sum_{i=1}^n\mu_i\nabla g_i(x^*)+\lambda\nabla h(x^*)=0
\]
and 
\begin{align*}
&g_i(x^*)=-x^*_i\le 0 \ \ \text{for all}\ \ i\in[n]
,\\ 
&h(x^*)=\sum_{i=1}^nx_i^*-1=0
,\\
&\mu_ig_i(x^*)=0\ \ \text{for all}\ \ i\in[n].
\end{align*} 
Since the gradient of $f$ can be computed to be
\[
\nabla f(x)=y-(\log x_1+1,...,\log x_n+1),
\]
the KKT condition becomes
\[
-y+(\log x_1+1,...,\log x_n+1)+\sum_{i=1}^n\mu_i(0,...,-1,...,0)+\lambda (1,...,1)=0.
\]
Suppose the feasible $x^*$ is interior point of $\Delta$, i.e., $x_i>0$, then we have for all $i\in[n]$, $\mu_i=0$. Then the KKT condition is reduced to the following equations
$
\log x_i+1+\lambda =y_i \ \ \text{for all}\ \ i\in[n]
,
\sum_{i=1}^nx_i =1.
$
This gives solution of $x_i$ and $\lambda$:
\begin{align*}
    x_i=\frac{e^{y_i}}{\sum_{s=1}^ne^{y_s}} \ \ \text{for all} \ \ i\in[n],\ \ 
\lambda=\log\left(\sum_{s=1}^ne^{y_s}\right)-1,
\end{align*}
thus we have completed the proof.
\end{proof}

\begin{lem}\label{mwu222}
The $(x^t_1, x^t_2)$ in (\ref{sse}) with entropy regularizer is the same as (MWU).
\end{lem}

\begin{proof}

In (\ref{sse}), we have $x^t_1  = \nabla h^*_1(y^{t}_1)$, thus
\begin{align}
x^t_1 & = \nabla h^*_1(y^{t}_1)\\
\\
& \overset{(\ref{nablaentropy})}{=}  \left ( \frac{e^{y^{t}_{1,s}}}{ \sum^{n_1}_{j=1} e^{y^{t}_{1,j}}  } \right )^{n_1}_{s = 1}\\
\\
& \overset{(\ref{yX1})}{=}  \left ( \frac{e^{y^0_{1,s} + (A^{(12)}X^{t}_2 )_s }}{    \sum^{n_1}_{j=1} e^{ y^0_{1,j} + (A^{(12)}X^{t}_2)_j  }}\right )^{n_1}_{s = 1} \\
\\
& =  \left ( \frac{e^{y^0_{1,s} + \eta (A^{(12)} \sum^{t-1}_{k=1}x^{k}_2 )_s }}{    \sum^{n_1}_{j=1}   e^{y^0_{1,j} + \eta (A^{(12)} \sum^{t-1}_{k=1}x^{k}_2 )_j }  }\right )^{n_1}_{s = 1} \\
\\
& = \left (\frac{ x^{t-1}_{1,s} e^{\eta ( A^{(12)} x^{t-1}_2   )_s}   }{ \sum^{n_1}_{j=1} x^{t-1}_{1,j} e^{\eta ( A^{(12)} x^{t-1}_2   )_j}} \right )^{n_1}_{s = 1}. 
\end{align}

The case of 2 agent is exactly same as  1 agent as they are symmetry, and we have
\begin{align}
x^t_2 &  =\left (\frac{ x^{t-1}_{2,s} e^{\eta ( A^{(21)} x^{t-1}_2   )_s}   }{ \sum^{n_2}_{j=1} x^{t-1}_{2,j} e^{\eta ( A^{(21)} x^{t-1}_2   )_j}} \right )^{n_2}_{s = 1}. 
\end{align}

That is same as (MWU) .
\end{proof}

\begin{lem}\label{mwu111}
The $(x^t_1, x^t_2)$ in (\ref{sss}) with entropy regularizer is the same as (AltMWU).
\end{lem}

\begin{proof}For 1 agent, from (\ref{sss}), we have $x^t_1  = \nabla h^*_1(y^{t+1}_1)$, thus
\begin{align}
x^t_1 & = \nabla h^*_1(y^{t+1}_1)\\
\\
& \overset{(\ref{nablaentropy})}{=}  \left ( \frac{e^{y^{t+1}_{1,s}}}{ \sum^{n_1}_{j=1} e^{y^{t+1}_{1,j}}  } \right )^{n_1}_{s = 1}\\
\\
& \overset{(\ref{yX1})}{=}  \left ( \frac{e^{y^0_{1,s} + (A^{(12)}X^{t+1}_2 )_s }}{    \sum^{n_1}_{j=1} e^{ y^0_{1,j} + (A^{(12)}X^{t+1}_2)_j  }}\right )^{n_1}_{s = 1} \\
\\
& =  \left ( \frac{e^{y^0_{1,s} + \eta (A^{(12)} \sum^t_{k=1}x^{k}_2 )_s }}{    \sum^{n_1}_{j=1}   e^{y^0_{1,j} + \eta (A^{(12)} \sum^t_{k=1}x^{k}_2 )_j }  }\right )^{n_1}_{s = 1} \\
\\
& = \left (\frac{ x^{t-1}_{1,s} e^{\eta ( A^{(12)} x^t_2   )_s}   }{ \sum^{n_1}_{j=1} x^{t-1}_{1,j} e^{\eta ( A^{(12)} x^t_2   )_j}} \right )^{n_1}_{s = 1}. \label{mwux1}
\end{align}

Note that the update rule of $x^t_1$ use $x^t_2$, this is a characteristic of (AltMWU).

For 2 agent, from (\ref{sss}) we have $x^{t}_2   =   \nabla h^*_2(y^{t}_2)$, thus
\begin{align}
x^t_2 & = \nabla h^*_2(y^{t}_2)\\
\\
& = \left ( \frac{e^{y^{t}_{2,s}}}{ \sum^{n_2}_{j=1} e^{y^{t}_{2,j}}  } \right )^{n_2}_{s = 1}\\
\\
& \overset{(\ref{yX2})}{=}  \left ( \frac{e^{y^0_{2,s} + (A^{(21)}X^{t}_1 )_s }}{    \sum^{n_2}_{j=1} e^{ y^0_{2,j} + (A^{(21)}X^{t}_1)_j  }}\right )^{n_2}_{s = 1} \\
\\
& =  \left ( \frac{e^{y^0_{2,s} + \eta (A^{(21)} \sum^{t-1}_{k=1}x^{k}_1 )_s }}{    \sum^{n_2}_{j=1}   e^{y^0_{2,j} + \eta (A^{(21)} \sum^{t-1}_{k=1}x^{k}_1 )_j }  }\right )^{n_2}_{s = 1} \\
\\
& =\left (\frac{ x^{t-1}_{2,s} e^{\eta ( A^{(21)} x^{t-1}_1   )_s}   }{ \sum^{n_2}_{j=1} x^{t-1}_{2,j} e^{\eta ( A^{(21)} x^{t-1}_1   )_j}} \right )^{n_2}_{s = 1}. \label{mwux2}
\end{align}

Combine (\ref{mwux2}) and (\ref{mwux1}),  we can see the update rule of $(x_1^t, x_2^t) $ is same as (AltMWU) .
\end{proof}

Combine Lemma \ref{mwu222} and Lemma \ref{mwu111}, we proved the second part of Proposition \ref{prop:primal-duel}. The first part is very similar, except the regularizers are changed to Euclidian norm. However, this change will not affect the proof,  so we omit it here.

\subsection{Proof of Euclidean norm regularizers}\label{Proof of Euclidean norm regularizers}

Note that for Euclidean norm regularizers, i.e., $h_i(x) = \lVert x \lVert^2 $, we have
\begin{align}\label{eunabla}
\nabla h^*_i(y) = \arg \max_{x  \in \BR^n} \{ \langle x,y \rangle -   \lVert x \lVert^2  \} = y.
\end{align}

\begin{lem}\label{3845-}
The $(x^t_1,x^t_2)$ in  (\ref{sse}) with Euclidean norm regularizer is the same as (GDA).
\end{lem}

\begin{proof}

For agent 1, we have
\begin{align}
x^t_1  \overset{(\ref{sse})}{=} \nabla h^*_1(y_1^{t}) \overset{(\ref{eunabla})}{=} y^{t}_1  = x_1^{t-1} + \eta \cdot A^{(12)}x_2^{t-1}.
\end{align}

Agent 2 is exactly same as agent 1 since in Euler method, two agent are symmetry. Thus we have shown the update rule of $(x^t_1,x^t_2)$ is same as (GDA).
\end{proof}

\begin{lem}\label{3845-1}
The $(x^t_1,x^t_2)$ in  (\ref{sss}) with Euclidean norm regularizer is the same as (AltGDA).
\end{lem}

\begin{proof}

For agent 1, we have
\begin{align}
x^t_1  \overset{(\ref{sss})}{=} \nabla h^*_1(y_1^{t+1}) \overset{(\ref{eunabla})}{=} y^{t+1}_1  = x_1^{t-1} + \eta \cdot A^{(12)}x_2^t.
\end{align}

For agent 2, we have

\begin{align}
x^t_2  \overset{(\ref{sss})}{=} \nabla h^*_1(y_2^{t}) \overset{(\ref{eunabla})}{=} y^{t}_2  = x_2^{t-1} + \eta \cdot A^{(21)}x_2^{t-1}.
\end{align}
Thus we have shown the update rule of $(x^t_1,x^t_2)$ is same as (AltGDA).

\end{proof}

\section{Proof of Section \ref{DE in FTRL}}\label{Proof of Section DE in FTRL}
In this appendix we prove results in Section \ref{DE in FTRL}. Proposition \ref{entropys} is proved in Section \ref{Proof of Proposition entropys}, and Proposition \ref{entropya} is proved
in Section \ref{Proof of Proposition entropya}. In fact, we prove a more general result, which states that when two players choose arbitrary regularizers that satisfies strongly convex and Lipschitz gradient condition except a bounded region on the domain, then the differential entropy of Euler discretization has linear growth rate, while  differential entropy of Symplectic discretization keeps constant. Note that both Euclidian norm regularizer and entropy regularizer satisfy these conditions, for example, entropy regularizer is 1-strongly convex on the interior points of simplex and has Lipschitz gradient except an arbitrary small neighbourhood of zero point.

The main technical lemma for proving Proposition \ref{entropys}  is Lemma \ref{smallstep}, which states for sufficient small step size, the update rule of Euler discretization of FTRL is an injective map. This injective property is necessary for calculating the evolution of differential entropy, see Lemma \ref{entropyevo}. The proof of Proposition \ref{entropya} is easier, as the symplectic discretization is naturally an injective map.

\subsection{Evolution of differential entropy under diffeomorphism}
The following result and its proof are informally stated in \citep{CPT2022}, for convenience of applying their statement later, we formulate it into a lemma as follows.
\begin{lem}\label{entropyevo}
Let $X \in \BR^d$ be a random vector with probability density function $g(x)$ and the support set of  $g(x)$ is $\mathcal{X}$. Assume $f : \BR^d \to \BR^d$ be a diffeomorphism, thus $f(X)$ is a random vector. Then we have
\begin{align}\label{formulaent}
S(f(X)) = S(X) + \int_{\mathcal{X}} g(x) \log \left( \lvert \det J_f(x) \lvert \right) dx
\end{align}
where $J_f(x)$ is the Jacobian matrix of $f$ at point $x \in \BR^d$.
\end{lem}

\begin{proof}
Denote $Y = f(X)$, and let $\widehat{g}(Y)$ represent the probability density function of $Y$, and $\mathcal{Y}$ be the support set of $Y$. 

Then we have
\begin{align}
S(Y) & = S(f(X)) = - \int_{\mathcal{Y}} \widehat{g}(y) \cdot \log \left( \widehat{g}(y) \right) dy \\
& = - \int_{\mathcal{Y}} g \left( f^{-1}(y) \right) \lvert \det J_{f^{-1}} (y) \lvert \cdot \log \left(  g(f^{-1}(y) ) \lvert \det J_{f^{-1}} (y) \lvert \right) dy \\
& = - \int_{\mathcal{X}} g(x) \lvert \det J_{f^{-1}} \left(  f(x) \right) \lvert \cdot \log \left(  g( x ) \lvert \det J_{f^{-1}} ( f(x) ) \lvert \right) \cdot  \lvert \det J_f(x) \lvert dx \\
& = - \int_{\mathcal{X}} g(x) \cdot \log \left(  g( x ) \lvert \det J_{f^{-1}} ( f(x) ) \lvert \right)  dx \label{ddddd} \\
& =  - \int_{\mathcal{X}} g(x) \log \left(  g(x)  \right) dx -  \int_{\mathcal{X}} g(x) \log \left(  \lvert   \det J_{f^{-1}} \left( f(x) \right)  \lvert  \right) dx \\
& = S(X)  + \int_{\mathcal{X}} g(x) \log \left(   \lvert \det J_f(x) \lvert  \right) dx,
\end{align}
where (\ref{ddddd}) comes from the inverse function theorem, which states 
\begin{align}
J_{f^{-1}} \left( f(x) \right) = \left( J_f(x) \right)^{-1}.
\end{align}
\end{proof}

\subsection{Technical lemmas for Proposition \ref{entropys}}
We first present several lemmas used later. 

\begin{lem}[Corollary 2.2 and 2.3 of \citep{hong1991jordan}]\label{ABeig}
Let $A,B \in \BR^{n \times n}$ be symmetry and positive semidefinite. Then $AB$ is diagonalizable and has nonnegative eigenvalues. Moreover, if $A$ is positive definite, then
the number of positive eigenvalues, negative eigenvalues, and $0$ eigenvalues of $AB$ are the same as $B$.
\end{lem}

\begin{lem}\label{AAtop}
If $A \in \BR^{n \times n}$ is a symmetry matrix, and $\lambda$ is an eigenvalue of $A$, then $\lambda^2$ is an eigenvalue of $A^2$.
\end{lem}

\begin{proof}
Since $A$ is a symmetry matrix, there is an invertible matrix $P$ makes 
\begin{align}
P \cdot A \cdot P^{-1} = 
\begin{bmatrix}
\lambda_1 & & \\
& \ddots & \\
& & \lambda_n
\end{bmatrix},
\end{align}
where $\lambda_1 ,..., \lambda_n$ are eigenvalues of $A$. Thus
\begin{align}
P \cdot (A)^2 \cdot P^{-1} & = \left(    P \cdot A \cdot P^{-1} \right) \cdot \left(    P \cdot A \cdot P^{-1} \right) \\
& = 
\begin{bmatrix}
(\lambda_1)^2 & & \\
& \ddots & \\
& & (\lambda_n)^2
\end{bmatrix},
\end{align}
this implies $\{ \lambda^2_i \}$ are eigenvalues of $A^2$.
\end{proof}

The following lemma is the standard Fenchel duality property, a proof can be found in  Theorem 1 \citep{zhou2018fenchel}.
\begin{lem}\label{dual}
Let $h : \mathcal{X} \to \BR$ be a $\mu$-strongly convex function with $L$-Lipschitz continuous gradient, let
\begin{align}
h^* (y) = \max_{x \in \mathcal{X}} \{ \langle x,y \rangle - h(x) \} 
\end{align}
be the convex conjugate of $h$, then we have 
\begin{itemize}
\item[(1)] $h^*$ is a $\frac{1}{L}$-strongly convex function.
\item[(2)] $h^*$ has  $\frac{1}{\mu}$-Lipschitz continuous gradient.
\end{itemize}
\end{lem}

\begin{lem}\label{inj}
Let $f : \BR^n \to \BR^n$ be a differentiable function on a convex set $U \subset \BR^n$, and 
\begin{align}
\lVert J_f (x) - I  \lVert < 1
\end{align}
 for any $x \in U$, where $\lVert \cdot \lVert$ is the $L^2$-operator norm, then $f$ is an injective map.
\end{lem}

\begin{proof}
Let $g(x) = f(x) - x$. Then for any $x \ne y$, we have
\begin{align}
\lVert f(x) - f(y) + y -x \lVert  & = \lVert g(x) - g(y) \lVert \\
& = \lVert J_g (\zeta) (x-y) \lVert \label{meanvalue} \\
& \le \lVert J_g (\zeta) \lVert \cdot \lVert x-y \lVert \\
& < \lVert x - y \lVert \label{mean1} ,
\end{align}
where (\ref{meanvalue}) use the mean value theorem, and (\ref{mean1}) is due to the fact that $\lVert J_g(\zeta) \lVert < 1$ for any $\zeta \in U$.  Thus $f(x) \ne f(y)$.
\end{proof}

\begin{lem}\label{smallstep}
 If the step size $\eta < \min \{\mu_1, \mu_2/ \lVert A^{(21)} \lVert^2 \}$, then the iterate map 
\begin{align}
\phi : (X_1^n, y_1^n) \to (X_1^{n+1}, y_1^{n+1})
\end{align} 
of Euler discretization of FTRL in Lemma \ref{Euler  discretization of FTRL} is an injective function.
\end{lem}

\begin{proof}Recall the iterate map $\phi : (X_1^n, y_1^n) \to (X_1^{n+1}, y_1^{n+1}) $ can be written as an Euler discretization with the following form
\begin{align}
y^{n+1}_1 =  y^n_1 - \eta \cdot \frac{\partial H}{\partial X_1} (X^n_1,y^n_1) \\
X^{n+1}_1 = X^n_1 + \eta \cdot \frac{\partial H}{\partial y_1}(X^n_1,y^n_1)
\end{align}
and the Hamiltonian function has form
\begin{align}
H(X_1,y_1) = h^*_1(y_1)+h^*_2(y_2(0) + A^{(21)} X_1).
\end{align}

Note that $H_1(X_1,y_1)$ is separable, i.e., $h^*_1( \cdot )$ is independent with $X_1$ and $h^*_2( \cdot )$ is independent with $y_1$, thus we have

\begin{align}
\frac{\partial^2 H}{\partial X_1 \partial y_1} = 0,\ \frac{\partial^2 H}{\partial y_1 \partial X_1}= 0.
\end{align}

Next we calculate the Jacobin matrix of $\phi$,

\begin{align}
J_{\phi} = 
\begin{bmatrix}
\frac{\partial y^{n+1}_1}{\partial{y^n_1}} & \frac{\partial y^{n+1}_1}{\partial{X^n_1}}    \\
\\
\\
\frac{\partial X^{n+1}_1}{\partial{y^n_1}} & \frac{\partial X^{n+1}_1}{\partial{X^n_1}}
\end{bmatrix}
& =
\begin{bmatrix}
I  - \eta \frac{\partial^2 H}{\partial y_1 \partial X_1} (X^n_1,y^n_1)    & - \eta  \frac{\partial^2 H}{ \partial^2 X_1} (X^n_1,y^n_1)   \\
\\
\\
\\
\eta  \frac{\partial^2 H}{ \partial^2 y_1} (X^n_1,y^n_1) & I + \eta \frac{\partial^2 H}{\partial X_1 \partial y_1} (X^n_1,y^n_1) 
\end{bmatrix} \\
\nonumber \\
& =
\begin{bmatrix}
I & - \eta (A^{(21)})^{\top} \cdot  \nabla^2 h^*_2 \cdot A^{(21)}     \\
\\
\\
\eta \nabla^2 h^*_1 & I
\end{bmatrix} \label{jacphi},
\end{align}

and 

\begin{align}
J_{\phi} - I= 
\begin{bmatrix}
0 & - \eta (A^{(21)})^{\top}\cdot \nabla^2 h^*_2 \cdot A^{(21)} \\
\\
\\
\eta \nabla^2 h^*_1 & 0
\end{bmatrix}.
\end{align}
Since $h_i$ is $\mu_i$-strongly convex, by Lemma \ref{dual}, $h^*_i$ has $\frac{1}{\mu_i}$-Lipschitz continuous gradient, thus we have
\begin{align}\label{normnabla}
\lVert \nabla^2 h^*_i \lVert \le \frac{1}{\mu_i}
\end{align}
holds at arbitrary points within the domain of $h^*_i$.

Next we estimate the $L^2$-operator norm of the matrix $J_{\phi} - I$, since the $L^2$-operator norm is
equivalent to the spectral norm, we have

\begin{align}
\lVert J_{\phi} - I \lVert & = \sqrt{\lambda_{\max} \left(  (J_{\phi} - I)^{\top} \cdot  (J_{\phi} - I)  \right)},
\end{align}

and
\begin{align}
 (J_{\phi} - I)^{\top} \cdot  (J_{\phi} - I) & =
\eta^2 \cdot
\begin{bmatrix}
\left( \nabla^2 h^*_1 \right)^2 & 0 \\
\\
0 & \left( (A^{(21)})^{\top}\cdot \nabla^2 h^*_2 \cdot A^{(21)} \right)^2
\end{bmatrix}.
\end{align}

Since both $\nabla^2 h^*_1$ and $ (A^{(21)})^{\top}\cdot \nabla^2 h^*_2 \cdot A^{(21)}$ are symmetry matrix, thus by Lemma \ref{AAtop}, eigenvalues
of $(J_{\phi} - I)^{\top} \cdot  (J_{\phi} - I)$ has form $\eta^2 \lambda^2$, where $\lambda$ is an eigenvalue of
$\nabla^2 h^*_1$ or $ (A^{(21)})^{\top}\cdot \nabla^2 h^*_2 \cdot A^{(21)}$.

Note that $h^*_i$ is a function of $X_1,y_1$, thus $\lambda$ is also a function of $X_1,y_1$, and it is not clear whether there is an upper bound on $\lambda$ without more information on the Hessian matrix of $h^*_i$.  However, as we have shown in (\ref{normnabla}), the $L^2$-operator norm of $\nabla^2 h^*_i$ has an upper bound  $\frac{1}{\mu_i}$, thus we have
\begin{align}
 0 \le \lambda < \max \left\{ \frac{1}{\mu_1},   \frac{ \lVert A^{(21)} \lVert^2 }{\mu_2}   \right\}.
\end{align}
Thus we can choose  $\eta < \min \left\{ \mu_1,   \frac{ \mu_2}{\lVert A^{(21)} \lVert^2}   \right\}$ to make 
\begin{align}
\lVert  J_{\phi} - I  \lVert  < 1,
\end{align}
and by Lemma \ref{inj}, $\phi$ is an injective map.
\end{proof}

\subsection{Proof of Proposition \ref{entropys}}\label{Proof of Proposition entropys}
\begin{proof}
By Lemma \ref{smallstep}, the iterate map of simultaneous FTRL is an diffeomorphism, thus we can use Lemma \ref{entropyevo} to calculate the evolution of differential entropy in simultaneous FTRL. We firstly prove differential entropy is a non-decrease function.

Recall (\ref{ham1}), the Hamiltonian function of FTRL is
\begin{align}
H(X^t_1,y^t_1 ) = h^*_1(y^t_1) + h^*_2( y^0_2 + A^{(21)}X^t_1 ),
\end{align}
and the iterate map $\phi : (y_1^n, X^1_n) \to (y_1^{n+1}, X^{n+1}_1) $ of simultaneous FTRL can be written as Euler discretization of continuous FTRL, i.e., 
\begin{align}
y^{t+1}_1 &= y^t_1 - \eta \cdot \frac{\partial H}{\partial X_1}( X^t_1,y^t_1  )\\
X^{t+1}_1&= X^t_1 + \eta \cdot \frac{\partial H}{\partial y_1}( X^t_1,y^t_1  ).
\end{align}

Recall from (\ref{jacphi}), the jacobian map of  $\phi$ is
\begin{align}
J_{\phi} (X,y)= 
\begin{bmatrix}
I & - \eta (A^{(21)})^{\top} \cdot  \nabla^2 h^*_2(X,y) \cdot A^{(21)}     \\
\\
\\
\eta \nabla^2 h^*_1(X,y) & I
\end{bmatrix},
\end{align}

thus

\begin{align}
\det \left( J_{\phi}(X,y) \right) = \det \left(I + \eta^2 \cdot \nabla^2 h^*_1(X,y)  \cdot (A^{(21)})^{\top} \cdot  \nabla^2 h^*_2(X,y) \cdot A^{(21)} \right).
\end{align}

Since  $ \nabla^2 h^*_1 $ and  $(A^{(21)})^{\top} \cdot  \nabla^2 h^*_2 \cdot A^{(21)}$ are both symmetry and positive semidefinite matrix, by Lemma \ref{ABeig}, their product is diagonalizable and has non-negative eigenvalues. Thus we have
\begin{align}
\det \left(J_{\phi}(X,y) \right) \ge 1.
\end{align} 
Combine this with (\ref{formulaent}), we have
\begin{align}
S(X^{t+1}_1,y^{t+1}_1) & = S\left(  \phi (X^{t+1}_1,y^{t+1}_1 ) \right) \\
& = S(X^{t}_1,y^{t}_1) + \int_{\mathcal{X}} g^t(X^{t}_1,y^{t}_1) \log \left( \lvert \det J_{\phi} ( X^{t}_1,y^{t}_1) \lvert \right) dX^t_1dy^t_1 \label{detja}\\
&\ge S(X^{t}_1,y^{t}_1) + \int_{\mathcal{X}} g^t(X^{t}_1,y^{t}_1) \log ( 1 ) dX^t_1dy^t_1 \label{sffhneikrjnvfg}\\
& = S(X^{t}_1,y^{t}_1),
\end{align}
where $g^t (\cdot)$ is the probability density function of $(X^{t}_1,y^{t}_1)$, and (\ref{sffhneikrjnvfg}) comes from  $\det \left(J_{\phi}(X,y) \right) \ge 1$. Thus $S(X^{t}_1,y^{t}_1)$ is a non-decreasing function.

Moreover, as $\log(1+x) > x/(1+x)$ for $x > 0$, thus from (\ref{detja}), to prove a linear growth rate of differential entropy, it is sufficient to prove a uniform lower bound on $\det \left(J_{\phi}(X,y) \right) $. With the assumption that $h_i(\cdot)$ has Lipschitz continuous gradient, by Lemma  \ref{dual},  $h_i^*(\cdot)$ is $\mu_i$-strongly convex, thus 
$\nabla^2 h^*_i(X,y) $ is positive definite, i.e., 
\begin{align}
\nabla^2 h^*_i(X,y) \succcurlyeq \mu_i I
\end{align}
for some $\mu_i > 0$. From Lemma \ref{ABeig}, with $A = \nabla^2 h^*_1(X,y) $ and $B = (A^{(21)})^{\top} \cdot  \nabla^2 h^*_2(X,y) \cdot A^{(21)}  $, we have 
\begin{align}
\eta^2 \cdot \nabla^2 h^*_1(X,y)  \cdot (A^{(21)})^{\top} \cdot  \nabla^2 h^*_2(X,y) \cdot A^{(21)} 
\end{align}
is a diagonalizable matrix, and there are all eigenvalues of $J_{\phi}(X,y)$ are real number and larger than $1$. Moreover, these eigenvalues has a uniform lower bound $1+c$, where $c$ is determined
by the strongly convex coefficients $\mu_i$ and the payoff matrix $A^{(21)}$.
\end{proof}

\subsection{Proof of Proposition \ref{entropya}}\label{Proof of Proposition entropya}

\begin{lem} \label{injinj}The update map $$(X_1^{t},y_1^{t}) \to  (X_1^{t+1},y_1^{t+1}) $$ from (\ref{agent 1 Symplectic discretize equation}) is an injective map.
\end{lem}

\begin{proof}Recall the update rule can be written as
\begin{align}
&y^{t+1}_1 =y^{t}_1 - \eta \frac{\partial H }{\partial X_1} (X^t_1,y^t_1)=y^{t}_1 + \eta  A^{(12)} \nabla h^*_2 (y^0_2 + A^{(21)} X^{t}_1) , \label{tt11}\\
&X^{t+1}_1 =X^{t}_1 + \eta  \frac{\partial H}{\partial y_1} (X^t_1,y^{t+1}_1)=X^{t}_1 + \eta \nabla h^*_1(y^{t+1}_1). \label{tt22}\\
\end{align}
Thus given  $(X_1^{t+1},y_1^{t+1}) $, we can directly find an unique $X^{t}_1 $ from (\ref{tt22}), i.e., $$X^{t}_1= X^{t+1}_1 - \eta \nabla h^*_1(y^{t+1}_1) .$$ Then use this $X^{t}_1$, we can also determine a unique $y^t_1$ from (\ref{tt11}), i.e., $$y^t_1 = y^{t+1}_1 - \eta  A^{(12)} \nabla h^*_2 (y^0_2 + A^{(21)} X^{t}_1).$$ This finish the proof.
\end{proof}
Now we are ready to prove  Proposition \ref{entropya}.
\begin{proof} Since the iterate map $\psi : (X_1^n,y_1^n) \to (X_1^{n+1},y_1^{n+1})$ in Symplectic discretization of FTRL defined \ref{Symplectic discretization of FTRL} in is naturally an injective map from Lemma \ref{injinj} , we can directly use lemma \ref{entropyevo}. We have
\begin{align}\label{fasdfa}
S(X_1^{n+1},y_1^{n+1}) - S(X_1^{n},y_1^{n})  = \int_{\mathcal{X}} g^n(X_1^n,y_1^n) \log \left( \lvert \det J_{\phi}(X_1^n,y_1^n) \lvert \right) dx
\end{align}
where $g^n( \cdot )$ is the probability density function of random vector $(X_1^n,y_1^n)$. 

Moreover, since $\psi$ is a symplectomorphism, we have
$
 \lvert \det J_{\psi}(X_1^n,y_1^n) \lvert = 1.
$
Thus  the right hand side of (\ref{fasdfa}) equals to $0$, and this implies $S(X_1^{n+1},y_1^{n+1}) = S(X_1^{n},y_1^{n}) $.
\end{proof}

\subsection{Numerical examples of Proposition Propositions \ref{entropys} and \ref{entropya}}

Although differential entropy plays important roles in several subjects, estimating the value of differential entropy under transformations is generally a challenging task. Even in the one-dimensional case, special methods need to be designed for calculating differential entropy \citep{hyvarinen1997new}. A recent review of this topic can be found in \citep{feutrill2021review}.

However, for the case of gradient descent, it is possible to calculate the variation of differential entropy due to the linear structure of the algorithm and the equality
\begin{align}
    S(AX) - S(X) = \log( \lvert \det(A) \lvert ).
\end{align}

In Figure 2, we present numerical experiments on the variation of differential entropy using game defined by
$A = [[1,-1],[-1,1]]$. Numerical results show differential entropy has a linear growth rate in simultaneous case and keeps invariant in alternating case, which support Propositions \ref{entropys} and \ref{entropya}.

\begin{figure}[H] 
\centering 
\includegraphics[width=0.7\textwidth]{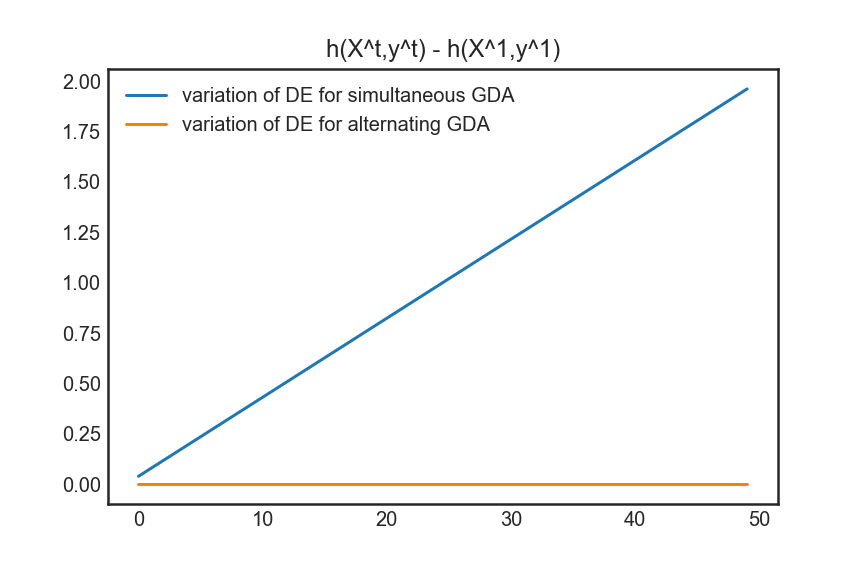} 
\caption{Variation of differential entropy for simultaneous and alternating GDA. The growth rate of the variation in DE is linear for simultaneous GDA, while it is $0$ for alternating GDA.} 
\label{f2} 
\end{figure}

\section{Proof of Theorem \ref{Covariance Evolution in FTRL with Euclidean Regularizer}}\label{proofer}

This appendix is divided into two parts. In Section \ref{Additional Backgrounds}, we presented necessary backgrounds from linear algebra, differential equation, and difference equation. In Section
\ref{ajfhaeblnljfrv}, we provide detailed prove of Theorem \ref{Covariance Evolution in FTRL with Euclidean Regularizer}.

\subsection{Additional Backgrounds }\label{Additional Backgrounds}
\subsubsection{Complex Jordan normal form}

We will consider the complex Jordan normal form of a real square matrix $A \in \BR^{d \times d}$. Let  $\Spec(A)$ be the set of eigenvalues of $A $. Consider $A$ acts on vector space $\BC^d$ as a linear operator. 

\begin{defn}[\textbf{Generalized Eigenvector}] \label{GE}
A vector  $v_m \in \BC^d$  is called a generalized eigenvector of \textbf{type} $m$ corresponding to the  eigenvalue $\mu$ if $$(A - \lambda I)^m v_m = 0$$ but  $$(A - \mu I)^{m-1} v_m \ne 0.$$
\end{defn}

\begin{defn}[\textbf{Jordan Chain}] \label{Jordanchian}
let $v_m$ be a generalized eigenvector of type m corresponding to the matrix $A$ and the eigenvalue $\mu$. The Jordan chain generated 
by $v_m$ is a set of m vectors $\{v_m, v_{m-1}, ..., v_1\}$ given by 
$$\begin{aligned}
v_{m-1} &= (A - \mu I)v_{m} \\
v_{m-2} &= (A - \mu I)^2v_{m}  =(A - \mu I)v_{m-1} \\
& ... \\
v_{1} &= (A - \mu I)^{m-1}v_{m} =(A - \mu I)v_{2}  \\
\end{aligned}$$
\end{defn}

\begin{rem}
If $\mu\in \BR$ is a real  eigenvalue of $A$, then the generalized eigenvectors of $\mu$ are also vectors over real numbers, and the Jordan chain are also made up by vectors over real numbers.
\end{rem}

\begin{prop}\label{Jordanchianlinear} A Jordan chain is a linearly independent set of vectors.
\end{prop}
\begin{proof} Let $\{v_m,v_{m-1},...,v_{1}\}$ be a Jordan chain generated by a type $m$ generalized eigenvector $v_m$ corresponding to
an eigenvalue $\lambda$ of $A$, and consider the equation
\begin{align}\label{ld}
c_m v_m + c_{m-1}v_{m-1} + ... + c_1v_1 = 0 
\end{align}
We will show $c_m = c_{m-1} = ... = c_1 = 0$. 

Multiply \eqref{ld} by $(A - \mu I)^{m-1}$, and note that for $j \le m-1$

$$\begin{aligned}
(A - \mu I)^{m-1} c_j v_j & = c_j (A - \mu I)^{m-j-1}(A - \mu I)^j x_j \\
& = 0
\end{aligned}$$

Thus \eqref{ld} becomes to be $c_m (A - \mu I)^{m-1}v_m = 0$. However, since $v_m$ is a type $m$ generalized eigenvector, we have
$$(A - \mu I)^{m-1}v_m  \ne 0,$$ 
thus  $c_m = 0$. Continuing this process, we will finally obtain 
$c_m = c_{m-1} = ... = c_1 = 0$. 
\end{proof}

\begin{prop}[page 366 of \citep{bronson1991matrix} ] Every $d \times d$ matrix has $d$ linearly independent generalized eigenvectors. 
\end{prop}

Given a Jordan chain $\{v_1, v_2, ... , v_m\}$ of length $m$, by Proposition \ref{Jordanchianlinear}, we will get a subspace spans by $\{v_1, v_2, ... , v_m\}$. The linear operator $A : \BC^d \to \BC^d$ can acts on vectors' set  $\{v_1, v_2, ... , v_m\}$. 

Denote $ [v_1, v_2 , ...., v_m]$ to be the matrix consists of column vectors $\{v_1, v_2, ... , v_m\}$, then $A$ acts on $ [v_1, v_2 , ...., v_m]$ as

$$\begin{aligned}
A [v_1, v_2 , ...., v_m] & = [Av_1 ,Av_2, ..., Av_m]\label{rb1} \\
\\
& = [\mu v_1 , v_1 + \mu v_2 ,..., v_{m-1} + \mu v_m] \\
\\
& = [v_1,v_2 , ..., v_m] 
\begin{bmatrix}
\mu& 1 & \\
&\mu&1 & \\
& & \ddots& \ddots \\
& & & \mu & 1\\
& & & & \mu\\
\end{bmatrix}.\label{rb2}
\end{aligned}$$

Thus we have

\begin{align}\label{jjj}
 [v_1, v_2 , ...., v_m]^{-1} A [v_1, v_2 , ...., v_m]
=
\begin{bmatrix}
\mu& 1 & \\
&\mu&1 & \\
& & \ddots& \ddots \\
& & & \mu & 1\\
& & & & \mu\\
\end{bmatrix}
\end{align}

which is a $m \times m$ upper triangular matrix, with eigenvalue $\mu$ on diagonal and  each non-zero off-diagonal entry equal to 1. Such an 
 upper triangular matrix is called a size m \textbf{Jordan block} of $A$ corresponding to eigenvalue $\mu$. 
  
 \begin{prop}[Page 367 of \citep{bronson1991matrix}]\label{cbs}
 Every $d \times d$ matrix $A$ has a set of $d \times d$ linearly independent generalized eigenvectors composed entirely of Jordan chains,  such a set of generalized eigenvectors  is a called a \textbf{canonical bases} of $A$.
 \end{prop}
 
\begin{defn}[\textbf{Generalized modal matrix}] \label{gmm}
Let $A$ be an $d \times d$ matrix. A generalized modal matrix $M$ for $A$ 
is a $d \times d$ matrix whose columns, considered as vectors, form a canonical basis for $A$ and appear in $M$ according to the 
following rules:
\begin{itemize}
\item[(1)] All vectors of the same chain appear together in adjacent columns of $M$.
 \item[(2)] Each chain appears in $M$ in order of increasing type.
\end{itemize}
\end{defn}

\begin{rem}
The Jordan chain corresponding to a real eigenvalue $\mu \in \BR$ of $A$ will composed by vectors in $\BR^d$, but if $\mu \in \BC$ is a complex eigenvalue of $A$, then the Jordan chain corresponding to $\mu$ will contain
vectors in $\BC^d$.
\end{rem}

Combine \eqref{jjj} and Proposition \ref{cbs}, we have following proposition.
\begin{prop}\label{CJF}Any matrix $A \in \BR^{d \times d}$ is similar to a matrix in Jordan normal form under the similarity transformation of a generalized modal matrix $M$ of $A$, i.e., 
$$J^{\BC} = M^{-1}AM$$
has block diagonal form $J^{\BC} = \oplus_{i} J_i  $ with Jordan blocks $J_i$ given with $\mu \in \Spec(A)$  by
$$J_{i} = \begin{bmatrix}
\mu & 1 & \\
&\mu&1 & \\
& & \ddots& \ddots \\
& & & \mu & 1\\
& & & & \mu\\
\end{bmatrix}.$$

The Jordan normal form is unique up to the order of the Jordan blocks.
\end{prop}

We have the following proposition that determines the number of a particular size of Jordan blocks in Jordan normal form corresponding to an eigenvalue $\lambda$ .  

\begin{prop}[Page 368 of \citep{bronson1991matrix}]\label{ker}Let $\lambda$ be an eigenvalue of $A$, and denote
$$ m = \max\{ i | \ \ker (A - \mu I)^i  \varsupsetneqq \ker (A - \mu I)^{i-1}  \}. $$

Denote $\rho_k$ as the number of linear independent generalized eigenvectors of type $k$ corresponding to the eigenvalue $\mu$
that appear in a canonical basis for $A$, then

$$ \rho_k = \dim \ker(A - \mu I)^k - \dim \ker(A - \mu I)^{k-1}\  (k = 1,2,...,m). $$

Since every Jordan chains of length $k$ in a canonical basis gives a size $k$ Jordan block, $ \rho_k $ is also the number of size $k$ Jordan blocks in the complex Jordan normal form corresponding to $\lambda$.
\end{prop}

There is another characterization on the size of the largest Jordan block corresponding eigenvalue based on the minimal polynomial of $A$: 

\begin{prop}[Theorem 3.3.6 in \citep{horn2012matrix}]\label{minipoly}  Let $m$ be the size of the largest Jordan block corresponding to the eigenvalue $\mu$ of a matrix $A$, then $m$ equals 
to the degree of the factor $(x - \mu)$ in the minimal polynomial of $A$.
\end{prop}

\subsubsection{Real Jordan normal form}
The real Jordan normal is important for computing exponential function of a matrix $A \in \BR^{d \times d}$. In the complex Jordan form, Jordan blocks may contain elements in $\BC- \BR$. Thus for a matrix $A \in \BR^{d \times d}$,  we cannot use its complex Jordan normal form to calculate exponential functions of $A$ directly. We need to define a standard form of $A$, which should only contain real numbers and keep the shape as a diagonal block matrix. This motive the definition of real Jordan normal form.

Let $\mu = \Re(\mu)+ i \Im(\mu), \Im(\mu) \ne 0 $ be an eigenvalue of $A$, and $v_m \in \BC^d$ is a generalized eigenvalue of type $m$ corresponding to $\mu$. Then the complex conjugate of $\mu$, denote by $\bar{\mu}$, is also an eigenvalue of $A$, and the complex conjugate of $v_m$, denote by $\bar{v}_m \in \BC^d$ is a generalized eigenvalue of type $m$ corresponding to $\bar{\mu}$. Let 
$$\{ v_{m-i} |\  v_{m-i} = (A - \mu I)^i v_m,\  i = 0,1,...,m-1   \}$$
be a Jordan chain of length $m$ corresponding to $\lambda$, then it gives a complex Jordan block of size $m$ in the complex Jordan normal form of $A$ as in \eqref{jjj}.  

The $2m$ vectors $\{ \Re(v_1),\Im(v_1), \Re(v_2), \Im(v_2), ... , \Re(v_m), \Im(v_m) \} \subset \BR^d$ will play the role of  complex generalized vectors of eigenvalues $\mu, \bar{\mu}$. It is directly to check $A$ acts on $[\Re(v_1),\Im(v_1), \Re(v_2), \Im(v_2), ... , \Re(v_m), \Im(v_m)]$ gives the following
matrix representation :

$$\begin{aligned}
&A[\Re(v_1),\Im(v_1), \Re(v_2), \Im(v_2), ... , \Re(v_m), \Im(v_m)]\label{cb1} \\
\\
& = [A\Re(v_1),A\Im(v_1), A\Re(v_2), A\Im(v_2), ... , A\Re(v_m), A\Im(v_m)]\\
\\
& = [\Re(\mu v_1) , \Im(\mu v_1), \Re(v_1) + \Re(\mu v_2), \Im(v_1)+\Im(\mu v_2), ... , \Re(v_{m-1}) + \Re(\mu v_m), \Im(v_{m-1})+\Im(\mu v_m) ]   \\
\\
& = [\Re(v_1),\Im(v_1), \Re(v_2), \Im(v_2), ... , \Re(v_m), \Im(v_m)] 
\begin{bmatrix}
D & I \\
   &  \cdot & \cdot \\
   & & \cdot & \cdot & \\
   & & & \cdot & \cdot & \\
   & & & & \cdot & I \\
      & & & &  & D \\
 \end{bmatrix}\label{cb2}
\end{aligned}$$
where $D = 
\begin{bmatrix}
\Re(\mu) & \Im(\mu) \\
\\
-\Im(\mu)& \Re(\mu)
\end{bmatrix}$
and
$I_2 = 
\begin{bmatrix}
1 & 0 \\
0 &1
\end{bmatrix}.$

The matrix

\begin{align}\label{rjbi}
\begin{bmatrix}
D & I \\
   &  \cdot & \cdot \\
   & & \cdot & \cdot & \\
   & & & \cdot & \cdot & \\
   & & & & \cdot & I \\
      & & & &  & D \\
 \end{bmatrix}
 \end{align}

is called the real Jordan blocks corresponding to a conjugate pair of image eigenvalues $\mu, \bar{\mu}$. 

As in the case of complex Jordan normal form, we need to  define the real generalized modal matrix, and under the similar transformation by real generalized modal matrix, the original matrix will be transformed into a real Jordan normal form.

\begin{defn}[\textbf{Real generalized modal matrix}] \label{rgmm}
Let $A$ be an $d \times d$ matrix. A real generalized modal matrix $M$ for $A$ 
is a $d \times d$ matrix whose columns, considered as vectors, are real or image parts of a complex generalized eigenvectors in a canonical basis for $A$ and appear in $M$ according to the 
following rules:
\begin{itemize}
\item[(1)] Real and image parts of a generalized eigenvectors corresponding to same image eigenvalues appear in the first columns of $M$
\item[(2)] If $\Re(v), \Im(v)$ appear in the real generalized modal matrix, $(\Re(v), \Im(v))$ appear in adjacent columns of $M$
\item[(3)] All vectors of the same chain appear together in adjacent columns of $M$.
 \item[(4)] Each chain appears in $M$ in order of increasing type.
\end{itemize}
\end{defn}

Note that comparing to definition \ref{gmm}, the new requirements are $(1)$ and $(2)$. $(1)$ will make the Jordan blocks as in \eqref{rjbi} appear firstly in the real Jordan normal form, and the necessary of $(2)$ can be seen from the derivation of \eqref{rjbi}.

\begin{prop}[ Theorem 1.2.3 in \citep{colonius2014dynamical}]\label{tr} Any matrix $A \in \BR^{d \times d}$ is similar to a matrix in the real Jordan normal from via a similarity transformation by $A$'s real generalized modal matrix. That is, $\exists M \in \BR^{d \times d}$ be $A$'s real generalized modal matrix to make

$$J^{\BR} = M^{-1} A M$$

where $J^{\BR} = (J_{\mu_{1}, \bar{\mu}_{1}} \oplus ... \oplus J_{\mu_{k}, \bar{\mu}_{k}} ) \oplus^l_{i=1}(J_{\mu_{k+i}}) $ with real Jordan blocks given for $\mu\in \Spec(\CL) \cap \BR$ by
$$J_{\mu} = \begin{bmatrix}
\mu & 1 & \\
&\mu&1 & \\
& & \ddots& \ddots \\
& & & \mu& 1\\
& & & & \mu\\
\end{bmatrix}$$
and for $\mu, \bar{\mu} = \lambda \pm  i \nu \in \Spec(\CL), \nu > 0 $, by 
$$J_{\mu, \bar{\mu}} = \begin{bmatrix}
D & I_2 & \cdot & \cdot & 0\\
0 & D & & & . \\
\cdot & \ddots & \ddots & .\\
\cdot & & & D & I_2\\
0 & & & 0 & D
\end{bmatrix}$$
where $D = 
\begin{bmatrix}
\lambda& - \nu \\
\nu & \lambda
\end{bmatrix}$
and
$I_2 = 
\begin{bmatrix}
1 & 0 \\
0 &1
\end{bmatrix}$.
\end{prop}

Note that the difference between complex and real Jordan form, complex and real generalized modal matrix are only in the Jordan blocks corresponding to a eigenvalue whose image part
is not $0$. Thus the number of a given size Jordan blocks corresponding to a real eigenvalue is same in complex and real Jordan form. For a pair of conjugate complex eigenvalues with nonzero
image part, every real Jordan block is a combination of two complex Jordan blocks.

\subsubsection{Solution formula for linear differential equation}

For a linear differential equation with constant coefficients $A \in \BR^{d \times d}$
$$\begin{aligned}
\frac{dx}{dt} (t) = A x
\end{aligned}$$
with initial condition $x(0)$, it's solution formula is given by $x(t) = e^{t A} x(0)$. Thus to understand the dynamic behavior of $x(t)$, it is necessary to calculate the matrix exponential matrix $e^{t A}$. This can be done by using the real Jordan normal form of $A$. Let $J_{A}$ be $A$'s real Jordan normal form, and $J_{A} =M A M^{-1} $, then $e^{t A} = M e^{t J_{A}} M^{-1}$. Thus calculation of $e^{tA}$ can be reduced to calculation of $e^{t J_{A}}$ and the real generalized modal matrix of $A$. The real generalized modal matrix of $A$ is a combination of real and image parts of $A$'s generalized eigenvectors, and in this section we consider calculation of $e^{t J_{A}}$. 

\begin{prop} (page 12 in \citep{colonius2014dynamical})\label{ExpJordan} Let $J_{\mu}$ be a real Jordan block of size $m \times m$ associated with the real eigenvalue $\mu$ of a matrix $A \in \BR^{d \times d}$. Then 
\begin{align}
J_{\mu} = \begin{bmatrix}
\mu & 1 & \\
&\mu&1 & \\
& & \ddots& \ddots \\
& & & \mu & 1\\
& & & & \mu\\
\end{bmatrix} \in \BR^{m \times m}
\end{align}
and 
\begin{align}\label{expr}
e^{J_{\mu}t} = e^{\mu t} \begin{bmatrix}
1 & t & \frac{t^2}{2!} & \cdot &  \cdot & \frac{t^{m-1}}{(m-1)!} \\
 &\cdot &\cdot & \cdot & \cdot & \cdot  \\
& & \cdot & \cdot &  \cdot & \cdot \\
& &  & \cdot &\cdot & \frac{t^2}{2!} \\
& & & &  \cdot& t \\
& & & & & 1
\end{bmatrix} \in \BR^{m \times m}.
\end{align}

 Let $J_{\mu,\bar{\mu}}$ be a real Jordan block of size $2m \times 2m$ associated with the real eigenvalue $\mu, \bar{\mu} = \lambda \pm i \nu, \nu > 0$ of a matrix $A \in \BR^{d \times d}$. With 

 $$D = 
\begin{bmatrix}
\lambda& - \nu \\
\nu & \lambda
\end{bmatrix},\ R : = R(t) = 
  \begin{bmatrix}
\cos(\nu t)& - \sin(\nu t) \\
\\
 \sin(\nu t) & \cos(\nu t)
 \end{bmatrix},\ 
 I_2 = \begin{bmatrix}
1& 0\\
0 & 1
 \end{bmatrix}, 
 $$
one obtains for 
\begin{align}
J_{\mu,\bar{\mu}} = 
\begin{bmatrix}
D & I_2 \\
   &  \cdot & \cdot \\
   & & \cdot & \cdot & \\
   & & & \cdot & \cdot & \\
   & & & & \cdot & I_2 \\
      & & & &  & D \\
 \end{bmatrix} \in \BR^{2m \times 2m},
 \end{align}
 and
\begin{align}\label{expc}
e^{J_{\mu,\bar{\mu}}t } = e^{\lambda t} 
 \begin{bmatrix}
R & tR & \frac{t^2}{2!}R & \cdot & \cdot & \frac{t^{m-1}}{(m-1)!} R\\
\\
 &\cdot &\cdot & \cdot &\cdot & \cdot  \\
& & \cdot & \cdot & \cdot & \cdot \\
& &  & \cdot & \cdot & \frac{t^2}{2!}R \\
& & & & \cdot & tR \\
& & & & & R
\end{bmatrix} \in \BR^{2m \times 2m}.
\end{align}
\end{prop}

\subsubsection{Solution formula for linear difference equation}

For a matrix $A \in \BR^{d \times d}$, a linear difference equation with coefficient matrix $A$ has form
\begin{align}\label{lde}
x_{n+1} = A x_n
\end{align}
By induction, \eqref{lde} has solution formula
\begin{align}\label{lde}
x_{n} = A^n x_0,\  x_0 \in \BR^{d}.
\end{align}
If $J_A =M^{-1} A M$ be the real Jordan normal form of $A$, then $x_n = M (J_A)^n M^{-1}x_0.$ Thus to solve \eqref{lde}, it is sufficient to know the formula for $ (J_A )^n$ and the real generalized modal matrix $M$. Moreover, if
$
J_{A} = \oplus_{i} J_i, 
$
then
$
(J_A)^n = \oplus_{i} (J_i)^n, 
$
thus we only need to consider the power of real Jordan blocks.

\begin{prop}[Page 19 in \citep{colonius2014dynamical}] \label{differencereal}
Let $J$ be a real Jordan block of size $m \times m$
 associated with a real eigenvalue $\mu$ of $A \in \BR^{d \times d}$. Then
 
\begin{align}
J_{\mu} &= \begin{bmatrix}
\mu &  & \\
&\mu& & \\
& & \ddots&  \\
& & & \mu & \\
& & & & \mu \\
\end{bmatrix} 
+ 
\begin{bmatrix}
0 & 1 & \\
&0&1 & \\
& & \ddots& \ddots \\
& & & 0& 1\\
& & & & 0\\
\end{bmatrix} \\
& = \mu I +N
\end{align}
with $N^m = 0$. Thus we have

\begin{align}
J_{\mu}^n = (\mu I +N)^n = \sum^{m-1}_{i=0} \binom{n}{i} \mu^{n-i}N^i.
\end{align}

Note that if $\mu > 1$, elements in $J_{\mu}^n$ will have an exponential growth rate as $\mu^{n}$. If $\mu = 1$, 
elements in $J_{\mu}^n$  have a polynomial growth rate. If $\mu < 1$, elements in $J_{\mu}^n$ tends to $0$ as $n$ growth.
\end{prop}

\begin{prop}[Page 20 in \citep{colonius2014dynamical}] \label{differenceim}
Let $J$ be a real Jordan block of size $2m \times 2m$
 associated with a pair of conjugate complex eigenvalue $\mu = \lambda + i \nu, \bar{\mu} = \lambda - i \nu$ of $A \in \BR^{d \times d}$. With
 $D = \begin{bmatrix}
 \lambda & - \nu \\
 \nu & \lambda
 \end{bmatrix}$
 and 
 $I_2 = \begin{bmatrix}
1 & 0 \\
0 & 1
 \end{bmatrix}$, one obtains that 
\begin{align}
J_{\mu,\bar{\mu}}
& =\begin{bmatrix}
D & 0\\
   &  \cdot & \cdot \\
   & & \cdot & \cdot & \\
   & & & \cdot & \cdot & \\
   & & & & \cdot & 0 \\
      & & & &  & D \\
 \end{bmatrix}
+
\begin{bmatrix}
0& I_2 \\
   &  \cdot & \cdot \\
   & & \cdot & \cdot & \\
   & & & \cdot & \cdot & \\
   & & & & \cdot & I_2 \\
      & & & &  & 0 \\
 \end{bmatrix}\\
 & = \tilde{D} + N
\end{align}
with $N^{m} = 0$. Moreover, since $\mu =\lvert \mu \lvert e^{i \phi} $ for some $\phi \in [0, 2 \pi)$, one can write  \begin{align}D = \begin{bmatrix}
 \lambda & - \nu \\
 \nu & \lambda
 \end{bmatrix}= \lvert \mu \lvert R,\ R =  \begin{bmatrix}
 \cos \phi & - \sin \phi \\
 \sin \phi & \cos \phi
 \end{bmatrix}.
 \end{align}

Thus 
\begin{align}
J_{\mu,\bar{\mu}}^n = (\tilde{D} +N)^n & = \sum^{m-1}_{i=0} \binom{n}{i} \tilde{D}^{n-i}N^i = \sum^{m-1}_{i=0} \binom{n}{i}\lvert \mu \lvert^{n-i} \tilde{R}^{n-i}N^i.
\end{align}
where $ \tilde{R}$ is a block diagonal matrix with matrix block $R$. Note that if $\lvert \mu \lvert > 1$, elements in $J_{\mu,\bar{\mu}}^n$ have exponential growth rate as
$\CO(\lvert \mu \lvert^n)$. If $\lvert \mu \lvert = 1$, elements in $J_{\mu,\bar{\mu}}^n$ have polynomial growth rate.  If $\lvert \mu \lvert < 1$, elements in $J_{\mu,\bar{\mu}}^n$ tends to $0$ as $n$ growth. 
\end{prop}

\subsection{Proof of Theorem \ref{Covariance Evolution in FTRL with Euclidean Regularizer}}\label{ajfhaeblnljfrv}

Now we are ready to prove Theorem \ref{Covariance Evolution in FTRL with Euclidean Regularizer}. The proof is divided into three parts :
\begin{itemize}
\item covariance evolution of continuous time equation, proved in \ref{proof3}.
\item covariance evolution of Euler discretization, proved in \ref{proof4}.
\item covariance evolution of Symplectic discretization, proved in \ref{proof5}.
\end{itemize}

\subsubsection{Covariance evolution of continuous equation}\label{proof3}

We firstly give a  summary of the proof. The continuous time equation of FTRL with Euclidean norm for agent 1 is written as:

\begin{align}
\frac{dy_1}{dt}&=-\frac{\partial H(X_1,y_1)}{\partial X_1}=-AA^{\top}X_1(t)+Ay_2(0)\label{e11}
\\
\frac{dX_1}{dt}&=\frac{\partial H(X_1,y_1)}{\partial y_1}=y_1(t).\label{e12}
\end{align}
Similarly for agent 2 we have
\begin{align}
\frac{dy_2}{dt}&=-\frac{\partial H(X_2,y_2)}{\partial X_2}=-A^{\top}AX_2(t)-A^{\top}y_1(0)
\\
\frac{dX_2}{dt}&=\frac{\partial H(X_2,y_2)}{\partial y_2}=y_2(t).
\end{align}

In the following, we will focus on the viewpoint of agent 1, thus we will consider equation (\ref{e11}) and (\ref{e12}). 

\begin{lem}The solution of the linear differential system consisted by equation  (\ref{e11}) and (\ref{e12}) and initial condition    $(y_1(t_0),X_1(t_0))$   can be written as
\begin{align}\label{solution111}
\begin{bmatrix}
y_1(t+t_0) \\
\\
X_1(t+t_0)
\end{bmatrix}=
e^{\CL t} 
\cdot
\begin{bmatrix}
y_1(t_0) \\
\\
X_1(t_0)
\end{bmatrix}
+
e^{\CL t}  \cdot
\int^{t+t_0}_{t_0} e^{-\CL s}A y_2(0)ds 
\end{align}
\end{lem}

\begin{proof}
It is directly to verify the derivate of right hind side satisfies  equation  (\ref{e11}) and (\ref{e12}), and the solution satisfies initial condition $(y_1(t_0),X_1(t_0))$. Due to the uniqueness of the solution  of linear differential equation, we can conclude this is the solution of equation  (\ref{e11}) and (\ref{e12}).
\end{proof}

Form (\ref{solution111}) we can also see that if uncertainty are introduced to the system at some initial time $t_0$, i.e., $(y_1(t_0), X_1(t_0))$ is a random variable, then the evolution of covariance of random variable  $(y_1(t+t_0), X_1(t+t_0))$ will not be affected by the term $\int^{t+t_0}_{t_0} e^{-\CL s}A y_2(0)ds $ since this is a determined quantity, thus will only affect the expectation of $(y_1(t+t_0), X_1(t+t_0))$ . Therefore, without loss of generality, we will let $y_2(0) = 0$  in the following.

Thus the continuous equation of agent 1 is
\begin{align}\label{Maineq}
\begin{bmatrix}
\frac{dy_1}{dt} \\
\\
\frac{dX_1}{dt}
\end{bmatrix}
= \CL 
\begin{bmatrix}
y_1(t) \\
\\
X_1(t)
\end{bmatrix}
\end{align}
where 
$$\CL = 
\begin{bmatrix}
0 & -AA^{\top} \\
\\
I & 0
\end{bmatrix} \in \BR^{2n \times 2n}.
$$

Since the solution of  (\ref{Maineq}) is 
$
\begin{bmatrix}
y_1(t)\\
\\
X_1(t)
\end{bmatrix}
=
e^{t \CL}
\begin{bmatrix}
y_1(0)\\
\\
X_1(0)
\end{bmatrix},
$
we have
$$\begin{aligned}
P(t + t_0) = 
e^{t \CL} P(t_0) (e^{t \CL})^{\top}
\end{aligned}.$$

Thus to analysis the behavior of $\Var(X_1(t)), \Var(y_1(t))$, and $\Cov(X_1(t),y_1(t))$, we need to understand the behavior of $e^{t \CL}$. A standard method to calculate the matrix exponential of a matrix with elements in $\BR$ is though the matrix's real Jordan normal form and real generalized modal matrix, see Proposition \ref{ExpJordan}. For our purpose, the most important question about the Jordan form of $\CL$ is : 
$$
\mbox{What is the size of the largest Jordan blocks corresponding to}\ \CL's \ \mbox{eigenvalues ?}
$$
Because this number will determine the growth rate of elements in $e^{t \CL}$.  In Proposition \ref{mpL}, we will determine the minimal polynomial of $\CL$, combine this result with Proposition \ref{minipoly}, we can get elements in $e^{t \CL}$ that will at most have a linear growth rate. However as we see in Theorem \ref{Covariance Evolution in FTRL with Euclidean Regularizer} there is a difference between $\Var(X_1(t))$ and $\Var(y_1(t))$, $\Var(X_1(t))$ may have quadratic growth and $\Var(y_1(t))$ will always be bounded, only the growth rate of $e^{t \CL}$ is not sufficient for one to explain this difference. To show that $\Var(y_1(t))$ is always bounded, we need a more detailed analysis of the real generalized modal matrix of $\CL$. We will see that the first $n$ rows of the real generalized modal matrix has many $0$ elements, and these $0$ elements will make the first $n$ rows of $e^{t \CL}$ not to have linear growth rate. This will be shown in Proposition \ref{etl}.

\begin{lem}\label{eigL} Let $$\CL = 
\begin{bmatrix}
0 & -AA^{\top} \\
\\
I & 0
\end{bmatrix}
$$ be the coefficient matrix of (\ref{Maineq}), then the eigenvalues of $\CL$ are pure imaginary numbers or $0$. Moreover, if $0$ is an eigenvalue of $\CL$, then its multiplicity is an even number.
\end{lem}

\begin{proof} Let $f(x)$ be the character polynomial of $-AA^{\top}$, thus 
\begin{align}\label{cpoly}
f(x) = \det (x I_{n \times n} + AA^{\top})
\end{align}
Firstly, every eigenvalue of $-AA^{\top}$ is a negative number or $0$. That is because if $\mu$ is an eigenvalue of $-AA^{\top}$ and $v$ is an eigenvector of $\mu$, then
$$\begin{aligned}
\mu \langle v,  v\rangle & =  \langle -AA^{\top} v,  v\rangle \\
& = - \langle A^{\top} v, A^{\top} v\rangle \\
& = -\norm { A^{\top} v}^2 \le 0,
 \end{aligned}$$
since $ \langle v,  v\rangle \ge 0$, thus we have $\mu \le 0$. This implies the zeros of (\ref{cpoly}) are $0$ or negative.

The character polynomial of $\CL$ is

\begin{align}\label{charpoly}
\det (\lambda I_{2n \times 2n} - \CL ) & = \det (\lambda^2 I_{n \times n} + AA^{\top}) \\
& = f(\lambda^2)
\end{align}
Thus the eigenvalues of $\CL$ are square roots of  zeros of (\ref{cpoly}). Since the zeros of (\ref{cpoly}) are $0$ or negative, thus the eigenvalues of $\CL$ are $0$ or pure imaginary numbers. Moreover, since (\ref{charpoly}) is an even polynomial, it means that the polynomial only have terms of even degree, thus if $0$ is a zero of (\ref{charpoly}), then its multiplicity is at least 2. 
\end{proof}

Next we will calculate the minimal polynomial of $\CL$. Combining the calculated results with Proposition \ref{minipoly}, we will get the size of the biggest Jordan blocks in the Jordan normal form of $\CL$. Before that, we need the following lemma.

\begin{lem}(Corollary 3.3.10. in \citep{horn2012matrix})\label{minilinear}Let $M \in \BR^{d \times d}$ and $f_M(x)$ be its minimal polynomial. Then $M$ is diagonalizable if and only if every eigenvalue of M has multiplicity 1 as a root of $f_M(x) = 0$.
\end{lem}

\begin{prop}[Minimal polynomial of $\CL$]\label{mpL} Let $\lambda_1 \ne \lambda_2 \ne ... \ne \lambda_l$ be the distinct eigenvalues of $-AA^{\top}$, thus $\forall i \in [l], \lambda_i \in \BR,\  \lambda_i \le 0$. Then the minimal polynomial of $\CL$ is :
$$f_{\CL}(x) =(x - \sqrt{\lambda_1} i)(x + \sqrt{\lambda_1} i) ... (x - \sqrt{\lambda_l} i)(x + \sqrt{\lambda_l} i) $$
Note that this implies that eigenvalues of $\CL$ are purely imaginary or  $0$, moreover, if $ \lambda_i = 0$ for some $i$, then the factor $x$ in $f_{\CL}(x)$ has degree $2$. 
\end{prop}
\begin{proof}
Let $f_{-AA^{\top}} (x)$ be the minimal polynomial of $-AA^{\top}$. Since $-AA^{\top}$ is symmetric, thus diagonalizable, by Lemma \ref{minilinear}, $f_{-AA^{\top}} (x)$ only contains linear factor of $(x - \lambda_i)$,where $\lambda_i$ is an eigenvalue of $-AA^{\top}$. Thus we have
$$f_{-AA^{\top}} (x) =  \prod _{i} (x - \lambda_i) $$

We claim $f_{-AA^{\top}}(\CL^2) = 0$, since:
$$ \CL^2 = \begin{bmatrix}
-AA^{\top} &0\\
\\
\\
 0 & -AA^{\top}
\end{bmatrix},$$
therefore
$$f_{-AA^{\top}}(\CL^2) = 
 \begin{bmatrix}
 f_{-AA^{\top}} (-AA^{\top}) &0\\
 \\
 \\
 0 &  f_{-AA^{\top}} (-AA^{\top})
 \end{bmatrix}
= 0.
$$

Thus if $f_{\CL}(x)$ is the minimal polynomial of $\CL$, we have 
\begin{align}\label{divide}
f_{\CL}(x) \  | \  f_{-AA^{\top}}(x^2) =  \prod _{j} (x - \sqrt{\lambda_j} i) (x + \sqrt{\lambda_j} i) 
\end{align}
Moreover, since every eigenvalue of $\CL$ is also a root of $ f_{\CL}(x) $, from (\ref{divide}), we have : 
\begin{itemize}
\item[(1)] If $\pm \sqrt{\lambda_j} i \ne 0$ is an eigenvalue of $\CL$, then the degree of $(x \pm \sqrt{\lambda_j} i)$ in $f_{\CL}(x)$ must be 1.
\item[(2)] If $\lambda_j = 0$ is an eigenvalue of $\CL$, then the degree of $x$ in $f_{\CL}(x)$ must be $1$ or $2$.
\end{itemize}

In the following arguments, we will prove that there is at least a real Jordan block corresponding to eigenvalue $0$ has size $2$. We have
$$
 \CL = \begin{bmatrix}
0 & & -AA^{\top} \\
\\
I & & 0
\end{bmatrix} 
,\  \CL^2 = \begin{bmatrix}
-AA^{\top} & &0\\
\\
 0 & & -AA^{\top}
\end{bmatrix}.
$$
Thus for some $x = (x_1,x_2) \in \Ker(\CL^2)$,where $ x_1,x_2 \in \BR^n$. Then we have
$$\begin{aligned}
\CL^2 \cdot
\begin{bmatrix}
x_1\\
\\
x_2
\end{bmatrix}
& = 
\begin{bmatrix}
-AA^{\top} \cdot x_1\\
\\
-AA^{\top} \cdot x_2
\end{bmatrix}
= 0 .
\end{aligned}$$
This implies $x_1,x_2 \in \Ker -AA^{\top}$. Since $-AA^{\top}$ has eigenvalue 0, $x_1,x_2$ may not equal to $0$. But

$$\begin{aligned}
\CL \cdot
\begin{bmatrix}
x_1\\
\\
x_2
\end{bmatrix}
& = 
\begin{bmatrix}
-AA^{\top} \cdot x_2\\
\\
 x_1
\end{bmatrix}
\end{aligned},$$
so if $(x_1,x_2) \in \Ker \CL$, $x_1$ must be $0$. This completes the proof of $\Ker(\CL) \ne \Ker(\CL^2)$. By Proposition \ref{ker}, there exists Jordan block of size $2$. 
\end{proof}

\begin{lem}\label{ljbJ} Let $J_0$ be the largest Jordan blocks of $\CL$ corresponding to $0$ eigenvalues, then elements in $e^{J_0 t}$ are of $\CO(t)$. Let $J_{\mu,\bar{\mu}}$ be the largest Jordan blocks of $\CL$ corresponding to  a pair of conjugate purely imaginary eigenvalues $(\mu,\bar{\mu})$, then elements in $e^{J_{\mu,\bar{\mu}} t} \in \CO(1)$.
\end{lem}

\begin{proof}From Proposition \ref{mpL}, the size of the largest Jordan blocks of $\CL$ corresponding to the $0$ eigenvalues are $2$. Thus the corollary follows from  (\ref{expr}) with $\mu = 0$ and $m=2$. From Proposition \ref{mpL}, the size of the largest Jordan blocks of $\CL$ corresponding to a pair of conjugate imaginary eigenvalues $(\mu,\bar{\mu})$ are $2$. So the corollary follows from  (\ref{expc}) with $\lambda = 0$ and $m=1$.
\end{proof}

Next we consider the set of real generalized vectors of $\CL$. Let $S_{\CL}$ be the set of vectors that appear in the real generalized  modal matrix $M_{\CL}$ . Then as shown in Proposition \ref{mpL}, these generalized eigenvectors corresponds to an imaginary eigenvalues or $0$ eigenvalues. If $ v \in S_{\CL}$ corresponds to $0$, $v$ may in a Jordan chain with length $1$ or length $2$. If $ v \in S_{\CL}$ corresponds to a purely imaginary number, then there exists some $w$ be the eigenvector of an imaginary eigenvalue and $v = \Re(w)$ or $v = \Im(w)$.
 
Thus we have
 
$$ S_{\CL} = (\cup_{(\mu, \bar{\mu})}  S_{\CL}(\mu, \bar{\mu}) )\cup   S_{\CL}(0) $$

where $$S_{\CL}(\mu, \bar{\mu})  = \{v \in \BR^{2n} |\ \exists \  w \ \mbox{be an eigenvector for}\  \mu\ \mbox{or}\ \bar{\mu} \in \BC - \BR \ , v = \Re(w) \ \mbox{or}\ v = \Im(w) \} ,$$ and
$$S_{\CL}(0)  = \{v \in \BR^{2n} |\ v \mbox{ is a generalized eigenvector for}\ 0 \} .$$
  
So as in definition \ref{rgmm},  $M_{\CL}$ has the following form:
 
\begin{align}
M_{\CL} =
\begin{bmatrix} 
 \coolover{\mbox{$m_1$ columns}}{v_1, & ... & ,v_{m_1} } ,&    \coolover{\mbox{$m_2$ columns}}{v_{1+m_1}, & ... &, v_{m_1 + m_2} }, &  \coolover{\mbox{$m_3$ columns}}{v_{m_1 + m_2 + 1} ,& ... & ,v_{m_1 +m_2 + m_3} } 
\end{bmatrix}
\end{align}
 where 
\begin{itemize}
\item $\{v_1, ...., v_{m_1} \} \subset \cup_{(\lambda, \bar{\lambda})}  S_{\CL}(\mu, \bar{\mu}) $ .
\item $\{v_{1+m_1} , ... , v_{m_1 + m_2} \} \subset  S_{\CL}(0) $, and each each $v_i$ is a Jordan chain of length $1$.
\item $\{ v_{m_1 + m_2+1} , ... , v_{m_1 +m_2 + m_3} \} \subset  S_{\CL}(0)$, $(v_i, v_{i+1})$ is a Jordan chain of length $2$, and $v_i = \CL v_{i+1}$.
\item $m_1 + m_2 + m_3 = 2n.$
\end{itemize}
   
\begin{lem}\label{firstmrow}
If $ w = (w_1 , w_2 ,..., w_{2n}) ^{\top} \ne 0$ is a type 1 generalized eigenvectors of $\CL$ corresponding to $0$ eigenvalue , then the first n components $(w_1 , w_2 , ... , w_{n} )^{\top}= 0$, and the the last n components $(w_{n+1}, ..., w_{2n})^{\top} \in \ker (-AA^{\top})$ . 

If $ s = (s_1 , s_2 ,..., s_{2n})^{\top} \ne 0$ is a type 2 generalized eigenvectors of $\CL$ corresponding to $0$ eigenvalue , then $(s_1, ..., s_n)^{\top} \ne 0$ and 
 $(s_1, ..., s_n)^{\top}, (s_{n+1}, ... , s_{2n})^{\top} \in \ker(-AA^{\top})$. 
\end{lem}
 
 \begin{proof} Firstly, note that  type 1 generalized eigenvectors of $\CL$ corresponding to $0$  are just vectors in $\ker(\CL)$. Thus if $w$ is a type 1 generalized eigenvectors of $\CL$ corresponding to $0$, we have
 
 $$\begin{aligned}
 \CL \cdot w & = 
 \begin{bmatrix}
 0 & -AA^{\top} \\
 \\
 I & 0
 \end{bmatrix} \cdot
 (w_1 , w_2 ,..., w_{2n}) ^{\top} \\
 & = 
  \begin{bmatrix}
 -AA^{\top} \cdot (w_{n+1}, ..., w_{2n})^{\top} \\
 \\
 \\
(w_1 ,..., w_{n})^{\top}
 \end{bmatrix} 
 = 0
 \end{aligned}$$
 Thus $(w_1 , w_2 , ... , w_{n} )^{\top}= 0$ and  $(w_{n+1}, ..., w_{2n})^{\top} \in \ker (-AA^{\top})$ . 
 
 Secondly, if $ s = (s_1 , s_2 ,..., s_{2n})^{\top} \ne 0$ is a type 2 generalized eigenvectors of $\CL$ corresponding to $0$ eigenvalue, then we have
 $$\CL \cdot s \ne 0\  \  \mbox{and} \ \ \CL^2 \cdot s = 0.$$
 
Since 
 $$\begin{aligned}
 \CL^2 \cdot s & = 
 \begin{bmatrix}
 -AA^{\top} & 0 \\
 \\
 0 & -AA^{\top}
 \end{bmatrix} \cdot
 (s_1 , s_2 ,..., s_{2n})^{\top} \\
 & = 
  \begin{bmatrix}
 -AA^{\top} \cdot (s_{n+1}, ..., s_{2n})^{\top} \\
 \\
 \\
 -AA^{\top} (s_1 ,..., s_{n})^{\top}
 \end{bmatrix} 
  = 0
 \end{aligned}$$
This implies $(s_1 ,..., s_{n})^{\top}, (s_{n+1}, ..., s_{2n})^{\top} \in \ker(-AA^{\top})$. Moreover, since $s$ is a type $2$ generalized eigenvector, we have $(s_1 ,..., s_{n})^{\top} \ne 0 $.
 \end{proof}

 From the correspondence between real generalized vectors and real Jordan blocks as described in Proposition \ref{tr}, the real Jordan form $J_{\CL}$ of $\CL$ under a similar transformation by real generalized modal matrix $M_{\CL}$ has form

 \begin{align}\label{J_CL}
J_{\CL}=\begin{bmatrix} 
 \coolover{\mbox{$m_1$ columns}}{J_{\mu_1, \bar{\mu}_1}& & & & & & }  \\ 
 & \ddots & & \\
& & J_{\mu_{m_1/2}, \bar{\mu}_{m_1/2}}  &\\
& & &\coolover{\mbox{$m_2$ columns}}{0 & & & &  }\\
& & & & \ddots \\
& & & & & 0 \\
& & & & & &\coolover{\mbox{$m_3$ columns}}{ J^2_{0} & & &}\\
& & & & & & & \ddots \\
& & & & & & & &  J^2_{0}\\
\end{bmatrix}
\end{align}

$J_{\mu_k, \bar{\mu}_k} = 
\begin{bmatrix}
\Re(\mu_k) & - \Im(\mu_k)\\
\\
\Im(\mu_k) & \Re(\mu_k)
\end{bmatrix} 
$ are real Jordan blocks corresponding to conjugate eigenvalue $(\mu_k, \bar{\mu}_k)$
 and
 $
 J^2_{0} = 
 \begin{bmatrix}
 0 & 1\\
 0 & 0
 \end{bmatrix}
 $ is the real Jordan block corresponding to the type $2$ generalized eigenvectors of $0$ eigenvalue. Thus we have
\begin{align}
e^{t J_{\CL}} = \oplus^{m_1}_{i = 1} (e^{t J_{\mu_i, \bar{\mu}_i}}) \oplus^{m_2}_{i=1} I_1 \oplus^{m_3}_{i=1} e^{t J^2_{0}}
\end{align}
where 
\begin{itemize}
\item $e^{t J_{\mu_i, \bar{\mu}_i}}$ is defined in (\ref{expc}) with $\lambda=0,\ m=1$. Thus elements in $e^{t J_{\mu_i, \bar{\mu}_i}}$ are in $\CO(1)$.
\item $I_1 = [1] \in \BR^{1 \times 1}$.
\item $e^{t J^2_{0}}$ is defined in (\ref{expr}) with $\mu = 0, m =2$. Matrix elements in $e^{t J^2_{0}}$ are in $\CO(t)$.
\end{itemize}

More precisely, we have
\begin{align}\label{etJ_L}
\begin{tikzpicture}
  \draw(0,0)node
  {$e^{t J_{\CL}}  =
  \left[
    \begin{array}{cccccccc}
\ddots & & & & & &\\ 
 & \ddots & & & & &\\
& & 1 & t & & \\
& & & 1& &  \\
 & & && \ddots &  \\
 & & &  && 1 &t & \\
 & & &  &&  &   1 \\
    \end{array}
  \right]
$};
\draw(-1,1.8)node[above]{$\overbrace{\rule{18mm}{0mm}}$}(-1,2.3)node[above]{$m_1 + m_2$};
\draw(1.2,0.5)node[above]{$\overbrace{\rule{25mm}{0mm}}$}(1.2,0.8)node[above]{$m_3$ \mbox{columns}};
\end{tikzpicture}
\end{align}
and only elements in the upper diagonal of the last $m_3$ columns belongs to $\CO(t)$, other elements are all bounded function of $t$.

\begin{lem}\label{etl}
 Let $(e^{t \CL})_{i}$ denote the $i$-th row of $e^{t \CL}$. Then we have 
 \begin{itemize}
 \item If $AA^{\top}$ has $0$ eigenvalues, then for $i \in \{1,2,...,n\}$, elements in $(e^{t \CL})_{i}$ are bounded and for $i \in \{n+1,n+2,...,2n\}$, $(e^{t \CL})_{i}$ 
belongs to $\CO(t)$.
\item If $AA^{\top}$ doesn't have $0$ eigenvalues, all elements in $e^{t \CL}$ are bounded.
\end{itemize}
\end{lem}

\begin{proof}
Let $(M_{\CL})^{-1} = [p_1,p_2, ..., p_{2n}], \ p_i \in \BR^{2n}$ be the inverse of the real generalized modal matrix of $\CL$, then we have
\begin{align}
e^{t \CL} 
= M_{\CL}  (e^{t J_{\CL}} [w_1, w_2 , ..., w_{2n}]) .
\end{align}
If   $AA^{\top}$ has $0$ eigenvalues, from (\ref{etJ_L}), we have $ e^{t J_{\CL}} w_i$ has form
\begin{align}\label{m3row}
\begin{tikzpicture}
  \draw(5,0)node
  {$
  \left[
    \begin{array}{c}
    \vdots \\
    \vdots \\
    \vdots \\
    \hline 
    t \\
   \CO(1) \\
    \vdots \\
    t \\
    \CO(1)
       \end{array}
  \right]
$};
\draw(5.5,1.2)node[right]{$\left.\rule{0mm}{10mm}\right\}$  $\in \CO(1) $, $m_1 + m_2$ rows};
\draw(5.5,-1)node[right]{$\left.\rule{0mm}{12mm}\right\}$ $m_3$ \mbox{ rows}};
\end{tikzpicture}
.\end{align}

From Lemma \ref{firstmrow}, the first $n$ rows of $M_{\CL}$ has form
 \begin{align}\label{nrowM}
\begin{tikzpicture}
  \draw(0,0)node
  {$
  \left[
    \begin{array}{cc|cccccc}
&   &   0 &               & 0 &           &      0 & \\
&   &   0 &            & 0 &                &      0 &  \\
&  &  0 &                     & 0 &      & 0       &  \\
&   &   \vdots &                                      & \vdots &                &       \vdots &  \\
&  &  0 &                     & 0 &      & 0       &  \\
\cdots & \cdots & \vdots &      v_{m+2}               &\vdots &  \cdots      & \vdots & v_{2n} \\
&   &    &                                      &  &                &        &  \\
&   &   0 &                                      & 0 &                &      0 &  \\
&   &    &                                      &  &                &        &  \\
&   &   \vdots &                                      & \vdots &                &       \vdots &  \\
    \end{array}
  \right]
$};
\draw(-2.3,2.4)node[above]{$\overbrace{\rule{15mm}{0mm}}$}(-2.2,2.8)node[above]{$m_1 + m_2$};
\draw(1,2.4)node[above]{$\overbrace{\rule{40mm}{0mm}}$}(0.9,2.8)node[above]{$m_3$ \mbox{columns}};
\draw(3.5,0)node[right]{$\left.\rule{0mm}{25mm}\right\}$ \mbox{n rows}};
\end{tikzpicture}.
\end{align}\label{54}

The first $n$ rows of $ M_{\CL} e^{t J_{\CL}} p_i $ is the matrix-vector product of (\ref{nrowM}) and (\ref{m3row}), we can see the term $t$ in (\ref{m3row}) will product
with term $0$ in (\ref{nrowM}), thus the first $n$ rows of $ M_{\CL} e^{t J_{\CL}} p_i $ belong to $\CO(1)$. The last $n$ rows of $ M_{\CL} e^{t J_{\CL}} p_i $ belong to $\CO(t)$ because there may exist nonzero elements in the last $n$ rows of $ M_{\CL}$ that can product with $t$ in (\ref{m3row}).

If  $AA^{\top}$ doesn't have $0$ eigenvalues, from Lemma \ref{ljbJ}, all elements in $e^{J_{\CL} t} $ are bounded and since $e^{t \CL} = M_{\CL}  e^{t J_{\CL}} (M_{\CL})^{-1}$, thus all elements in $e^{t \CL}$ are bounded.
\end{proof}

The covariance evolution directly follows from above calculation.

\begin{proof}[Proof of covariance in continuous time equation] Denote $$P(t_0) = 
\begin{bmatrix}
\Var(y_1(t_0)) & \Cov(y_1(t_0), X_1(t_0)) \\
\\
 \Cov(y_1(t_0), X_1(t_0)) & \Var(X_1(t_0))
\end{bmatrix}$$ and
$e^{t \CL} = \begin{bmatrix}
A(t) & B(t) \\
\\
C(t) & D(t)
\end{bmatrix}$. Since $P(t + t_0)  = e^{t \CL} P(t_0) (e^{t \CL})^{\top} $, we have
$$\begin{aligned}
& P(t + t_0)   = 
\begin{bmatrix}
A(t) & B(t) \\
\\
C(t) & D(t)
\end{bmatrix}
\begin{bmatrix}
\Var(y_1(t_0)) & \Cov(y_1(t_0), X_1(t_0)) \\
\\
 \Cov(y_1(t_0), X_1(t_0)) & \Var(X_1(t_0))
\end{bmatrix}
 \begin{bmatrix}
(A(t))^{\top} & (C(t) )^{\top}\\
\\
(B(t) )^{\top}& (D(t))^{\top}
\end{bmatrix} \\
\\
& = \begin{bmatrix}
A(t) \Var(y_1(t_0)) + B(t)  \Cov(y_1(t_0), X_1(t_0)) & A(t) \Cov(y_1(t_0), X_1(t_0)) + B(t)  \Var(X_1(t_0)) \\
\\
C(t) \Var(y_1(t_0))  + D(t) \Cov(y_1(t_0), X_1(t_0))  & C(t)\Cov(y_1(t_0), X_1(t_0)) + D(t)  \Var(X_1(t_0)) 
\end{bmatrix}
 \begin{bmatrix}
(A(t))^{\top} & (C(t) )^{\top}\\
\\
(B(t) )^{\top}& (D(t))^{\top}
\end{bmatrix} \\
\\
& =
\begin{bmatrix}
 (A\Var(y_1) + B  \Cov )A^{\top}  +( A \Cov + B \Var(X_1)) B^{\top} &  (A\Var(y_1) + B  \Cov )C^{\top} + ( A \Cov + B \Var(X_1)) D^{\top}\\
 \\
(C \Var(y_1) + D\Cov) A^{\top} + ( C\Cov+ D \Var(X_1)) B^{\top} & (C \Var(y_1) + D\Cov ) C^{\top} + ( C\Cov+ D \Var(X_1)) D^{\top}
 \end{bmatrix} \\
 \\
 & = \begin{bmatrix}
 \Var(y_1(t_0 + t)) & \Cov(y_1(t_0 + t), X_1(t_0 + t)) \\
 \\
 \Cov(y_1(t_0 + t), X_1(t_0 + t)) & \Var(X_1(t_0 + t))
 \end{bmatrix}
\end{aligned}$$ 

From Lemma \ref{etl}, if $AA^{\top}$ has $0$ eigenvalue, elements in $A(t),B(t) $ are bounded and elements in  $C(t),D(t) $ are in $\CO(t)$ and there exist elements in $C(t),D(t) $ belongs to $\Theta(t)$ , thus from above equation we have  $\Var(y_1(t_0 + t))  \in \CO(1)$, $\Cov(y_1(t_0), X_1(t_0 + t)) \in \Theta(t)$,$ \Var(X_1(t_0 + t)) \in \Theta(t^2)$. Moreover, if  $AA^{\top}$ doesn't have $0$ eigenvalue, Lemma \ref{etl} implies all elements in $A(t), B(t), C(t), D(t)$ are bounded, thus all elements in $P(t + t_0)$ are bounded.
\end{proof}

\subsubsection{Covariance evolution of Euler discretization}\label{proof4}

\begin{proof}
The Euler discretization of (\ref{Maineq}) with step size $\eta$ can be written as 
\begin{align}
\begin{bmatrix}
y^t_1\\
\\
X^t_1 
\end{bmatrix}
& =
\begin{bmatrix}
y^{t-1}_1\\
\\
X^{t-1}_1 
\end{bmatrix}
+
\begin{bmatrix}
0 & -\eta AA^T \\
\\
\eta I &  0
\end{bmatrix}
\cdot
\begin{bmatrix}
y^{t-1}_1\\
\\
X^{t-1}_1 
\end{bmatrix}
& =
\begin{bmatrix}
I & -\eta AA^T \\
\\
\eta I & I 
\end{bmatrix}
\cdot
\begin{bmatrix}
y^{t-1}_1\\
\\
X^{t-1}_1 
\end{bmatrix}.
\end{align}
We want to prove if $(y^t_1, X^t_1)$ evolve as this equation, then $\Cov(y^t_1),\Cov(X^t_1)$ have exponential growth rate. 
This can be done by calculating the real Jordan normal form of 
$\begin{bmatrix}
I & -\eta AA^T \\
\\
\eta I & I 
\end{bmatrix}$. Since $AA^{\top} \in \BR^{n \times n}$ is a diagonalizable matrix, there exists a matrix $P$, such that

\begin{align}
P AA^{\top} P^{-1} =
\begin{bmatrix}
\gamma_1 & & & \\
& \gamma_2 & & \\
& & \ddots & \\
& & & \gamma_n
\end{bmatrix}
\end{align}

where $\gamma_1 ,..., \gamma_n$ are eigenvalues of $AA^{\top}$. Moreover, since $AA^{\top}$ is a positive semidefinite matrix, we have $\gamma_i \ge 0, i \in [n]$. 

Then we have
\begin{align}\label{matrix_sim}
\begin{bmatrix}
P & &\\
\\
&& P
\end{bmatrix}
\begin{bmatrix}
I & -\eta AA^T \\
\\
\eta I & I 
\end{bmatrix}
\begin{bmatrix}
P^{-1} & \\
\\
& P^{-1}
\end{bmatrix}
= 
\begin{bmatrix}
\begin{array}{cccc|cccc}
1 & & & & -\gamma_1 \eta & & & \\
& 1 & & & &-\gamma_2 \eta & & \\
& & \ddots & & & & \ddots & \\
& & & 1& & & & - \gamma_n \eta \\ \hline
\eta & & & &1& & & \\
& \eta & & & &1 & & \\
& & \ddots & & & & \ddots & \\
& & &\eta  & & & &1
\end{array}
\end{bmatrix}.
\end{align}
This matrix is similar to $\begin{bmatrix}
I & -\eta AA^T \\
\\
\eta I & I 
\end{bmatrix}$, thus they have same character polynomial defined by

\begin{align}
\det (\mu I - & \begin{bmatrix}
\begin{array}{cccc|cccc}
1 & & & & -\gamma_1 \eta & & & \\
& 1 & & & &-\gamma_2 \eta & & \\
& & \ddots & & & & \ddots & \\
& & & 1& & & & - \gamma_n \eta \\ \hline
\eta & & & &1& & & \\
& \eta & & & &1 & & \\
& & \ddots & & & & \ddots & \\
& & &\eta  & & & &1
\end{array}
\end{bmatrix} ) \\
& = \det ( \begin{bmatrix}
\begin{array}{cccc|cccc}
\mu - 1 & & & & \gamma_1 \eta & & & \\
& \mu - 1 & & & &\gamma_2 \eta & & \\
& & \ddots & & & & \ddots & \\
& & &\mu - 1& & & & -\gamma_n \eta \\ \hline
- \eta & & & &\mu-1& & & \\
&- \eta & & & &\mu-1 & & \\
& & \ddots & & & & \ddots & \\
& & &- \eta  & & & &\mu-1
\end{array}
\end{bmatrix} ) \\
& = \det(
\begin{bmatrix}
(\mu-1)^2 + \gamma_1 \eta^2 & & & \\
& (\mu-1)^2 + \gamma_2 \eta^2 & & \\
& & \ddots & \\ 
&&&(\mu-1)^2 + \gamma_n \eta^2 
\end{bmatrix}
)\\
& = \prod^n_{i =1}( (\mu-1)^2 + \gamma_i \eta^2 )
\end{align} 

Recall we have $\gamma_i \ge 0, i \in [n]$, thus the root of character polynomial is
$
\mu_j, \bar{\mu}_j = 1 \pm i \sqrt{\gamma_j} \eta
,\ j \in [n]$. Moreover, we have $\lvert  \mu_j \lvert= 1 + \gamma_j \eta^2 \ge 1$. Thus by Proposition \ref{differenceim}, elements in
$\begin{bmatrix}
I & -\eta AA^T \\
\\
\eta I & I 
\end{bmatrix}^t$ have an exponential growth rate as $\Theta(\lvert  \mu \lvert^t )$, where $\mu$ is the eigenvalue of $AA^{\top}$ which has  largest norm.  

Since 
\begin{align} P(t)& = 
\begin{bmatrix}
\Var(y^t_1) & \Cov(y^t_1, X^t_1)\\
\Cov(X^t_1, y^t_1) & \Var(X^t_1)
\end{bmatrix} \\
&=
\begin{bmatrix}
I & -\eta AA^T \\
\\
\eta I & I 
\end{bmatrix}^{t-t_0}
\begin{bmatrix}
\Var(y^{t_0}_1) & \Cov(y^{t_0}_1, X^{t_0}_1)\\
\\
\Cov(X^{t_0}_1, y^{t_0}_1) & \Var(X^{t_0}_1)
\end{bmatrix}
(\begin{bmatrix}
I & -\eta AA^T \\
\\
\eta I & I 
\end{bmatrix}^{t-t_0})^{\top}
\end{align}

Since elements in matrix $\begin{bmatrix}
I & -\eta AA^T \\
\\
\eta I & I 
\end{bmatrix}^{t-t_0}
$ has growth rate $\Theta(\lvert  \mu \lvert^{t} )$, thus elements in $P(t)$ has growth rate $\Theta(\lvert  \mu \lvert^{2t} )$.
\end{proof}

\subsubsection{Covariance evolution of Symplectic discretization}\label{proof5}

We firstly determine the Symplectic discretization of continuous FTRL (\ref{Maineq}) with Euclidian norm regularizer. 

\begin{lem}
Discrete continuous FTRL (\ref{Maineq}) with (\ref{type1}), we get

\begin{align}
\begin{bmatrix}
y_1^{n+1}\\
\\
X_1^{n+1}
\end{bmatrix}
= 
\begin{bmatrix}
I & -\eta AA^{\top} \\
\\
\eta & I - \eta^2 AA^{\top}
\end{bmatrix}
\cdot
\begin{bmatrix}
y_1^{n}\\
\\
X_1^{n}
\end{bmatrix} .
\end{align}
\end{lem}

\begin{proof} Directly calculate gives
\begin{align}
y_1^{t+1}& = y^t_1 - \eta AA^{\top}X^t_1, \\
X_1^{t+1}& = X^t_1 + \eta y_1^{t+1} \\
& = X^t_1 + \eta y^t_1 - \eta^2 AA^{\top}X^t_1.
\end{align}
Combine above gives
\begin{align}
\begin{bmatrix}
y_1^{n+1}\\
\\
X_1^{n+1}
\end{bmatrix}
= 
\begin{bmatrix}
I & -\eta AA^{\top} \\
\\
\eta & I - \eta^2 AA^{\top}
\end{bmatrix}
\cdot
\begin{bmatrix}
y_1^{n}\\
\\
X_1^{n}
\end{bmatrix} .
\end{align}
This finish the proof.
\end{proof}

Let $\CM_{\eta}  = \begin{bmatrix}
I & -AA^{\top} \cdot \eta\\
\\
\eta & I - AA^{\top} \cdot \eta^2 
\end{bmatrix}$. Since $AA^{\top}$ is diagonalizable, and $AA^{\top}$'s eigenvalues are all non-negative, there exists a matrix $P \in \BR^{n \times n}$ and $P$ invertible, such that

\begin{align}\label{diagAAtop}
P^{-1} AA^{\top}P & = 
\begin{bmatrix}
\gamma_1 & & & \\
& \gamma_2 & & \\
& & \ddots & \\
& & & \gamma_n
\end{bmatrix}
\end{align}
where $\gamma_1 , ..., \gamma_n \ge 0$ are eigenvalues of $AA^{\top}$.

\begin{lem}\label{M_eta}
$\CM_{\eta}$ has real eigenvalues if and only if $AA^{\top}$ has $0$ eigenvalue, and in this case the only real eigenvalue of $\CM_{\eta}$  equals to $1$.
Moreover, every image eigenvalue of $\CM_{\eta}$ has norm equals to $1$.
\end{lem}

\begin{proof}

With (\ref{diagAAtop}), we have
\begin{align}
\det (\mu I- \CM_{\eta}) & = 
\det (
\begin{bmatrix}
\mu - I & AA^{\top} \eta\\
\\
- \eta & \mu - I + AA^{\top} \eta^2
\end{bmatrix}
)\\
& = \det( (\mu - I )(\mu - I + AA^{\top}\eta^2) + AA^{\top}\eta^2    )\\
& = \det (\mu^2 - 2 \mu + I + \mu AA^{\top} \eta^2) \\
& \overset{(\ref{diagAAtop})}{=} \det (
\begin{bmatrix}
\mu^2 - 2 \mu + 1 + \mu \eta^2 \gamma_1 & & & \\
& \mu^2 - 2 \mu + 1 + \mu \eta^2 \gamma_2 & & \\
&&\ddots & \\
&&& \mu^2 - 2 \mu + 1 + \mu \eta^2 \gamma_n
\end{bmatrix}) \\
& = \prod^n_{i = 1} (\mu^2 - 2 \mu + 1 + \mu \eta^2 \gamma_i)
\end{align}

From above, we can see if $\mu$ makes $\det(\mu I - \CM_{\eta}) = 0$, then there exists some $i \in [n]$, such that
\begin{align}\label{quadratic}
\mu^2 - 2 \mu + 1 + \mu \eta^2 \gamma_i = 0.
\end{align}

(\ref{quadratic}) is a quadratic function about $\mu$, and has solution $\mu = \frac{(2 - \eta^2 \gamma_i) \pm \sqrt{(\eta^2 \gamma_i - 2)^2 - 4}}{2}$.

To make $\mu \in \BR$, we need $(\eta^2 \gamma_i - 2)^2 - 4 \ge 0$. For sufficient small $\eta$ ($\eta^2 \gamma_i \le 2$), that can only happen when $\gamma_i = 0$, and in this case we have $\mu = 1$.
Moreover, if  $(\eta^2 \gamma_i - 2)^2 - 4 < 0$, we have
\begin{align}
\lVert\mu \lVert^2& =\lVert \frac{(2 - \eta^2 \gamma_i) \pm i \sqrt{ 4 - (\eta^2 \gamma_i - 2)^2}}{2} \lVert ^2\\
& = \frac{(2 - \eta^2 \gamma_i)^2 + 4 - (\eta^2 \gamma_i - 2)^2 }{4}\\
& = 1
\end{align}

Thus we have
\begin{itemize}
\item $\CM_{\eta}$ has a real eigenvalue if and only if $\gamma = 0$ is an eigenvalue of $AA^{\top}$, and in this case the real eigenvalue of $\CM_{\eta}$ equals to $1$.
\item  Every image eigenvalue of $\CM_{\eta}$  has norm equals to $1$.
\end{itemize}
\end{proof}

\begin{lem}\label{mu=1}
The largest Jordan blocks corresponding to $\mu =1$ have size $2$.
\end{lem}

\begin{proof}As in proposition \ref{ker}, the number of size $k$ Jordan blocks corresponding to eigenvalue $1$ is determined by the dimension of linear space $\ker(\CM_{\eta} - I)^{k}$ and 
$\ker(\CM_{\eta} - I)^{k-1}$. Since similar matrixes have same minimal polynomial and same kernel space, thus we can change $\CM_{\eta}$ to any similar matrix. Thus by (\ref{diagAAtop}), we only need to consider 
\begin{align}
\begin{bmatrix}
P & 0\\
0 & P
\end{bmatrix}
\CM_{\eta}
\begin{bmatrix}
P^{-1} & 0\\
0 & P^{-1}
\end{bmatrix}
- I 
= 
\begin{bmatrix}
\begin{array}{cccc|cccc}
0 & & & & -\gamma_1 \eta & & & \\
& 0 & & & &-\gamma_2 \eta & & \\
& & \ddots & & & & \ddots & \\
& & & 0 & & & & - \gamma_n \eta \\ \hline
\eta & & & &-\gamma_1 \eta^2 & & & \\
& \eta & & & &-\gamma_2 \eta^2 & & \\
& & \ddots & & & & \ddots & \\
& & & \eta & & & & - \gamma_n \eta^2
\end{array}
\end{bmatrix}  \label{bigmatrix}
\end{align} 

Our claim is : $k = 3$ is the smallest number to make $\rho_k = 0$ in  proposition \ref{ker}, thus by proposition \ref{ker} the largest Jordan block 
corresponding to eigenvalue $1$ have size $2$ . 

As shown in lemma \ref{M_eta}, $1$ is an eigenvalue of $\CM_{\eta}$ if and only if $0$ is an eigenvalue of $AA^{\top}$.
Assume $0$ is an eigenvalue of $AA^{\top}$ with multiplicity $m$, then it is easy to see (\ref{bigmatrix}) has rank $2n - m$.
A direct calculate show the square of (\ref{bigmatrix}) equals to

\begin{align}
\begin{bmatrix}
\begin{array}{cccc|cccc}
-\gamma_1 \eta^2 & & & & -\gamma_1 \eta^3 & & & \\
& -\gamma_2 \eta^2 & & & &-\gamma_2 \eta^3 & & \\
& & \ddots & & & & \ddots & \\
& & & -\gamma_n \eta^2 & & & & - \gamma_n \eta^3 \\ \hline
-\gamma_1 \eta^3 & & & &-\gamma_1 \eta^2 + \gamma^2_1 \eta^4& & & \\
&  -\gamma_2 \eta^3 & & & &-\gamma_2 \eta^2 + \gamma^2_2 \eta^4 & & \\
& & \ddots & & & & \ddots & \\
& & & -\gamma_n \eta^3  & & & & -\gamma_n \eta^2 + \gamma^2_n \eta^4
\end{array}
\end{bmatrix}  \label{bigmatrix2}
\end{align} 
Thus if $0$ is an eigenvalue of $AA^{\top}$ with multiplicity $m$, then there are $2m$ rows in (\ref{bigmatrix2}) be $0$, thus  (\ref{bigmatrix2}) has rank $2n - 2m$, this shows the $\rho_2$ in 
proposition \ref{ker} is not zero, which implies there exists Jordan blocks with size $2 \times 2$ corresponding to eigenvalue $1$.

Moreover, the cubic of (\ref{bigmatrix}) equals to

\begin{align}
\begin{bmatrix}
\begin{array}{cccc|cccc}
\gamma^2_1 \eta^4 & & & & \gamma_1 \eta^3 - \gamma^3_1 \eta^5& & & \\
& \gamma^2_2 \eta^4 & & & & \gamma_2 \eta^3 - \gamma^3_2 \eta^5& & \\
& & \ddots & & & & \ddots & \\
& & & \gamma^2_n \eta^4 & & & & \gamma_n \eta^3 - \gamma^3_n \eta^5 \\ \hline
-\gamma_1 \eta^3 & & & &2 \gamma^2_1 \eta^4 - \gamma^3_1 \eta^6& & & \\
&  -\gamma_2 \eta^3 & & & &2 \gamma^2_2 \eta^4 - \gamma^3_2 \eta^6 & & \\
& & \ddots & & & & \ddots & \\
& & & -\gamma_n \eta^3  & & & &2 \gamma^2_n \eta^4 - \gamma^3_n \eta^6
\end{array}
\end{bmatrix}  \label{bigmatrix3}
\end{align} 
We can see the rank of (\ref{bigmatrix3}) is also $2n - 2m$, thus $\rho_3$ in 
proposition \ref{ker} is zero. This finish the proof of our claim.
\end{proof}

The following lemma determines $\CM_{\eta}$'s  type $1$ and type $2$ generalized eigenvectors  of $\CM_{\eta}$ corresponding to eigenvalue $\mu = 1$. Recall $v \in \BR^{2n}$ is a type $2$ generalized  eigenvectors  of $\CM_{\eta}$ corresponding to eigenvalue $\mu = 1$ if $v \in \ker(\CM_{\eta} - I)^2$ but  $v \notin \ker(\CM_{\eta} - I)$. 
\begin{lem}\label{genvecMeta}
 Let $x,y \in \BR^{n}$, then
\begin{itemize}
\item[(1)] If $(x,y) \in \BR^{2n}$ is a type $2$ generalized eigenvectors  of $\CM_{\eta}$ corresponding to eigenvalue $\mu = 1$, then $x,y \in \ker(AA^{\top})$ and $x \ne 0$.
\item[(2)] If $(x,y) \in \BR^{2n}$ is a type $1$ generalized eigenvectors  of $\CM_{\eta}$ corresponding to eigenvalue $\mu = 1$, then $y \in \ker(AA^{\top})$ and $x = 0$.
\end{itemize}
\end{lem}

\begin{proof}We have
\begin{align}
(\CM_{\eta} - I)^2 = 
\begin{bmatrix}
-AA^{\top} \eta^2 & (AA^{\top} )^2 \eta^3 \\
-AA^{\top} \eta^3 & -AA^{\top} \eta^2 +  (AA^{\top} )^2 \eta^4
\end{bmatrix}.
\end{align}
Thus if $(x,y) \in \ker (\CM_{\eta} - I)^2$, we have
\begin{align}
-AA^{\top} \eta^2 x + (AA^{\top} )^2 \eta^3 y = 0 \label{gv1}\\
-AA^{\top} \eta^3 x +  (-AA^{\top} \eta^2 +  (AA^{\top} )^2 \eta^4)y = 0 \label{gv2}
\end{align}
$(\ref{gv2}) - \eta \cdot (\ref{gv1})$ gives
\begin{align}
-AA^{\top} \eta^2 y = 0.
\end{align}
Thus we have $y \in \ker (AA^{\top})$, and take this back to (\ref{gv2}), we get $x \in \ker(AA^{\top})$. 

Moreover, it is directly to verify if $(x,y) \in \ker(\CM_{\eta} - I)$, then $x = 0$ and $y \in \ker(AA^{\top})$. Thus if $(x,y) $ is a type $2$ generalized eigenvector, $x \in \ker(AA^{\top}) ,\ x \ne 0$
and $ y \in \ker (AA^{\top})$.

\end{proof}

\begin{lem}\label{imvalue}
For an image eigenvalue $\mu \ne 1$ of $\CM_{\eta}$,  the largest real Jordan blocks corresponding to $\mu$ have size $2$.
\end{lem}

\begin{proof} The proof is similar to lemma \ref{mu=1}. Let $\mu \ne 1$ be an image eigenvalue of $\CM_{\eta}$, by (\ref{quadratic}), there exists an eigenvalue
$\gamma_i$ of $AA^{\top}$ to make 

\begin{align}
\mu^2 - 2 \mu + 1 + \mu \eta^2 \gamma_i = 0.
\end{align}
Moreover, the algebraic multiplicity of $\mu$ as an eigenvalue of $\CM_{\eta}$ is same as the algebraic multiplicity of $\gamma$ as an eigenvalue of $AA^{\top}$. We assume $\mu \ne 1$
has algebraic multiplicity $m$.

Consider the matrix
\begin{align}
\mu -
\begin{bmatrix}
P & 0\\
0 & P
\end{bmatrix}
\CM_{\eta}
\begin{bmatrix}
P^{-1} & 0\\
0 & P^{-1}
\end{bmatrix}
= 
\begin{bmatrix}
\begin{array}{cccc|cccc}
\mu - 1 & & & & \gamma_1 \eta & & & \\
& \mu-1 & & & &\gamma_2 \eta & & \\
& & \ddots & & & & \ddots & \\
& & & \mu-1 & & & &  \gamma_n \eta \\ \hline
-\eta & & & &\mu - 1 + \gamma_1 \eta^2 & & & \\
& -\eta & & & &\mu - 1 +\gamma_2 \eta^2 & & \\
& & \ddots & & & & \ddots & \\
& & & -\eta & & & & \mu - 1 + \gamma_n \eta^2
\end{array}
\end{bmatrix}  \label{bigimagematrix}
\end{align} 

It is easy to see if 

\begin{align}\label{ewrwer}
\mu^2 - 2 \mu + 1 + \mu \eta^2 \gamma_i = 0
\end{align}
the $i$-th and $n+i$-th row of (\ref{bigimagematrix})  are linearly dependent. Moreover, if  if $\mu$ as rank $m$, the multiplicity  of $\gamma_i$ as an eigenvalue of $AA^{\top}$
is also $m$. So there are $m$ rows in the first $n$ rows of (\ref{bigimagematrix}) linearly dependent on  $m$ rows in the last $n$ rows of (\ref{bigimagematrix}). Thus if $\mu$ as rank $m$,  (\ref{bigimagematrix}) has rank $2n-m$. 

A directly calculate shows the square of (\ref{bigimagematrix})
equals to 
\begin{align}
\begin{bmatrix}
\begin{array}{ccc|ccc}
(\mu - 1)^2 - \gamma_1 \eta^2 & &              & \gamma_1 \eta (2 \mu -2 + \gamma_1 \eta^2)& &  \\
& \ddots &                                                       & & \ddots & \\
& & (\mu - 1)^2 - \gamma_n \eta^2              & & &  \gamma_n \eta (2 \mu -2 + \gamma_n \eta^2) \\ \hline
-\eta & &                                                          & -\gamma_1 \eta^2 + (\mu-1 + \gamma_1 \eta)^2 & &  \\
& \ddots &                                                        & & \ddots & \\
& &  -\eta                                                          & & & -\gamma_n \eta^2 + (\mu-1 + \gamma_n \eta)^2 
\end{array}
\end{bmatrix}  \label{bigimagematrix2}
\end{align} 
Moreover, if $\mu$ and $\gamma_i$ satisfy (\ref{ewrwer}), the $i$-th row and $n+i$-th row are linearly dependent, thus same as matrix in (\ref{bigimagematrix}),  \label{bigimagematrix2} has rank $2n-m$ if $\mu$ is an eigenvalue of $\CM_{\eta}$ with multiplicity $m$. 

Thus we have $\rho_2 = 0$ in proposition \ref{ker}, this implies the real Jordan blocks of a pair of conjugate eigenvalues $(\mu,\bar{\mu})$ have size $2$.
\end{proof}

\begin{lem}\label{i-th row of the matrix} Let $t$ be a positive integer, and $(\CM^t_{\eta} )_i$ be the  i-th row of the matrix  $\CM^t_{\eta}$ for $i = 1,2,...,2n$. We have
\begin{itemize}
\item If $AA^{\top}$ is singular, then for $1 \le i \le n$, elements in  $(\CM^t_{\eta} )_i$ is bounded, for $n+1 \le i \le 2n$, elements in  $(\CM^t_{\eta} )_i$ have growth rate $\CO(t)$.
\item If $AA^{\top}$ is non-singular, then for all $1 \le i \le 2n$, elements in  $(\CM^t_{\eta} )_i$ is bounded.
\end{itemize}
\end{lem}

\begin{proof} Denote the real generalized modal matrix of $\CM_{\eta}$ by $M$, then we have  $\CM^t_{\eta} = M J^t_{\CM_{\eta}} M^{-1} $, where $J_{\CM_{\eta}}$ is the real Jordan normal form of $\CM_{\eta}$.  If  $AA^{\top}$ has $0$ eigenvalue, then from Lemma \ref{M_eta} and Lemma \ref{mu=1}, $1$ is the only possible real eigenvalue of $\CM_{\eta}$, and the largest size real Jordan block have size $2$. Moreover in this case, form Lemma \ref{genvecMeta} the real generalized modal matrix of $\CM_{\eta}$ has same structure as the real generalized modal matrix of $\CL$ as in (\ref{nrowM}), thus the proof follows from a same argument as in Proposition \ref{etl}.  

If  $AA^{\top}$ doesn't have $0$ eigenvalue, all eigenvalues of $\CM_{\eta}$ are image numbers and form Lemma \ref{M_eta} these eigenvalues have norm $1$. From Lemma \ref{imvalue} the largest size real Jordan blocks of $\CM_{\eta}$  is $2 \times 2$, thus from Proposition \ref{differenceim} with $m=1$, all elements in $J^t_{\CM_{\eta}}$ are bounded, thus all elements in  $\CM^t_{\eta} $ are bounded.
\end{proof}

The covariance evolution of Symplectic discretization directly follows from above calculations. Denote the covariance matrix  at time $t+t_0$ by $P(t+t_0)$ and $\CM^t_{\eta} = 
\begin{bmatrix}
A^t & B^t\\
C^t & D^t
\end{bmatrix}$, then we have
\begin{align} 
& P(t + t_0)  = 
\begin{bmatrix}
\Var(y^{t+t_0}_1) & \Cov(y^{t+t_0}_1, X^{t+t_0}_1)\\
\\
\Cov(X^{t+t_0}_1, y^{t+t_0}_1) & \Var(X^{t+t_0}_1)
\end{bmatrix} \\
&=
\CM^t_{\eta}
\begin{bmatrix}
\Var(y^{t_0}_1) & \Cov(y^{t_0}_1, X^{t_0}_1)\\
\\
\Cov(X^{t_0}_1, y^{t_0}_1) & \Var(X^{t_0}_1)
\end{bmatrix}
 (\CM^t_{\eta})^{\top} \\
 & = 
 \begin{bmatrix}
A^t & B^t\\
C^t & D^t
\end{bmatrix}
\begin{bmatrix}
\Var(y^{t_0}_1) & \Cov(y^{t_0}_1, X^{t_0}_1)\\
\\
\Cov(X^{t_0}_1, y^{t_0}_1) & \Var(X^{t_0}_1)
\end{bmatrix}
 (\begin{bmatrix}
A^t & B^t\\
\\
C^t & D^t
\end{bmatrix})^{\top}
\\
& = \begin{bmatrix}
A^t \Var(y^{t_0}_1)  + B^t \Cov(X^{t_0}_1, y^{t_0}_1)  & A^t \Cov(y^{t_0}_1, X^{t_0}_1)  + B^t  \Var(X^{t_0}_1) \\
\\
C^t \Var(y^{t_0}_1)  + D^t \Cov(X^{t_0}_1, y^{t_0}_1)  & C^t \Cov(y^{t_0}_1, X^{t_0}_1)  + D^t  \Var(X^{t_0}_1) 
\end{bmatrix}
 (\begin{bmatrix}
(A^t)^{\top} & (C^t)^{\top}\\
\\
(B^t)^{\top} & (D^t)^{\top}
\end{bmatrix})^{\top} \label{ahfbeouahl}
 \end{align}
The finial equation (\ref{ahfbeouahl}) equals to
\begin{align}
\begin{bmatrix}
P^t_1 & P^t_2 \\
\\
P^t_3 & P^t_4
\end{bmatrix},
\end{align}
where
\begin{align}
&P^t_1 = (A^t \Var(y^{t_0}_1)  + B^t \Cov)(A^t)^{\top} + (A^t \Cov  + B^t  \Var(X^{t_0}_1) ) (B^t)^{\top}, \\
&P^t_2 = (A^t \Var(y^{t_0}_1)  + B^t \Cov)(C^t)^{\top} + (A^t \Cov  + B^t  \Var(X^{t_0}_1) ) (D^t)^{\top}, \\
&P^t_3 = (C^t \Var(y^{t_0}_1)  + D^t \Cov)(A^t)^{\top} + (C^t \Cov  + D^t  \Var(X^{t_0}_1) ) (B^t)^{\top}, \\
&P^t_4 = (C^t \Var(y^{t_0}_1)  + D^t \Cov)(C^t)^{\top} + (C^t \Cov  + D^t  \Var(X^{t_0}_1) ) (D^t)^{\top}.
\end{align}
From Lemma \ref{i-th row of the matrix}, when $AA^{\top}$ is non-singular, all elements in $\CM^t_{\eta}$ are bounded, thus elements in $P(t + t_0)$ is bounded. When
 $AA^{\top}$ is singular, elements in $A^t,B^t$ are bounded and elements in $C^t,D^t$ has linear growth rate, thus $\Var(y^t) \in \CO(1)$, $\Cov(X^t,y^t) \in \Theta(t)$ and $\Cov(X^t_i,X^t_j) \in \Theta(t^2)$.

\section{Proof of Section \ref{last}}\label{Details of Section "last"}
\subsection{Riemannian Game Dynamics}
We collect minimum amount of terminologies on Riemannian game dynamics, for a complete treatment on this topic, we refer to \citep{MertiSandholm}. For the case of population game $\mathcal{G}(\mathcal{A},v)$ where $\mathcal{A}$ is the strategy set and $v$ is the set of utilities, the \emph{gain from motion} from state $x\in\mathcal{X}$ along $z\in\mathbb{R}^{\mathcal{A}}$ is defined as
\[
G^v(x;z)=\sum_{\alpha\in\mathcal{A}}v_{\alpha}(x)z_{\alpha}.
\]
The \emph{cost of motion} $C(x;z)$ represents the intrisic difficulty of moving from state $x$ along a given displacement vector $z$, and it is defined to be
\[
C(x;z)=\frac{1}{2}g_x(z,z)
\]
where $g$ is a smooth assignment of symmetric positive definite matrices $g_x$ to each state $x\in\mathcal{X}$. The vector of motion from state $x$ is required to maximize the difference between the gain of motion $G^v(x;z)$ and the cost of motion $C(x;z)$ subject to
\begin{equation}\label{RGameDyn}
\dot{x}=\argmax_{z\in T_{x}\mathcal{X}}\{G^v(x;z)-C(x;z)\}.
\end{equation}
The dynamics \eqref{RGameDyn} is called \emph{Riemannian game dynamics}.

\subsection{Symplectic Geometry}\label{SG}
\paragraph{Symplectic form.} In order to present Gromov's non-squeezing theorem and its implication in uncertainty principle, we need some terminology of Symplectic Geometry. Roughly speaking, symplectic geometry studies the geometry of the space (of even dimension) equipped with the symplectic form. Take Euclidean geometry for example, it studies the vector space $\mathbb{R}^n$ with an inner product structure $\langle,\rangle:\mathbb{R}^n\times\mathbb{R}^n\rightarrow\mathbb{R}$ called the Euclidean structure. A symplectic form on a even dimensional space $\mathbb{R}^n$ is a \emph{skew-symmetric} bilinear map $\omega(\cdot,\cdot):\mathbb{R}^n\times\mathbb{R}^n\rightarrow\mathbb{R}$, satisfying
\begin{itemize}
\item $\omega(u,v)=-\omega(v,u)$ for all $u,v\in\mathbb{R}^n$.
\item $\omega(v,v)=0$ for all $v\in\mathbb{R}^n$.
\item $\omega(u,v)=0$ for all $v\in\mathbb{R}^n$ implies that $u=0$.
\end{itemize}
A typical symplectic form is the bilinear map defined by matrix 
$
J=\left(\begin{array}{cc} 0&I_n\\-I_n&0\end{array}\right)
$,
where $I_n$ denotes the identity matrix on $\mathbb{R}^n$. A basis $\{u_1,...,u_n,v_1,...,v_n\}$ of $\BR^{2n}$ is called $\omega$-standard if $\omega(u_j,u_k)=-\omega(v_j,v_k) = 0$ and $\omega(u_j,v_k)=\delta_{jk}$.

\paragraph{Symplectomorphism.} A symplectomorphism $\varphi$ between symplectic vector spaces $(\mathbb{R}^n,\omega)$ and $(\mathbb{R}^n,\omega')$ is a linear isomorphism $\varphi:\mathbb{R}^n\rightarrow\mathbb{R}^n$ such that $\varphi^*\omega'=\omega$, where $(\varphi^*\omega')(u,v):=\omega'(\varphi(u),\varphi(v))$. More generally, let $f:(\mathbb{R}^n,\omega)\rightarrow(\mathbb{R}^n,\omega')$ be a diffeomorphism. Then $f$ is a symplectomorphism if $f^*\omega'=\omega$. These definitions generalize to manifold settings, but further discussion is beyond the scope of current paper. In the plane, symplectic form represents area, so a symplectic mapping is equivalent to an area preserving mapping.

\paragraph{Linear symplectic width.}The main technique in the analysis of general regularizers is to leverage the power of symplectic geometry, especially a classic work of Gromov in 1980's \citep{G85}, to obtain a lower bound for the covariance of the conjugate coordinates. It is known as "Gromov's Non-squeezing Theorem", 
\begin{theorem}[Gromov, 1985.]
If $R<r$, there does not exist Hamiltonian map $\varphi:\mathbb{R}^{2n}\rightarrow\mathbb{R}^{2n}$ such that $\varphi(B(r))\subset Z(R)$, where $B(r)=\{(x,y)\in\mathbb{R}^{2n}:\norm{x}^2+\norm{y}^2\le r^2\}$ and $Z(R)=\{(x,y)\in\mathbb{R}^{2n}:x_i^2+y_i^2\le R^2\}$ for any $i\in[n]$. 
\end{theorem}
Gromov's non-squeezing theorem asserts the following fact: Let $B(r)$ be a ball in the phase space $\mathbb{R}^n\times\mathbb{R}^n$, with center $(a,b)$ and radius $r$:
$
B(r):\norm{x-a}^2+\norm{y-b}^2\le r^2
$.
The orthogonal projection of this ball on any plane of coordinates  alway contains a disc of radius $r$. Now suppose that the ball is moved by a Hamiltonian flow $\varphi(t,\cdot)$, i.e., each point of $B(r)$ serves as an initial condition of a system of Hamiltonian equations. By Liouville's Theorem, the image $\varphi(t,B(r))$ at any moment $t$ has the volume the same as the initial shape ball $B(r)$, but the shape is distorted. If we pair the conjugate coordinates $x_i$ and $y_i$, then the projection of the deformed ball on any $(x_i,y_i)$-plane will never decrease below its original value $\pi r^2$.  The FTRL algorithm induces a linear Hamiltonian system whose solution is a linear symplectic mapping on phase space $\mathbb{R}^{n_1}\times\mathbb{R}^{n_1}$. It suffices to consider a narrowed concept called ``Linear symplectic width" which is defined below.
\begin{definition}[Linear symplectic width, \citep{hsiao2006fundamental}]
The linear symplectic width of an arbitrary subset $A\subset\mathbb{R}^{2n}$, denoted as $w_L(A)$, is defined as:
\[
w_L(A)=\sup_{r\in\mathbb{R}^+}\left\{\pi r^2:\phi(B^{2n}(r))\subset A\ \text{for some }\phi\in AS_p(\mathbb{R}^{2n})\right\},
\]
where $AS_p(\mathbb{R}^{2n})$ denotes the group of affine symplectomorphisms, i.e., linear map followed by translation.
\end{definition}

\subsection{Proof of Theorem \ref{prop:nonlinear}}\label{apd3}
The first ingredient in proving Proposition \ref{prop:nonlinear} is based on standard results of Taylor series method in \citep{benaroya2005probability}. Suppose $Y=g(X)$ where $X$ is a random variable with $\mu_X$ the mean value. A full Taylor expansion of $Y=g(X)$ about the mean value yields
\[
Y=g(X)|_{X=\mu_X}+(X-\mu_X)\frac{dg}{dx}|_{X=\mu_X}+\frac{1}{2!}(X-\mu_X)^2\frac{d^2g}{dx^2}|_{X=\mu_X}+...
\]
and the expectation is given by
\[
\mu_Y=g(\mu_X)+\frac{\sigma^2_X}{2}g''(\mu_X)+...
\]
which holds due to $E[X-\mu_X]$=0. Furthermore, the variance of $Y$ can be estimated as follows.
\begin{align}
\sigma_Y^2=E[Y^2]-\mu_X^2\backsimeq g^2(\mu_X)+\sigma_X^2\left([g'(\mu_X)]^2+g(\mu_X)g''(\mu_X)\right)-\left(g(\mu_X)+\frac{\sigma_X^2}{2}g''(\mu_X)\right)^2.
\end{align}
Therefore, if we assume that $\sigma_X^4\ll \sigma_X^2$ which can be implied by $\sigma_X\ll 1$, the approximation to order $\sigma_X^2$ for the variance is 
\[
\sigma_Y^2\backsimeq\sigma_X^2\left(g'(\mu_X)\right)^2.
\] 
This estimate enables one to focus only on the linearization of a general differentiable map on the multivariable case, with inevitably some extra assumption on higher order derivatives. In general, suppose $Y=g(X_1,...,X_n)$, the Taylor series expansion about the mean value of each variable yield
\[
\Var[Y]=\sum_{i=1}^n\left(\frac{\partial g(\mu_1,...,\mu_n)}{\partial X_i}\right)^2\sigma_i^2+\sum_{i=1}^n\sum_{j=1,j\ne i}^n\left(\frac{\partial g(\mu_1,...,\mu_n)}{\partial X_i}\right)\left(\frac{\partial g(\mu_1,...,\mu_n)}{\partial X_j}\right)\rho_{ij}\sigma_i\sigma_j+...
\]
where $\rho_{ij}=\frac{\Cov(X_i,X_j)}{\sigma_i\sigma_j}$. 

The other ingredient we need is the linear symplectic width evolving under a time-dependent linear Hamiltonian flow, which is the result of \citep{hsiao2006fundamental}. In the setting of classic mechanics, the position and momentum $(\vec{q},\vec{p})$ in a Hamiltonian system, the covariance matrix of $X=(q_1,...,q_n,p_1,...,p_n)$ is given by 
\[
P=E[XX^{\top}]
\]
where we assume the system is zero-mean for convenience. If the system is linear
\[
\dot{X}=A(t)X
\]
with $X(t)=\Phi(t,t_0)X_0$ its solution, then the covariance is mapped as $P=\Phi P_0 \Phi^{\top}$. Furthermore, if we partition $P$ into blocks such that
\[
P=\left[\begin{array}{ccc}
P_{11}&\dots&P_{1n}
\\
\vdots&\ddots&\vdots
\\
P_{n1}&\dots&P_{nn}
\end{array}\right]
\]
where 
\[
P_{ij}=\left[\begin{array}{cc}
E[q_iq_j^{\top}]&E[q_ip_j^{\top}]
\\
\\
E[p_iq_j^{\top}]&E[p_ip_j^{\top}]
\end{array}\right].
\]
Then Theorem 3 of \citep{hsiao2006fundamental} asserts that
\[
\abs{P_{ii}(t)}\ge\left(\frac{w_L(P_0)}{\pi}\right)^2\ \ \text{for all}\ \ i=1,...,n.
\]

The third ingredient is the Taylor expansion of differential mapping $f:V\rightarrow W$ between higher dimensional spaces \citep{Conrad}. In general suppose $V=\mathbb{R}^n$ and $W=\mathbb{R}^m$. Let $U\subset V$ be open and let $f_i:U\rightarrow\mathbb{R}$ denote the $i$th component of $f$, so f is described as a map $f=(f_1,...,f_m)$. Let $p\ge 0$ be a non-negative integer. Then $f$ is $C^p$ map if and only if all $p$-fold iterated partial derivatives of the $f_i$'s exist and are continuous on $U$. Suppose $\text{Hom}(V,W)$ be the space of linear mappings from $V$ to $W$. Then the higher derivative $D^pf$ is a multi-linear mapping from $V^p\rightarrow W$, i.e.,
\[
D^pf:U\rightarrow\text{Mult}(V^p,W),
\]
where $V^p=V\times...\times V$ is the $p$-th fold of Cartesian product of $V$. Choose $a\in U$ and $r>0$ such that a small neighborhood $b_r(a)\subset U$ for a choice of norm on $V$. Choose $h=\sum h_je_j\in V$ with $\norm{h}<r$. For non-negative integer $k\le p$, the higher order derivative as multi-linear mapping acting on $T_aU$ is given by the following expression,
\[
\frac{(D^kf)(a)}{k!}(h^{(k)})=\sum_{i_1+...+i_n=k}\frac{1}{i_1!...i_n!}h_1^{i_n}...h_n^{i_n}\frac{\partial^kf}{\partial x_1^{i_1}...\partial x_n^{i_n}}(a)
\]
where $h^{(k)}=(h,...,h)\in V^k$ and the sum is taken overa ll ordered $n$-tuples $(i_1,...,i_n)$ of non-negative integer whose sum is $k$. To be more concrete
\[
\frac{\partial^kf}{\partial x_1^{i_1}...\partial x_n^{i_n}}(a)=\left(\frac{\partial^kf_1}{\partial x_1^{i_1}...\partial x_n^{i_n}}(a),...,\frac{\partial^kf_m}{\partial x_1^{i_1}...\partial x_n^{i_n}}(a)\right)
\]
which is a vector in $W$. The Taylor formula for differentiable mapping $f:V\rightarrow W$ is given as follows,
\[
f(a+h)=\sum_{j=0}^p\frac{(D^jf)(a)}{j!}(h^{(j)})+R_{p,a}(h)
\]
in $W$, where
\[
R_{p,a}(h)=\int_0^1\frac{(1-t)^{p-1}}{(p-1)!}((D^pf)(a+th)-(D^pf)(a))(h^{(p)})dt
\]
satisfies
\[
\norm{R_{p,a}(h)}\le C_{p,h,a}\norm{h}^p,\ \ \lim_{h\rightarrow0}C_{p,h,a}=0
\]
with
\[
C_{p,h,a}=\sup_{t\in[0,1]}\frac{\norm{(D^pf)(a+th)-(D^pf)(a)}}{p!}.
\]
With above settings, we are ready to prove the theorem.

\begin{proof}
Denote $X=(X_i^{t_0},y_i^{t_0})$ for short, and let $\phi_t(X)$ be the flow of Hamiltonian system of $X$. The Taylor expansion with respect to mean of $X$, say $\mu$, is computed as follows,
\[
\phi_t(X)=\phi_t(\mu)+D\phi_t(\mu)(X-\mu)+\frac{1}{2}D^2\phi_t(X)(X-\mu)^{(2)}+...
\]
Since the covariance matrix of the random vector given $X$ as a random vector is encoded in the covariances of $\phi_t^i(X)$ and $\phi_t^j(X)$ for $i,j\in[n]$, i.e., $\Cov(\phi_t^i(X),\phi_t^j(X))$. The fundamental property of covariance implies that for each pair $(i,j)$, $\Cov(\phi_t^i(X),\phi_t^j(X))$ is an infinite sum of covariances given by the Taylor formular. By assumption on the covariance of the initial input $X$ such that the covariance is small, it suffices to use covariance matrix of $D\phi_t(\mu)(X-\mu)$ to approximate the covariance matrix of $\phi_t(X)$, since all the terms other than the linear ones in $\Cov(\phi_t^i(X),\phi_t^j(X))$ is of the higher power of the entries of $X-\mu$. Formally we have
\[
\Cov(\phi_t(X))\approx\Cov(D\phi_t(\mu)(X-\mu)).
\]
Since we further assume that the Hamiltonian flow $\phi_t(\cdot)$ has Lipschitz derivatives of arbitrary order, uniformly, the remainder in the approximation is bounded by a constant multiplied by higher power of entries in the covariance matrix of $X$. Recall that we use notation $X=(X_i^{t_0},y_i^{t_0})$, and apply Theorem 3 of \citep{hsiao2006fundamental} to the linear part $D\phi_t(\mu)(X-\mu)$ in the Taylor formula, we have 
\[
(\Delta X_{i,\alpha}^t\Delta y_{i,\alpha}^t)^2-(\Cov(X_{i,\alpha}^t,y_{i,\alpha}^t))^2\ge\frac{w_L^2(P_0)}{\pi^2}
\]
provided the remainder in approximation is zero. Thus for any number strictly less than $\frac{w_L^2(P_0)}{\pi^2}$, say $\frac{1}{2}\frac{w_L^2(P_0)}{\pi^2}$, as long as the covariance entries are small enough, we can have
\[
(\Delta X_{i,\alpha}^t\Delta y_{i,\alpha}^t)^2-(\Cov(X_{i,\alpha}^t,y_{i,\alpha}^t))^2\ge\frac{1}{2}\frac{w_L^2(P_0)}{\pi^2}.
\]
The proof completes.
\end{proof}

\subsection{Experiments on primal space}\label{appdenix_actural_strategy}

Note that primal space are closely related to canonical coordinates $(X_1(t),y_1(t))$ by
\begin{align*}
    &x_1(t) = \nabla h^*_1(y_1(t)) \\
    &x_2(t) = \nabla h^*_2(A^{(21)}X_1(t) + y_2(0)).
\end{align*}
It is hoped that Theorem \ref{prop:nonlinear} has its impact on primal space.

We firstly do experiments on how the operator $ \nabla h^*(\cdot)$ and the matrix-vector $A^{(21)}X_1(t)$ can affect the wave shape of the canonical coordinates $(y_1(t),X_1(t))$. Our observation can be summerized as 
\begin{observation}\label{ob1}
    For different  $i_1,i_2 \in [n_1]$, the local minima of $\Delta (x_{1,i_1}(t))$ is close to the local minima of $\Delta (y_{1,i_2}(t))$. See (a) in Figure \ref{Variance_nablah_y}.
\end{observation}

\begin{observation}\label{ob2}
    For different $i_1 \in [n_1]$ and $i_2 \in [n_2]$, if $\ \exists \ \epsilon > 0$ such that $x_{1,i_1}(t) > \epsilon$, then the local minima of $\Delta\left(X_{1,i_1}(t)\right)$, is close to local minima of $\Delta\left((A^{(21)}X_{1}(t) + y_2(0))_{i_2}\right)$. See (b) in Figure \ref{Variance_nablah_y}.
\end{observation}

Observation \ref{ob1} can be explained as viewing the function graph of $\Delta (x_{1,i_1}(t))$ and $\Delta (y_{1,i_2}(t))$ as waves, then the peaks(troughs) of wave of $\Delta (x_{1,i_1}(t))$ are close to the peaks(troughs) of wave of $\Delta (y_{1,i_2}(t))$. Observation \ref{ob2} has the same meaning. Note that the condition $\ \exists \epsilon > 0$ such that $x_{1,i_1}(t), x_{1,i_2}(t)> \epsilon$ is necessary because otherwise $\lim_{t \to \infty} x_{1,i_1}(t) = \lim_{t \to \infty} x_{1,i_2}(t) = 0$, thus the player will not use these strategies as time process, and the variance of these two strategies will equal to $0$, thus the uncertainty inequality will not hold for these strategies.

 Although the above two observations are less rigid, however, combine with  Theorem \ref{prop:nonlinear}, they provides insights into the covariance evolution of primal space. Note that the second term implies that the positions of the local minima/maxima of $\Delta y_{2}(t) = \Delta(A^{(21)}X_1(t) + y_2(0))$ are close to those of $\Delta(X_1(t))$. Due to the symmetry between the two players, combining this with the first term allow us to conclude that the position of the local minima/maxima of $\Delta x_{2}(t)$ are close to those of $\Delta X_{1}(t)$. Furthermore, based on the first term, we can also infer that the position of the local minima/maxima of $\Delta x_{1}(t)$ are close those of $\Delta y_{1}(t)$. Therefore, by combining these two observations with Proposition \ref{prop:nonlinear}, we can conclude that the positions of local minima in $\Delta x_{1}(t)$ are close to those of local maxima in $\Delta x_{2}(t)$, and vice versa.

\begin{figure}[h]
\centering
\subfigure[Functions graph of $\Delta (x_{1,i_1}(t))$ and $\Delta (y_{1,i_2}(t))$.]{
\includegraphics[clip,width=0.45\columnwidth]{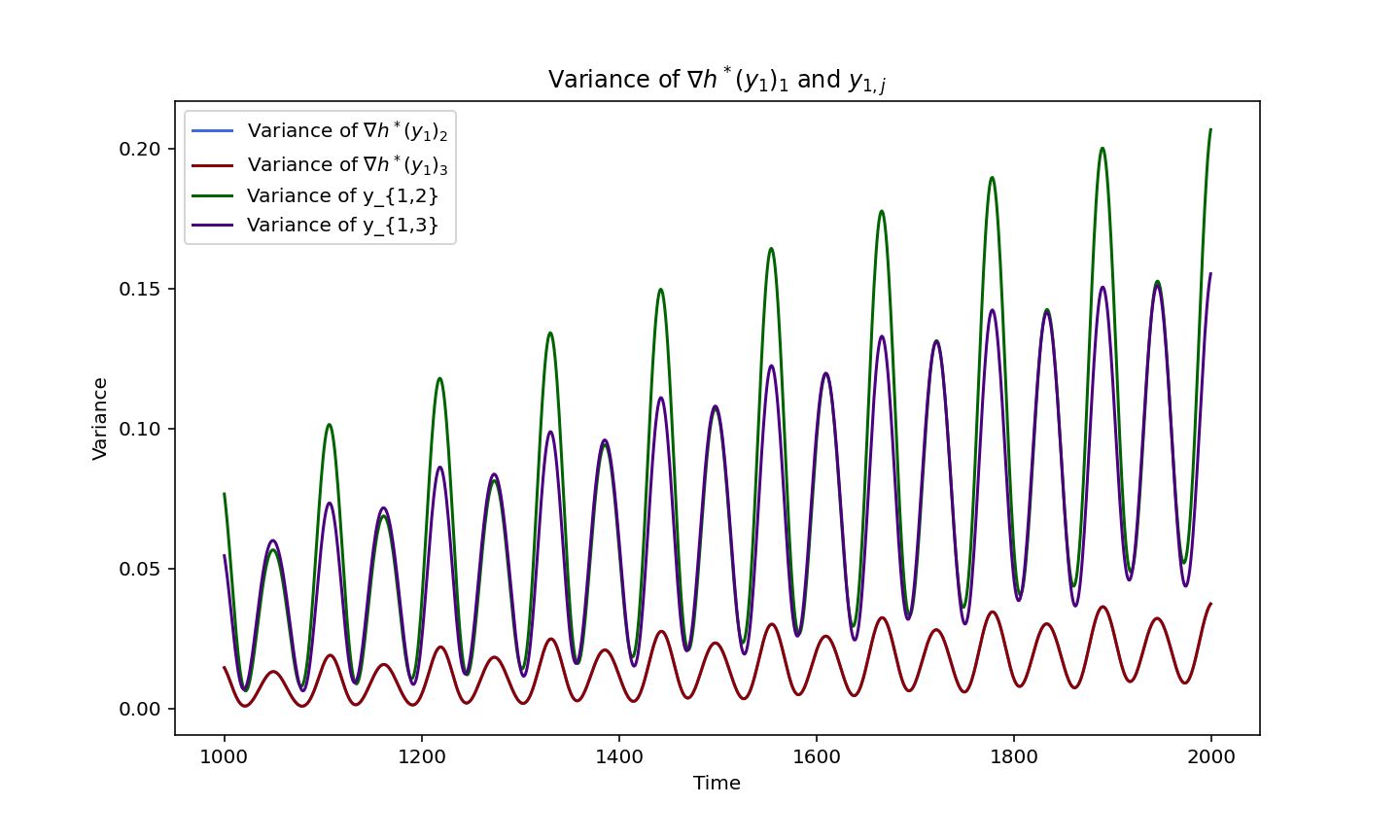}
}
\subfigure[Functions graph of $\Delta(A^{(21)}X_1(t) + y_2(0))$ and  $\Delta(X_1(t))$.]{
\includegraphics[clip,width=0.45\columnwidth]{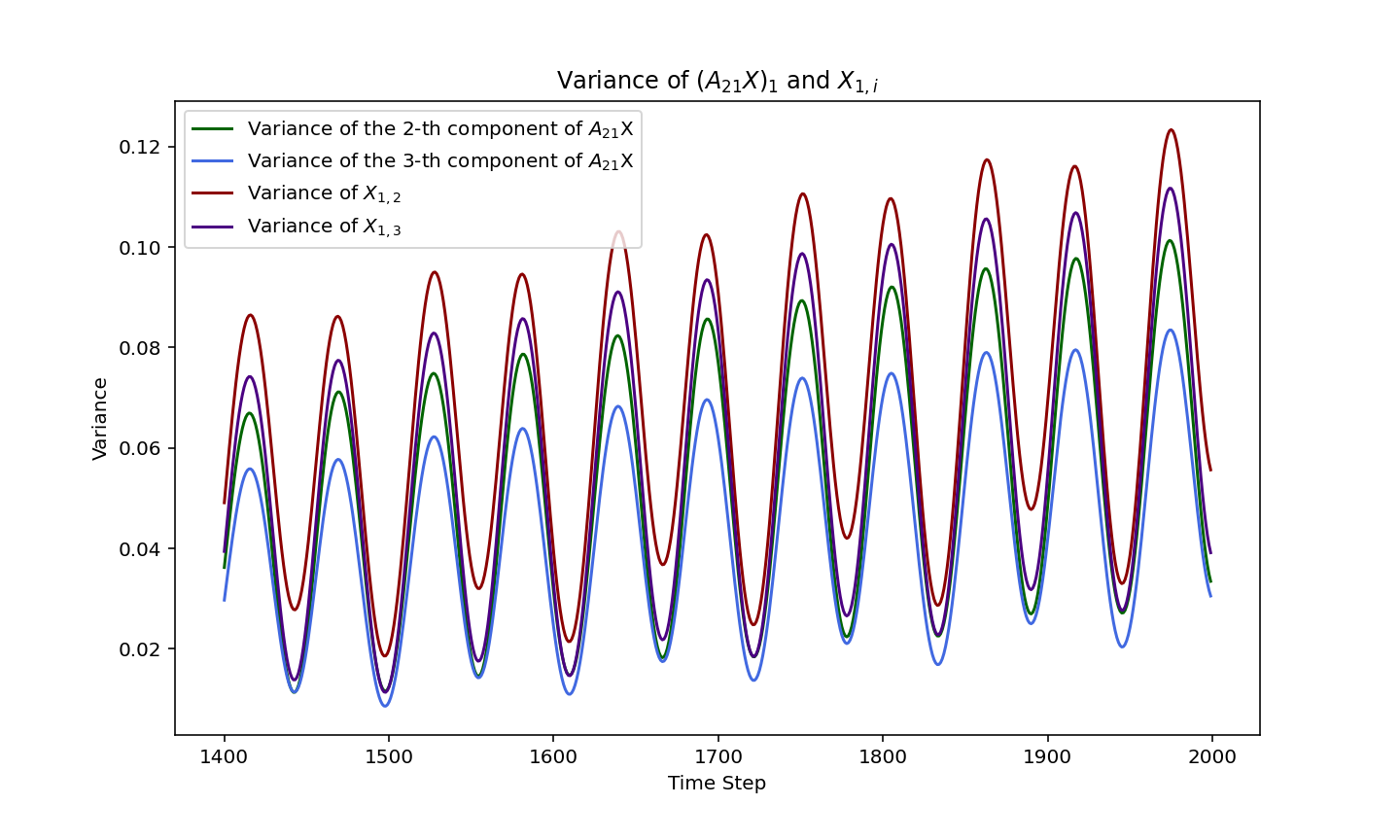}
}
\caption{Variance evolution of classic Euler discretization}
\label{Variance_nablah_y}
\end{figure}

In Figure \ref{Variance_nablah_y},  experiments are presented in the case of MWU algorithms and a randomly generated $3*3$ game. The regularizer is chosen as $h(x) = \sum^{n_i}_{i=1}x_i \ln x_i $ and the constraint is chosen as simplex constrain $\mathcal{X} = \{ x_{i}\ | \ \sum^{n}_{i=1} x_{i} = 1, x_{i} > 0\}$. Note that in this case, we have
\begin{align*}
    \nabla h^*(y) = \left( \frac{e^{y_i}}{\sum^n_{s=1} e^{y_s}} \right)^n_{i=1}.
\end{align*}
In Figure \ref{Variance_nablah_y} (a), we can see although variance of $\nabla h^*(y_1)_1$ has smaller value than variance of $y_{1,i}$, but the positions of local minima/maxima of these curves are very close. Similarly, in Figure \ref{Variance_nablah_y} (b), but the positions of local minima/maxima of these curves are also very close. 

In the following, we provide additional experiments on the evolution of variance in the primal space. We observe that if the dimension of the game is low, for example, less than 10, we can see a similar covariance evolution in the primal space as shown in Figure \ref{uncertaintyprinciple}. However, when the dimension becomes very high, this pattern can disappear.

\begin{figure}[h]
\centering
\subfigure[]{
\includegraphics[clip,width=0.45\columnwidth]{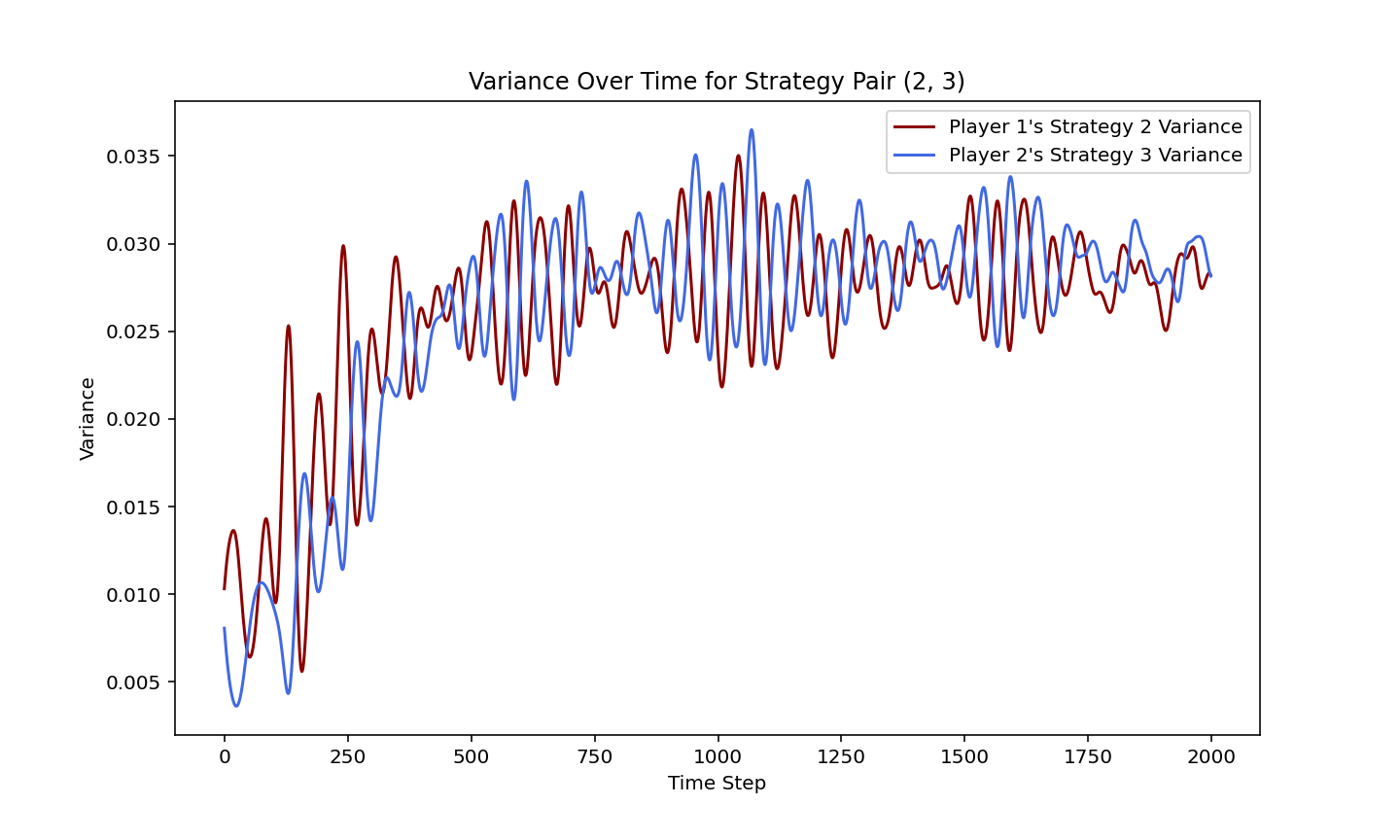}
}
\subfigure[]{
\includegraphics[clip,width=0.45\columnwidth]{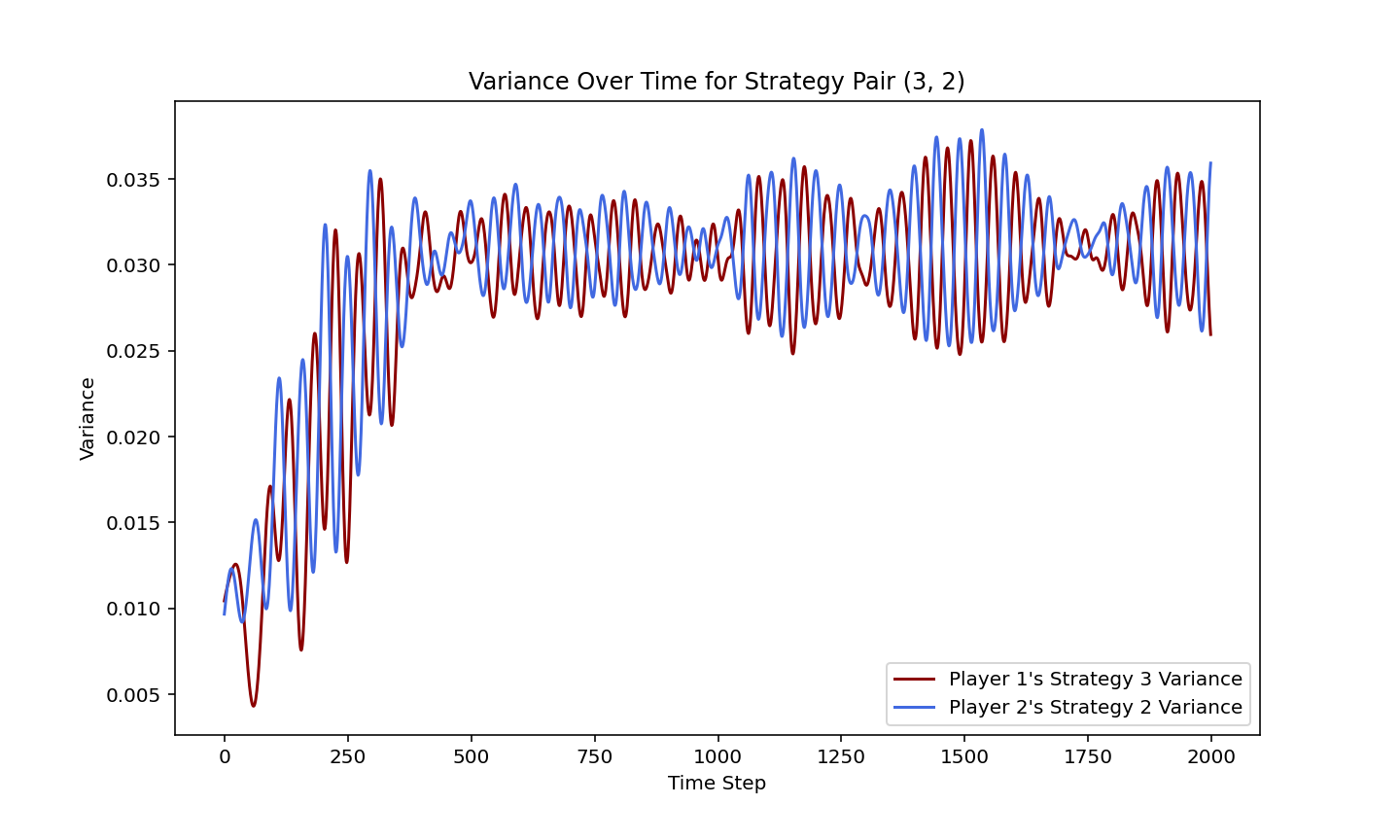}
}
\caption{Covariance evolution on a randomly generated 3*3 game}
\label{}
\end{figure}

\begin{figure}[h]
\centering
\subfigure[]{
\includegraphics[clip,width=0.45\columnwidth]{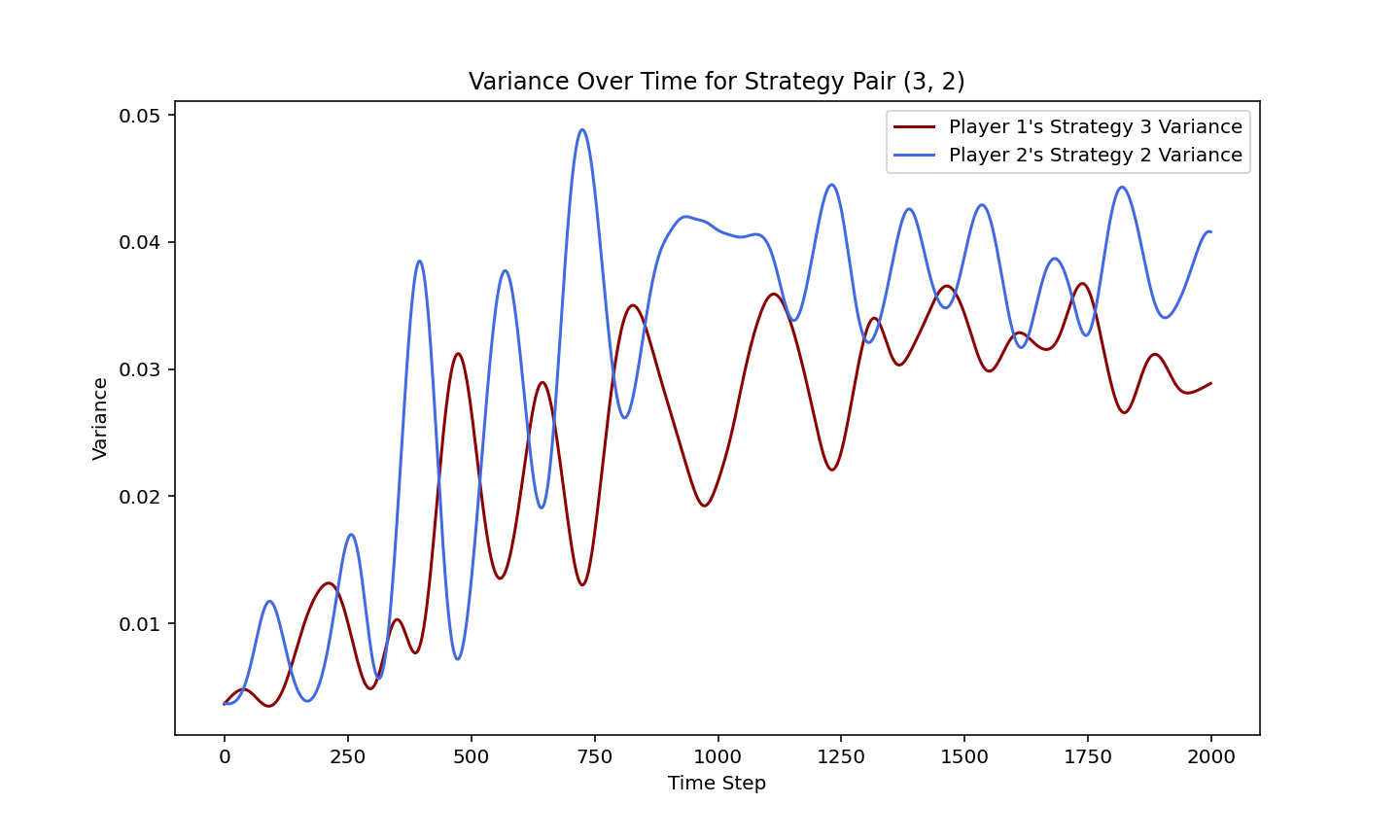}
}
\subfigure[]{
\includegraphics[clip,width=0.45\columnwidth]{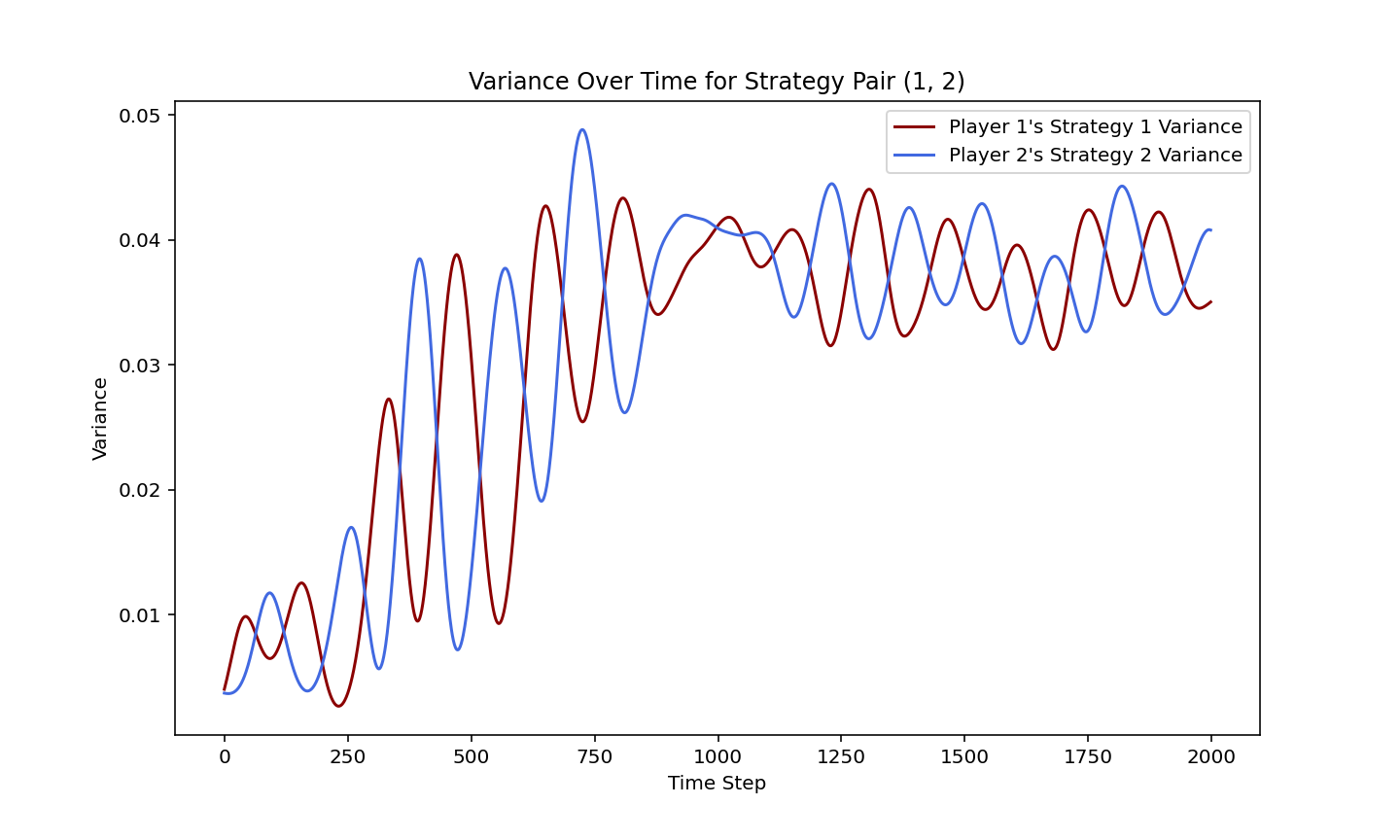}
}
\caption{Covariance evolution on a randomly generated 5*5 game}
\label{}
\end{figure}

\begin{figure}[h]
\centering
\subfigure[]{
\includegraphics[clip,width=0.45\columnwidth]{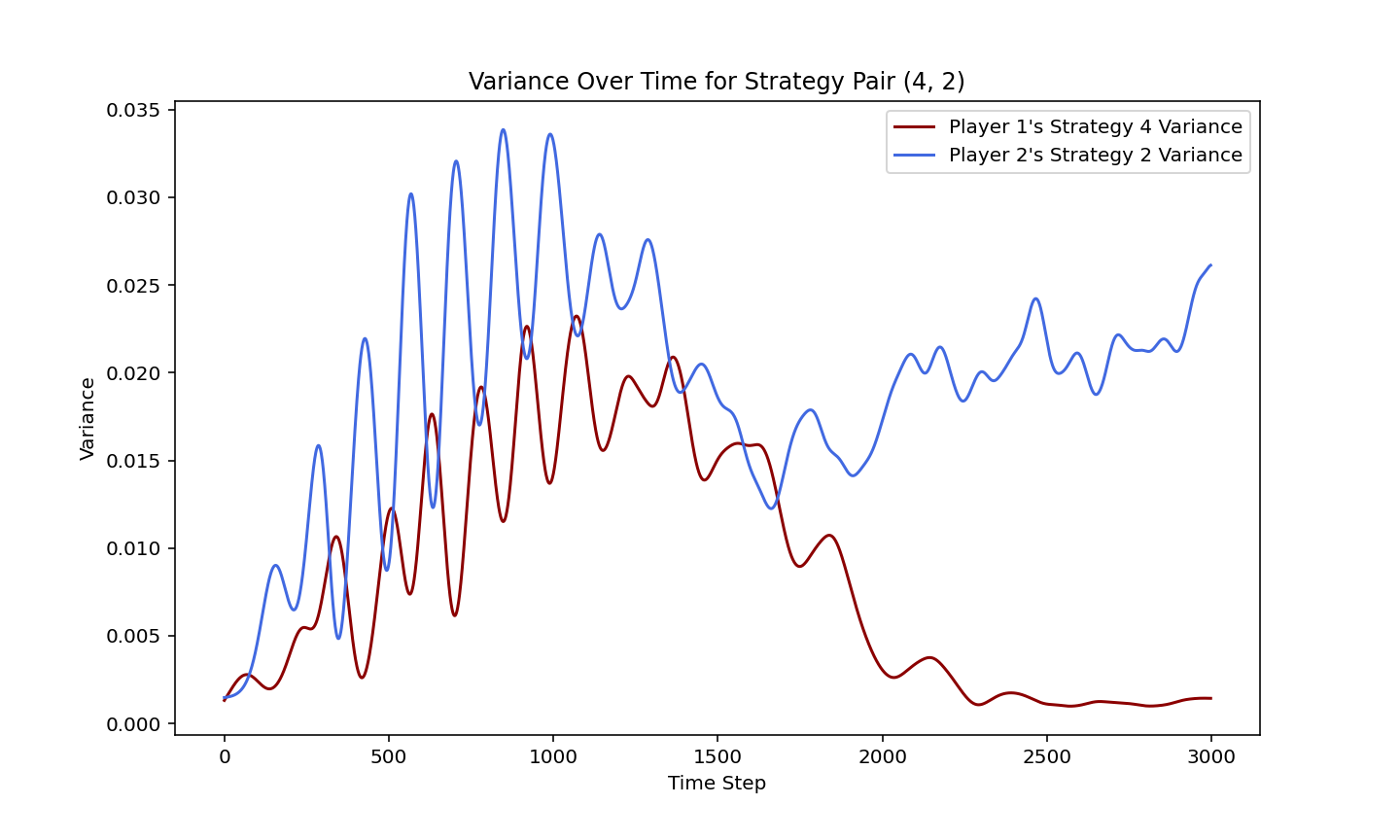}
}
\subfigure[]{
\includegraphics[clip,width=0.45\columnwidth]{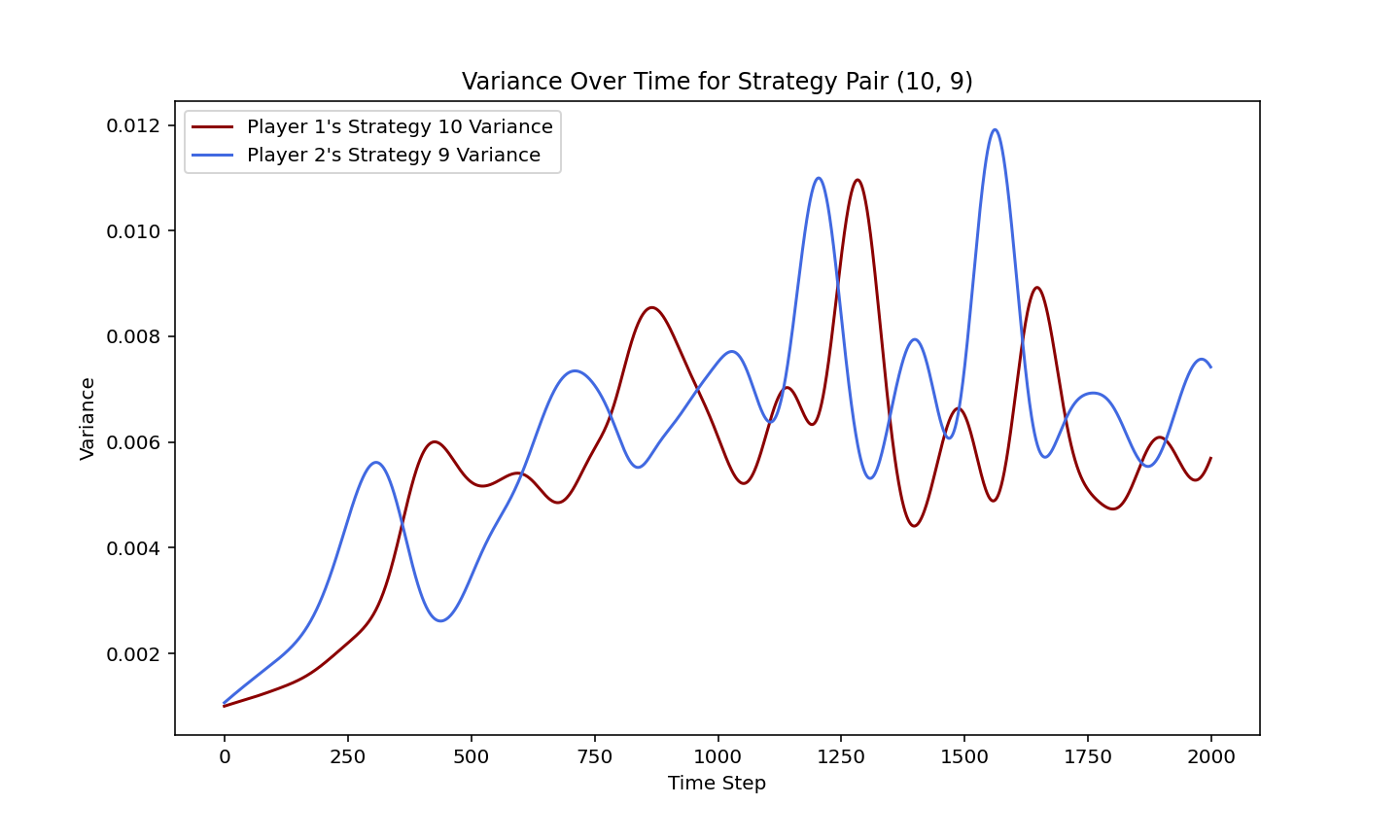}
}
\caption{Covariance evolution on a randomly generated 10*10 game}
\label{}
\end{figure}

\begin{figure}[h]
\centering
\subfigure[]{
\includegraphics[clip,width=0.45\columnwidth]{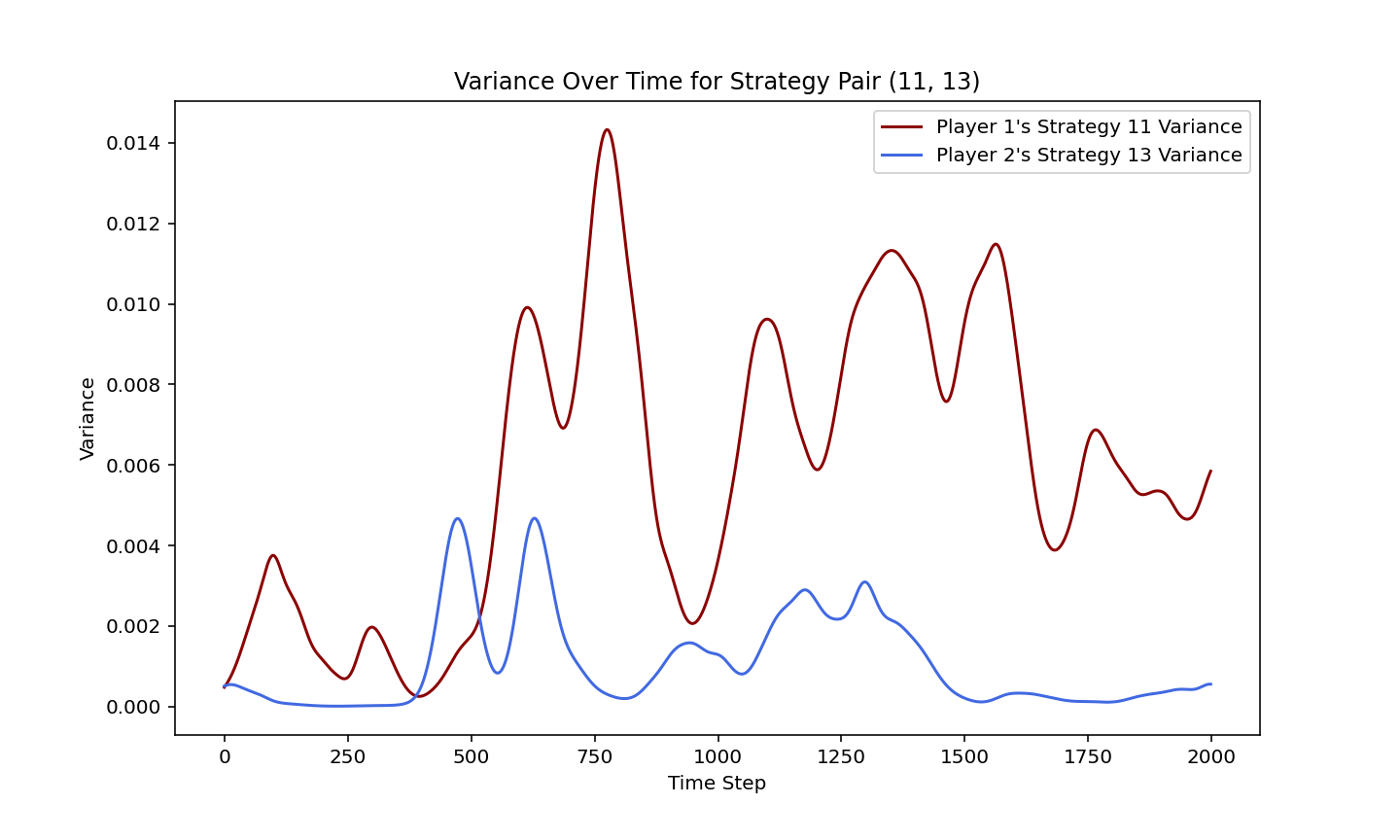}
}
\subfigure[]{
\includegraphics[clip,width=0.45\columnwidth]{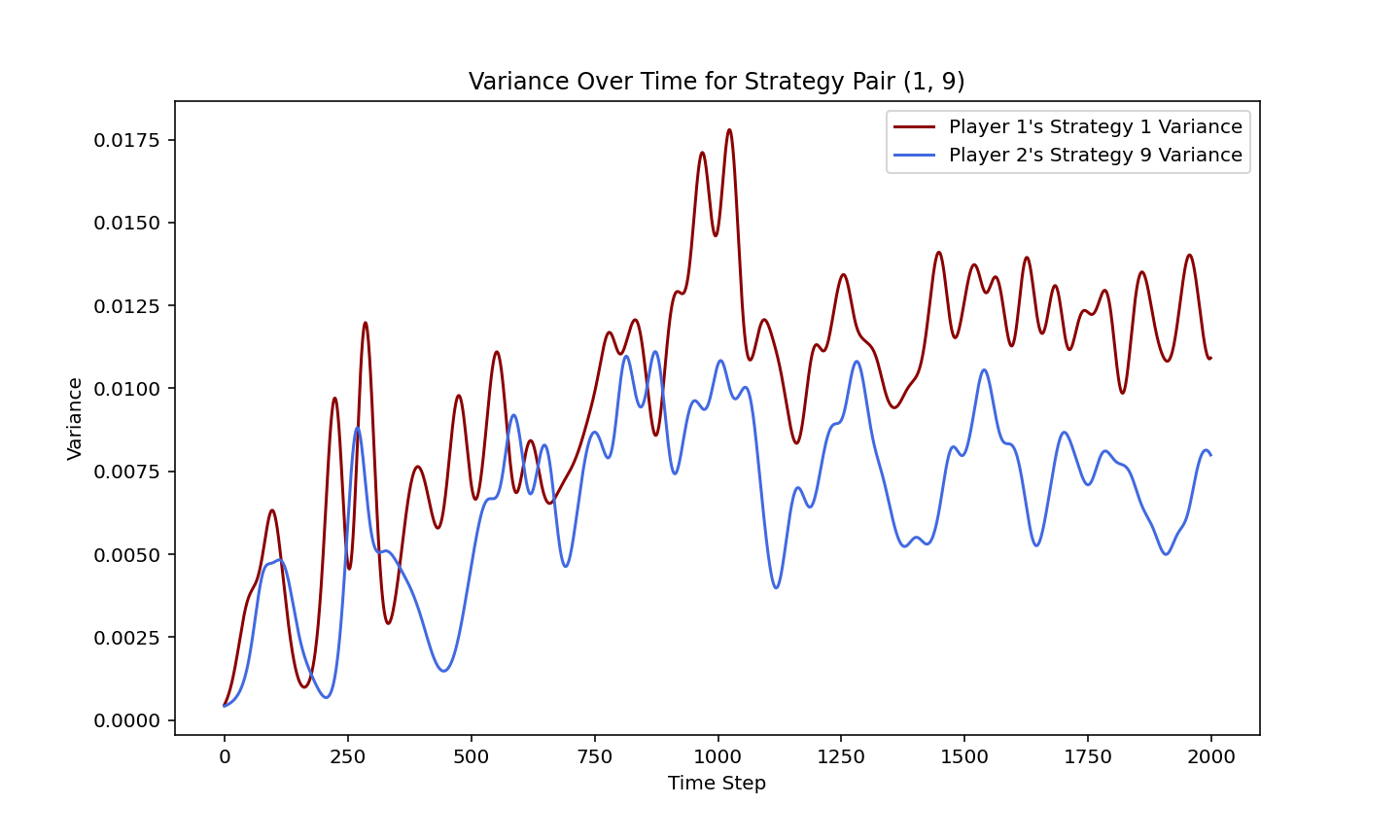}
}
\caption{Covariance evolution on a randomly generated 15*15 game}
\label{}
\end{figure}

\newpage
\section{Experiments for non-singular cases.}\label{eeee1}


\textbf{Continuous time FTRL.} We illustrate how $\Var(X_{1,1}(t))$ and $\Var(y_{1,1}(t))$ evolve with continuous time FTRL with payoff matrices
\begin{align*}
    &A_4 = [[1,-2],[-1,1]], A_5 =[[2,-3],[-1,5]] ,\\
    &A_6 =[[2,-1.5],[-2,3]].
\end{align*}

See (a)(b) Figure \ref{continuous va2}.
In (a) the $\Var(X_{1,1}(t))$ is bounded, and in (b) $\Var(y_{1,1}(t))$ is bounded, which support results of continuous time part in Theorem \ref{Covariance Evolution in FTRL with Euclidean Regularizer}
for the non-singular cases.

\begin{figure}[h]
    \centering
    \subfigure[$\Var(X_{1,1}(t))$, non-singular]{
    	\includegraphics[width=3in]{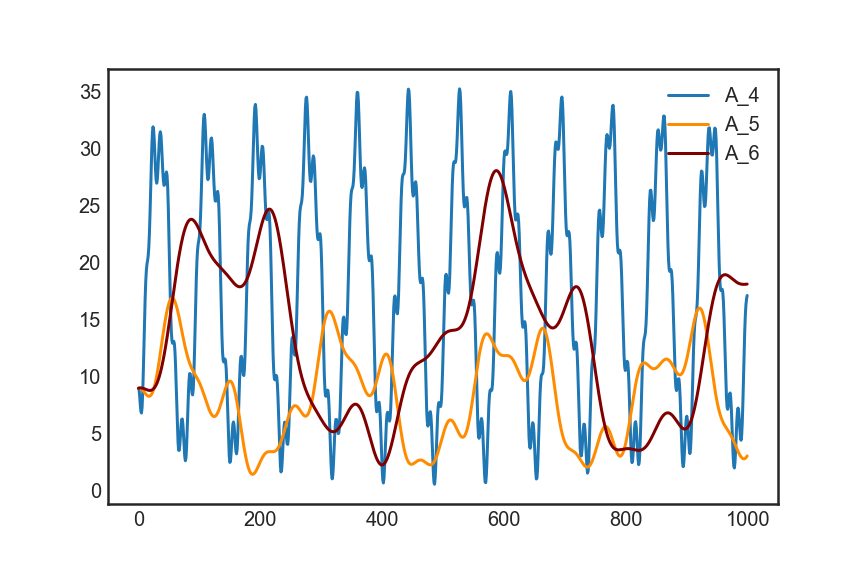}
        \label{}
    }
    \subfigure[$\Var(y_{1,1}(t))$, non-singular]{
	\includegraphics[width=3in]{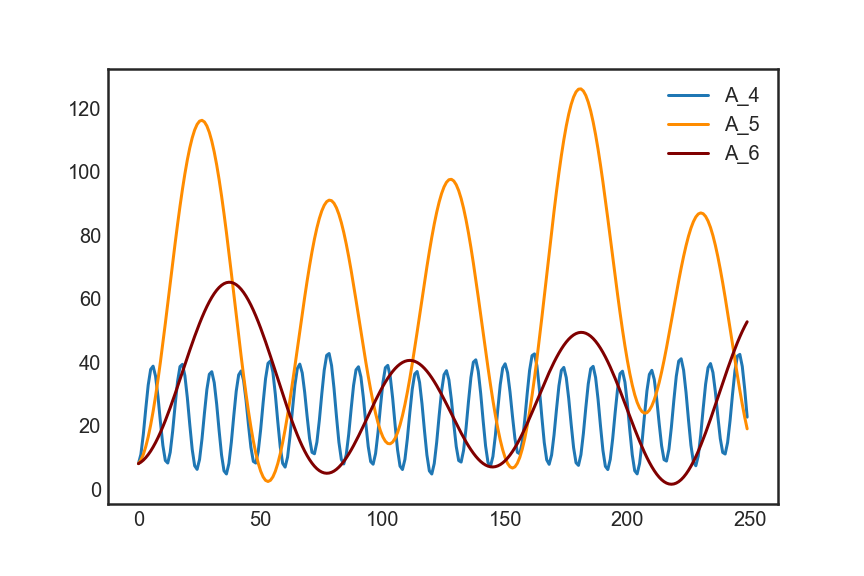}
        \label{}
    }
    \caption{Variance evolution of continuous FTRL}
    \label{continuous va2}
\end{figure}

\textbf{Symplectic discretization.} We illustrate how $\Var(X^t_{1,1})$ and $\Var(y^t_{1,1})$ evolves with symplectic discretization, the payoff matrices are given as follows: 
\begin{align*}
    &B_4 = [[1,-1.1],[-1,1]], B_5 = [[1,-1.2],[-1,1]], \\
    &B_6 = [[1,-1.3],[-1,1]].
\end{align*}
See Figure \ref{sym va2}. From the experimental results, we can see the variance behavior of symplectic discretization is same as continuous case, which support results of symplectic discretization part of Theorem \ref{Covariance Evolution in FTRL with Euclidean Regularizer} for the non-singular cases.

\begin{figure}[h]
    \centering
    \subfigure[$\Var(X^t_{1,1})$, non-singular]{
    	\includegraphics[width=3in]{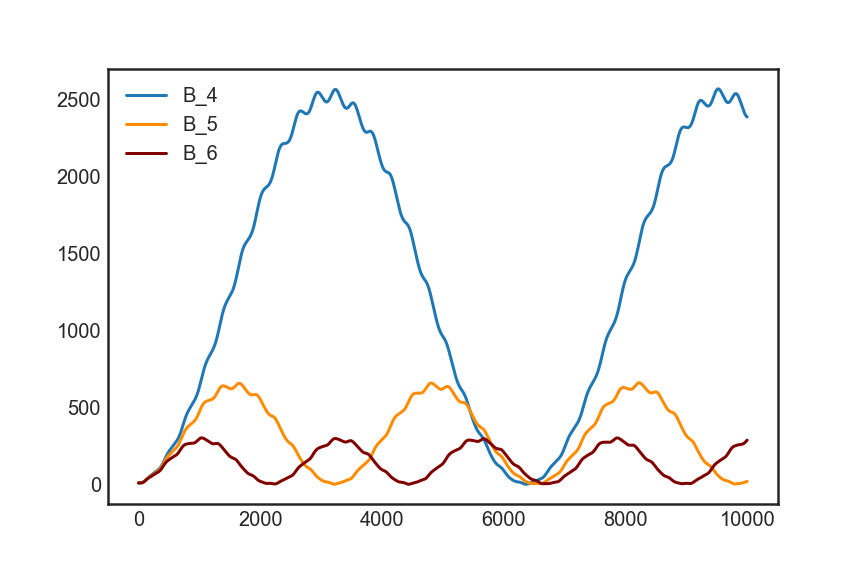}
        \label{}
    }
    \subfigure[$\Var(y^t_{1,1})$, non-singular]{
	\includegraphics[width=3in]{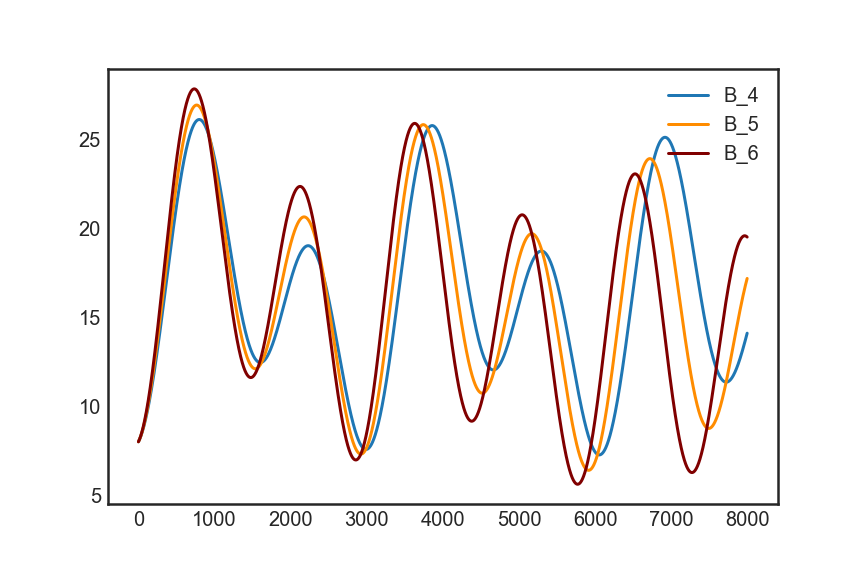}
        \label{}
    }
    \caption{Variance evolution of Symplectic discretization}
    \label{sym va2}
\end{figure}

\end{document}